\documentclass{article}
\usepackage[whole]{bxcjkjatype}
\usepackage{graphicx}
\usepackage{braket}
\usepackage{import}
\usepackage{amsmath}
\usepackage[margin=1in]{geometry}
\usepackage{amsfonts}
\usepackage[utf8]{inputenc}
\usepackage[T1]{fontenc}
\usepackage{verbatim}
\usepackage{algorithm}
\usepackage{algpseudocode}
\usepackage{booktabs}

\usepackage{multirow}
\algrenewcommand\algorithmicrequire{\textbf{Input:}}
\algrenewcommand\algorithmicensure{\textbf{Output:}}

\usepackage[sort&compress,numbers]{natbib}
\usepackage{url}
\usepackage{hyperref}
\hypersetup{colorlinks,linkcolor=blue,urlcolor=blue,citecolor=blue,hypertexnames=false}
\usepackage{graphicx}
\usepackage[T1]{fontenc}
\usepackage{amsthm}
\usepackage{amssymb}
\newtheorem{definition}{Definition}[section]
\newtheorem{proposition}{Proposition}[section]

\newtheorem{claim}{Claim}[subsection]
\newtheorem{theorem}{Theorem}[]

\newtheorem{lemma}{Lemma}[section]
\newtheorem{corollary}{Corollary}[section]

\newtheorem{remark}{Remark}[section]
\usepackage{vwcol}
\usepackage{tikz}
\usepackage{tikz-cd}
\usepackage{mathtools}
\usepackage{tasks}
\usepackage{appendix}
\usepackage{algorithm}
\usepackage{algpseudocode}
\usepackage{makecell}
\usepackage{wrapfig}
\usepackage{xspace}
\usepackage{makecell}
\usepackage{array}
\usepackage{multirow}
\usepackage{arydshln}
\usepackage{quantikz}
\usetikzlibrary{calc}
\usetikzlibrary{decorations.pathreplacing,calligraphy}
\usepackage{makecell}

\DeclareMathOperator{\im}{im}

\DeclareMathOperator{\CNOT}{\mathsf{CNOT}}

\usepackage{macros}

\title{Linear-Time Encodable and Decodable Quantum Error-Correcting Codes}
\author{
Adam Wills\thanks{Center for Theoretical Physics --- a Leinweber Institute, Massachusetts Institute of Technology, Cambridge, MA;\\Hon Hai (Foxconn) Research Institute, Taipei, Taiwan. Email: \texttt{a\_wills@mit.edu}}
\and Ting-Chun Lin\thanks{Department of Physics, University of California at San Diego, La Jolla, CA; Hon Hai (Foxconn) Research Institute, Taipei, Taiwan. Email: \texttt{til022@ucsd.edu}}
\and Rachel Yun Zhang\thanks{CSAIL, Massachusetts Institute of Technology, Cambridge, MA. Email: \texttt{rachelyz@mit.edu}}
\and Min-Hsiu Hsieh\thanks{Hon Hai (Foxconn) Research Institute, Taipei, Taiwan. Email: \texttt{min-hsiu.hsieh@foxconn.com}}
}

\begin{document}

\maketitle

\begin{abstract}
Recent years have seen rapid development in the subject of quantum coding theory, with breakthroughs on many exciting classes of codes, including quantum LDPC codes, quantum locally testable codes, and quantum codes with interesting transversal gates. However, a natural class of quantum codes, which has been well-studied classically, has not yet been treated: those which can be quickly encoded and decoded. This problem concerns the channel capacity setting, where a noise channel sits between perfect encoding and unencoding/decoding operations; this is the setting that is relevant for communication between fault-tolerant quantum computers.

In this work, we construct asymptotically good quantum codes that can be encoded and unencoded by quantum circuits of logarithmic depth and consisting of a linear total number of gates. The classical decoding algorithms also run in logarithmic depth and use $\mathcal{O}(n \log n)$ gates, or alternatively a linear number of gates but with higher depth. We further construct explicit and asymptotically good quantum codes whose encoding, unencoding and decoding all use a linear number of gates, and additionally whose encoding and unencoding may be run in logarithmic depth.

\end{abstract}

\section{Introduction}

Motivated by the need to overcome the pervasive effects of noise in quantum systems, recent years have seen a rapid expansion in the study of quantum coding theory. The field has seen the development of many classes of quantum codes, some of which have natural classical analogues, and some of which are inherently quantum. Most notably on the former, in 2021, the existence of asymptotically good quantum LDPC codes was resolved in the affirmative by Panteleev and Kalachev~\cite{panteleev2022asymptotically}, ending a 20-year search. 
This result mirrors the construction of the first explicit and asymptotically good classical LDPC codes—the expander codes of Sipser and Spielman~\cite{sipser1996expander}—and sparked a flurry of follow-up constructions \cite{9996782, 10.1145/3564246.3585101, lin2022goodquantumldpccodes}.

Interestingly, a related work due to Spielman~\cite{spielman1995linear}, which came out at approximately the same time, presented a breakthrough result using related techniques on a different class of classical codes. These were the first known asymptotically good codes to have linear-time encoding and decoding algorithms, where the \textit{time}, also called the \textit{complexity}, of an algorithm refers to the total number of gates used in the algorithm. Spielman also showed that these codes admit a parallel encoding algorithm of logarithmic depth and linear time, and a parallel decoding algorithm of logarithmic depth and $\mathcal{O}(n\log n)$ time.

The results of Spielman's paper apply in the \textit{code capacity} setting. That is, a classical encoding algorithm, which runs without faults, takes in a bit string representing the message to be sent, and outputs the codeword corresponding to that message. That codeword is then passed through a noise channel, and at the other end, a faultless decoding algorithm attempts to recover the message. The relevance of this setting on the classical side is quite clear; indeed, classical processors can be assumed for many purposes to be essentially noiseless, whereas noise occurring during the transmission of a signal, or during long-term storage can be more severe. On the other hand, it is understandable on the face of it why this class of codes has not received attention yet quantumly. Indeed, quantum information is extremely susceptible to noise from environmental sources, and imperfections in the quantum devices themselves cause further problems. Accordingly, quantum information can never be stored, or computed ``raw'', that is, not encoded in a code, and be assumed to remain accurate to any useful degree. As such, it is unsurprising that significant attention has been devoted to quantum codes and fault-tolerance schemes that assume imperfect operations, and do not allow important quantum information to be held unencoded~\cite{aharonov1997fault, gottesman2013fault, yamasaki2024time, nguyen2025quantum}; one might call this the setting of quantum fault tolerance.

On the other hand, the code capacity setting has also been widely studied in quantum error correction, and more broadly in quantum information theory~\cite{schumacher1996quantum, bennett1996mixed, knill1997theory, lloyd1997capacity, divincenzo1998quantum, shor2002quantum, shor2004equivalence, devetak2005private, devetak2005capacity, hastings2009superadditivity, HW2010TIT, Datta2013Smooth, BDH14TIT, ZZHS19TIT}. Here, the setting is the natural analogue of the classical code capacity setting. That is, a quantum message (for example, some number of qubits in a certain state) is encoded, via a faultless quantum operation, into a quantum codestate. This state is then sent through a noise channel. Then, a faultless quantum processor attempts to return the noisy codeword to the original message, usually assisted by a (faultless) classical algorithm. Of this latter process, the quantum part is called the unencoding operation, and the classical part is called the decoding operation. Many problems in this setting have been well-studied, especially on the capacities of various quantum channels, and conditions for the correctability of certain noise models.

The complexities of the operations essential to the channel capacity setting, the quantum encoding and unencoding, as well as the classical decoding, have not been studied. That is, quantum codes with low-depth or low-time encoders, unencoders and decoders, have not been discovered. However, we can consider a key reason why this is a worthwhile problem to study: \textit{quantum communication}. Given two fault-tolerant quantum processors, there are numerous reasons that they may wish to send information between them. In such a case, it would be generically expected that the environment between the quantum computers is far noisier than the environment of the processors themselves, in which quantum information is already protected by some host error-correcting code. That is, using the fault-tolerant operations of the devices, information could be perfectly encoded and unencoded from an error-correcting code, which is then run at the logical level of the quantum computers' host error-correcting code (stated another way, the two codes are concatenated). In such a case, having low-depth encoding, unencoding, and decoding operations is highly desirable to avoid a bottleneck in, for example, a distributed quantum computation or quantum information theoretic task that is being performed.

As well as the application in quantum communication, we believe studying the depth and complexity of the encoding, unencoding and decoding operations over quantum channels is an inherently interesting problem. Indeed, quantum error-correcting codes find utility across quantum information theory and quantum complexity theory, and understanding the complexity and depth of these operations is an important thing to understand.

Our main results follow. The depths and complexities of each algorithm are summarised in Table~\ref{tab:complexities_summary}.

\begin{theorem}\label{thm:main_randomised}
    There exists an asymptotically good quantum error-correcting code over qubits which may be encoded and unencoded using quantum circuits with a linear number of gates. They may also be decoded using classical circuits with a linear number of gates. The codes also have parallel encoding and unencoding quantum circuits of logarithmic depth and a linear number of gates. They have a parallel classical decoding algorithm that runs in logarithmic depth and uses a total number of gates $\mathcal{O}(n \log n)$. The codes may be constructed with any rate in $(0,1)$.
\end{theorem}

\begin{theorem}\label{thm:main_explicit}
    There exists an explicit construction of an asymptotically good quantum error-correcting code over qubits which may be encoded and unencoded using quantum circuits with a linear number of gates. Moreover, they may be decoded using a classical circuit with a linear number of gates. The codes also have parallel encoding and unencoding quantum circuits of logarithmic depth and a linear number of gates. The codes may be constructed with any rate in $(0,1)$.
\end{theorem}
In the statements of Theorems~\ref{thm:main_randomised} and~\ref{thm:main_explicit}, the phrase ``asymptotically good'' is a common term in coding theory which tells us that the code has a constant rate and constant relative distance. Having a constant rate means that the ratio between the number of the code's logical qubits (the number of qubits in the message being communicated), and the number of the code's physical qubits (the number of qubits that must be sent over the channel) is some constant (is bounded away from zero as the code's number of physical qubits diverges). Similarly, having a constant relative distance means that any set of the physical qubits, whose size is some constant fraction of the total, may be arbitrarily damaged, and the original message may be recovered, and indeed in this case the message may be recovered using the described unencoders and decoders of low depth and total complexity.

\renewcommand{\arraystretch}{1.25}

\begin{table}[ht]
\centering

\begin{tabular}{llcc}
\toprule
Construction & Task & Sequential Algorithms & Parallel Algorithms \\
\midrule

\multirow{3}{*}{Randomised}
& Encoding   & Linear & Log depth, linear total \\
& Unencoding & Linear & Log depth, linear total \\
& Decoding   & Linear & Log depth, $\mathcal{O}(n\log n)$ total \\

\addlinespace
\midrule
\addlinespace

\multirow{3}{*}{Explicit}
& Encoding   & Linear & Log depth, linear total \\
& Unencoding & Linear & Log depth, linear total\\
& Decoding   & Linear & Future work\\

\bottomrule
\end{tabular}

\caption{Complexities of algorithms in each construction.
Depth/total gate counts for encoding and unencoding refer to
\textit{quantum} gates; decoding refers to \textit{classical} gates.}\label{tab:complexities_summary}
\end{table}

\subsection{Overview of the Methods}

We now overview the methods we use to construct our quantum error-correcting codes. This is explained at a high level assuming knowledge of the essential elements of quantum error correction, but is made accessible to a broad theoretical computer science audience via the preliminary material in Section~\ref{sec:prelims}.

We start in Section~\ref{sec:overview_qecc_from_qerc} by explaining how these codes may themselves by constructed from certain concatenations of quantum error-reduction codes. Then, in Section~\ref{sec:overview_qerc_from_lossless}, we explain how these quantum error-reduction codes may be constructed from a particular expanding graph that we term a ``lossless $Z$-graph''. Finally, in Section~\ref{sec:overview_lossless_construction}, we describe how our lossless $Z$-graphs are constructed, both randomly and explicitly.

\subsubsection{Quantum Error-Correcting Codes from Quantum Error-Reduction Codes}\label{sec:overview_qecc_from_qerc}

To overview our construction, we begin by describing the first construction of linear-time encodable and decodable \textit{classical} error-correcting codes due to Spielman~\cite{spielman1995linear}. In this paper, we will find that certain elements of that construction quantise seamlessly, whereas others present much greater roadblocks in their quantisation.

Spielman constructs classical error-correcting codes with fast encoders and decoders by taking \textit{classical error-reduction codes} with the same properties and concatenating them in a particular structure. An error-reduction code works as follows. We begin with the message, a sequence of logical bits $x \in \mathbb{F}_2^n$ that we wish to communicate to a receiver. This will be encoded into the error-reduction code, and so we also start with $m$ bits initialised in some fixed state such as $0$. An encoding circuit is then run on the collective bits in order to produce the codestate corresponding to the message $x$. Once we have the codestate corresponding to $x$, the $n$ bits that initially corresponded to the message are called \textit{message bits}, whereas the $m$ bits initialised before the encoding circuit in a fixed state are called \textit{check bits}. In total, they form the codeword of length $n+m$.

Typically, for an error-reduction code, we would think of this encoding circuit as taking place with extremely low complexity and depth; indeed, in~\cite{spielman1995linear}, the error-reduction codes have constant-depth encoding circuits (of linear complexity). Of course, having a constant-depth encoding circuit means that the code must have a constant distance, and thus cannot be a good error-correcting code, but what can it do? Indeed, suppose that in the noise channel, $v$ errors occur on the message bits, and $t$ errors occur on the check bits, where $v$ and $t$ are some small constant fraction of the block length. The code is said to be an error-reduction code with error reduction $\epsilon$ if it is possible (via some error-reduction algorithm) to recover the original message up to some residual error of size at most $\epsilon \cdot t$. In some sense, therefore, the error-reduction code is still resilient to check bit errors even if it cannot correct them entirely. As well as the encoding algorithm having low depth and complexity, one should think of the error-reduction algorithm as also having low depth and complexity. The low depths and complexities of the encoding and error-reduction algorithms for the error-reduction codes will translate into low depths and complexities for the encoding and decoding algorithms of the error-correcting code, respectively.

Let us now describe the concatenation structure~\cite{spielman1995linear} that allows us to construct classical error-correcting codes from classical error-reduction codes. As always in this section, we omit details, and present the high-level idea. We construct a family of error-correcting codes $(\mathcal{Q}_k)_{k=0}^\infty$ inductively, where the base case $\mathcal{Q}_0$ may simply be chosen to be a random code of constant size. Then, $\mathcal{Q}_k$ is constructed from an error-reduction code $\mathcal{R}$ and the error-correcting code $\mathcal{Q}_{k-1}$, as follows.\footnote{The error-reduction code $\mathcal{R}$ will technically be some family of error-reduction codes, but we omit these details for this high-level description.} Describing the encoding circuit suffices to describe the code itself. The encoding is depicted in Figure~\ref{fig:cascade_structure}.

\begin{figure}[h!]
    \centering
    \includegraphics[width=0.35\linewidth]{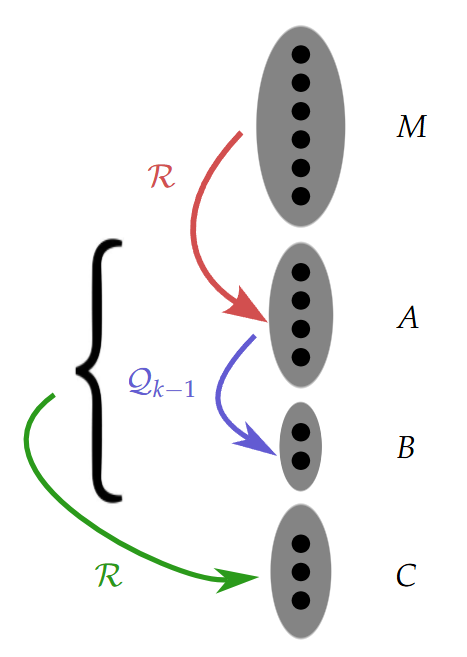}
    \caption{High-level construction of the classical error-correcting code $\mathcal{Q}_k$ from the classical error-correcting code $\mathcal{Q}_{k-1}$ and the classical error-reduction code $\mathcal{R}$. We will use the same structure to form quantum error-correcting codes from quantum error-reduction codes.}
    \label{fig:cascade_structure}
\end{figure}

\begin{enumerate}
    \item The code $\mathcal{Q}_k$ has a set of message bits $M$ which we wish to encode. First, $M$ is encoded into the code $\mathcal{R}$ using a set of check bits $A$.
    \item The bits in $A$ are encoded as message bits into the code $\mathcal{Q}_{k-1}$ using another set of bits $B$ as check bits.
    \item The bits in $A \cup B$ are encoded as message bits into the code $\mathcal{R}$ using a final set of check bits $C$.
\end{enumerate}
First, we note that if $\mathcal{R}$ has a very high rate (greater than $1/2$), then the family of codes $\mathcal{Q}_k$ constructed in this way has a constant rate. Moreover, the length of the code family $\mathcal{Q}_k$ grows exponentially with $k$. In turn, this means that a number of layers of concatenation logarithmic in the block length is used to encode $\mathcal{Q}_k$. This alone has consequences for the encoding circuit. Indeed, it is quite quick to show that if $\mathcal{R}$ has a constant-depth encoding circuit with linear complexity, then $\mathcal{Q}_k$ has a logarithmic-depth encoding circuit with linear complexity.

Next, let us describe the decoding of $\mathcal{Q}_k$. There exist both sequential and parallel algorithms; here we describe only the sequential algorithm for brevity. First, assume that the noise channel has left errors on small constant fractions of the bits in $M, A, B$ and $C$. To decode the code, first, the error-reduction algorithm is run for the code $\mathcal{R}$ on the bits $A \cup B \cup C$. Doing so leaves us with bits $M \cup A \cup B$, however, the number of errors in $A \cup B$ has been reduced (relative to the initial number in $C$). Indeed, the idea is that the number of errors in $A \cup B$ has been reduced to the extent that they are now fully correctable by the error-correcting code $\mathcal{Q}_{k-1}$, using its own decoding algorithm. After using the decoding algorithm of $\mathcal{Q}_{k-1}$, we are now left with the bits in $M \cup A$ as a noisy codeword of the code $\mathcal{R}$, however, the errors on $A$ have been completely removed. Since this code $\mathcal{R}$ has zero errors on its check bits (bits in $A$), we may thus recover the message $M$ perfectly, by the definition of an error-reduction code. Here, one may show that the linear complexity of the error-reduction algorithm for $\mathcal{R}$ translates to a linear complexity of decoding algorithm for $\mathcal{Q}_k$. The parallel decoding algorithm follows via slightly different considerations, but, ultimately, a constant-depth and linear-complexity parallel error-reduction algorithm for $\mathcal{R}$ translates into a logarithmic-depth decoding algorithm for $\mathcal{Q}_k$ with complexity $\mathcal{O}(n \log n)$.

Now that we have described Spielman's procedure at a high level for constructing error-correcting codes from error-reduction codes, let us consider how these steps may be quantised. Indeed, we will go on to show that the high-level picture of constructing quantum error-correcting codes as concatenations of quantum error-reduction codes, using the same structure as that described above, passes smoothly to the quantum case. The majority of the difficulties are faced in the construction of the quantum error-reduction codes themselves, but we tackle these in later sections.

First, let us define quantum error-reduction codes; this is a new notion in the literature as far as we know. The quantum error-reduction codes, and the quantum error-correcting codes that we use them to construct, are all stabiliser codes, in particular CSS codes~\cite{calderbank1996good, gottesman1997stabilizer}. CSS codes may be encoded by taking $n$ qubits in a particular state that we wish to encode, as well as $m$ qubits all in the state $\ket{+}$, and $m$ further qubits in the state $\ket{0}$, and running a unitary encoding circuit on them. The $n$ qubits are known as the ``message qubits'', whereas the two groups of $m$ qubits are known as the ``$X$-check qubits'' and ``$Z$-check qubits'', respectively: collectively called ``check qubits''. Note that the encoding circuit will simply be chosen to be a circuit of $\CNOT$ gates, and in this case, we end up with a CSS code with $m$ checks of $X$ type and $m$ checks of $Z$ type. 

After encoding, we imagine that all $n+2m$ qubits in the codeblock are sent through a noise channel, resulting in some small constant fraction of the message qubits and check qubits being afflicted by errors. Suppose that $v$ message qubits are afflicted by errors, and $t$ check qubits are afflicted by errors. At a high level, we call the code a quantum error-reduction code with error reduction $\epsilon$, if there is some process that can return to us the $m$ message qubits, afflicted by up to $\epsilon\cdot t$ errors. A more formal statement may be found in Definition~\ref{def:qerc} based on the material in Section~\ref{sec:enc_unenc_dec_css}.

One interesting comment to make about the error-reduction process is that it is a hybrid classical-quantum process. Indeed, a quantum circuit unitarily unencodes the noisy codeword from the codespace (this is simply the inverse of the encoding circuit). If no errors occurred, this would just return the message qubits to the state sent as a message, all of the $X$-check qubits to $\ket{+}$, and all of the $Z$-check qubits to $\ket{0}$. However, in general, errors occur, leaving some $X$-check qubits in the state $\ket{-}$, and some $Z$-check qubits in the state $\ket{1}$, as well as some residual Pauli error on the message qubits. A quantum measurement of the $X$-check qubits in the $X$ basis, as well as the $Z$-check qubits in the $Z$ basis, reveals the syndrome, that is, the subset of the stabilisers that the noise has violated. It is then up to a classical error-reduction algorithm to decide on a Pauli correction for the message qubits upon input of this syndrome. This Pauli correction should reduce the size of the error on the message qubits to within the desired size, which is $\epsilon\cdot t$.

We will go on to show in detail that this picture plays out well at a high level for quantum codes in Section~\ref{sec:erc_to_ecc}. One complication that will arise will be in the handling of classical versus quantum information and algorithms. Indeed, we will have to be careful to interlace quantum unencoding operations, as well as classical error-reduction algorithms, in order to prevent a problem of error spreading (described in the next section). This issue is particularly noticeable when handling the parallel algorithms, for which we will need to continue with multiple rounds of error reduction even after the quantum information has been measured out and become classical. We defer the remaining details to Section~\ref{sec:erc_to_ecc}.

\subsubsection{Quantum Error-Reduction Codes from Lossless \texorpdfstring{$Z$}{}-Graphs}\label{sec:overview_qerc_from_lossless}

Whereas the translation from error-reduction codes to error-correcting codes quantises smoothly, up to being careful about details as mentioned, the construction of the quantum error-reduction codes themselves will present several challenges not present on the classical side, as we now describe.

One's first attempt to create a quantum error-reduction code might be to define a quantum CSS code whose $X$-checks and $Z$-checks both correspond to classical error-reduction codes. That is, one might imagine a quantum CSS code with $m$ $X$-check qubits, $n$ message qubits, and $m$ $Z$-check qubits, whose parity-check matrices take the form
\begin{alignat}{4}
    H_X &= (I&&|A&&|0&&)\label{eq:qerc_first_attempt_X_checks}\\
    H_Z &= (0&&|B&&|I&&).\label{eq:qerc_first_attempt_Z_checks}
\end{alignat}
In this equation, the physical qubits of the code are separated into $m$ $X$-check qubits, $n$ message qubits, and $m$ $Z$-check qubits, respectively. $I$ denotes an $m \times m$ identity matrix, and $0$ denotes an $m\times m$ zeros matrix. The matrices $A,B \in \mathbb{F}_2^{m \times n}$ denote the support of the $X$ checks and $Z$ checks on the message qubits, respectively. One might hope that this forms a quantum error-reduction code if the matrices $(I|A)$ and $(I|B)$ are parity-check matrices for classical error-reduction codes on $m$ check bits and $n$ message bits. Of course, the immediate problem one faces is commutativity, that is, this choice of $H_X$ and $H_Z$ only form a valid CSS code if $H_X\cdot H_Z^T = 0$, which is if and only if $A\cdot B^T = 0$. On the face of it, this is problematic because we do not know constructions of classical error-reduction codes for which this is true.

The problem is, however, even worse than this, as follows. Suppose that one did have classical error-reduction codes for which this holds. One could try to implement this as shown in Figure~\ref{fig:qerc_first_attempt}. That is, the $X$ checks are encoded via some $\CNOT$s corresponding to the matrix $A$, and the $Z$ checks are encoded via some $\CNOT$s corresponding to the matrix $B^T$.\footnote{Formally, this would mean that we perform a $\CNOT$ gate from the $i$'th $X$-check qubit to the $j$'th message qubit for every $(i,j) \in [m]\times [n]$ for which $A_{ij} = 1$, and similarly for $B^T$.} Assuming $AB^T = 0$, this would indeed encode the message qubits in the desired state $\ket{\psi}$ into the quantum CSS code defined by the Equations~\eqref{eq:qerc_first_attempt_X_checks} and~\eqref{eq:qerc_first_attempt_Z_checks}.

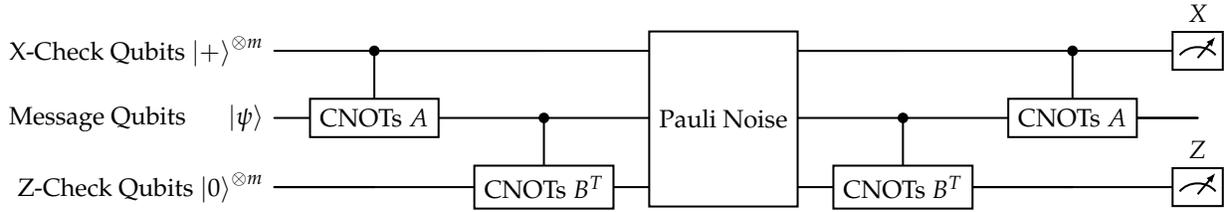
\begin{figure}[ht]
    \centering
    \scalebox{0.95}{\begin{quantikz}[row sep=0.4cm, column sep=0.5cm]
\lstick{X-Check Qubits $\ket{+}^{\otimes m}$}
  & \ctrl{1}
  & \qw
  & \gate[wires=3]{\text{Pauli Noise}}
  & \qw
  & \ctrl{1}
  & \meter{X} \\
\lstick{Message Qubits\hspace*{0.5cm} $\ket{\psi}$}
  & \gate[wires=1]{\text{CNOTs } A}
  & \ctrl{1}
  &
  & \ctrl{1}
  & \gate[wires=1]{\text{CNOTs } A}
  & \qw \\
\lstick{Z-Check Qubits $\ket{0}^{\otimes m}$}
  & \qw
  & \gate[wires=1]{\text{CNOTs } B^T}
  &
  & \gate[wires=1]{\text{CNOTs } B^T}
  & \qw
  & \meter{Z}
\end{quantikz}}
    \caption{A first attempt at constructing and operating a quantum error-reduction code. This figure depicts the process of encoding, noise, unencoding, and stabiliser measurement.}
    \label{fig:qerc_first_attempt}
\end{figure}
After the encoding, noise occurs on all the qubits. Then, the code is unencoded by inverting the encoding circuit,\footnote{Note that in Figure~\ref{fig:qerc_first_attempt}, the encoding circuit is inverted simply by running the $\CNOT$s after the Pauli noise in the opposite order to the $\CNOT$s before the Pauli noise.} and finally the $X$-check qubits are measured out in the $X$ basis, and the $Z$-check qubits are measured out in the $Z$ basis, thus measuring the stabilisers of the code. One's hope is that $X$ errors on the message qubits are measured by the $Z$ checks, and $Z$ errors on the message qubits are measured by the $X$ checks, and may be reduced.

At this point, however, we find the most drastic problem that one encounters when one tries to construct a quantum error-reduction code: the problem of error spreading. The problem is that $X$ errors on the $X$-check qubits, and $Z$ errors on the $Z$-check qubits can spread into the message qubits during the unencoding circuit, and are not detected or reduced by the code. In particular, note that $X$ errors on the $X$-check qubits get transferred into the message qubits via the block ``$\CNOT$s $A$'' in the unencoding, and $Z$ errors on the $Z$-check qubits get transferred into the message qubits via the block ``$\CNOT$s $B^T$'' in the unencoding. In some sense, we need a construction that treats all of the different types of errors on each type of qubit at the same time.

In this work, we find that these problems can all be handled simultaneously by a construction of quantum error-reduction codes that is based on a particular two-way expanding bipartite graph that we term a ``lossless $Z$-graph''. The idea is as follows. First, we wish to add in some support of the $X$ checks to the $Z$-check qubits, and $Z$ checks to the $X$-check qubits, in order to attempt to catch $Z$ errors happening on the $Z$-check qubits, and $X$ errors happening on the $X$-check qubits, respectively. Indeed, we now let our parity-check matrices take the form
\begin{alignat}{4}
    H_X &= (\,I&&|A&&|C&&)\\
    H_Z &= (D&&|B&&|\,I&&),
\end{alignat}
where $C, D \in \mathbb{F}_2^{m \times m}$. Considering the issue of commutativity, we have that this forms a valid quantum CSS code only if $H_X \cdot H_Z^T = 0$, which is if and only if
\begin{equation}\label{eq:C_matrix}
    C = A\cdot B^T + D^T.
\end{equation}
Our strategy will be to specify the matrices $A, B$ and $D$, and then simply let $C$ be chosen according to this equation.

First, we note that the code of this form may be simply operated in the code capacity setting as shown in Figure~\ref{fig:qerc_explain}. This figure shows that the encoding and unencoding circuits have been augmented with new groups of $\CNOT$ gates from the $X$-check qubits to the $Z$-check qubits according to the matrix $D^T$. Note that we have also given names to the $X$ support and $Z$ support of the noise occurring on each group of qubits. It may be initially surprising that the matrix $C$ does not appear in this figure. However, we show in Section~\ref{sec:construct_quantum_error_reduction_codes} that the (un)encoding circuit shown consisting of the three blocks of $\CNOT$s suffices to (un)encode the quantum code, and this holds for any matrices $A, B$ and $D$. Importantly, if the matrices $A, B$ and $D$ are chosen to be sparse, then the encoding may be run in constant depth, and with a linear number of gates.
\begin{figure}[ht]
    \centering
    \scalebox{0.75}{\begin{quantikz}[row sep=0.4cm, column sep=0.5cm]
\lstick{X-Check Qubits $\ket{+}^{\otimes m}$}
  & \ctrl{1}
  & \qw
  & \ctrl{2}
  & \gate[wires=1]{X_xZ_x}
  & \ctrl{2}
  & \qw
  & \ctrl{1}
  & \meter{X} \\
\lstick{Message Qubits\hspace*{0.5cm} $\ket{\psi}$}
  & \gate[wires=1]{\text{CNOTs } A}
  & \ctrl{1}
  & \qw
  & \gate[wires=1]{X_qZ_q}
  & \qw
  & \ctrl{1}
  & \gate[wires=1]{\text{CNOTs } A}
  & \qw \\
\lstick{Z-Check Qubits $\ket{0}^{\otimes m}$}
  & \qw
  & \gate[wires=1]{\text{CNOTs } B^T}
  & \gate[wires=1]{\text{CNOTs }D^T}
  & \gate[wires=1]{X_zZ_z}
  & \gate[wires=1]{\text{CNOTs }D^T}
  & \gate[wires=1]{\text{CNOTs } B^T}
  & \qw
  & \meter{Z}
\end{quantikz}}
    \caption{The encoding, noise, unencoding, and syndrome measurement of our quantum error-reduction codes. The Pauli noise on each group of qubits is labelled. For example, the $X$ support of the noise on the message qubits may be labelled by a $X$-type Pauli $X_q$, which may equivalently be thought of as a bit string in $\mathbb{F}_2^n$.}
    \label{fig:qerc_explain}
\end{figure}
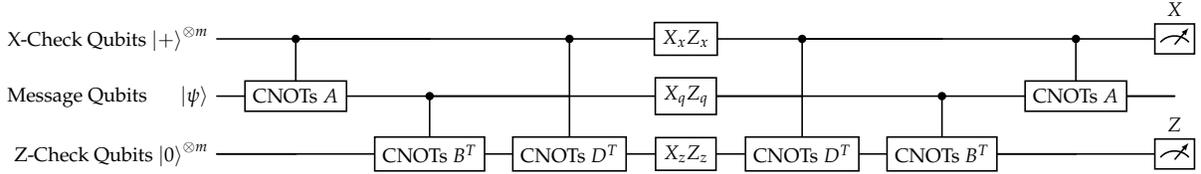

By considering only the matrices $A, B$ and $D$, and encoding circuits of the above form, and letting $C$ be according to the form of Equation~\eqref{eq:C_matrix}, we have sidestepped the problem of ensuring commutativity, but retain the problem of error spreading, and the need to suppress errors that have spread to the message qubits. In order to make progress, we need to consider the exact syndromes that are measured, as well as the residual errors on the message qubits after error spreading in the unencoding circuit. Indeed, one can show that the syndromes measured from the $Z$ checks and $X$ checks are, respectively,
\begin{align}
    \sigma_z &= D\cdot X_x + B \cdot X_q + X_z,\\
    \sigma_x &= Z_x + A\cdot Z_q + C\cdot Z_z.
\end{align}
Moreover, we find that the residual errors on the message qubits, of $X$ type and $Z$ type are, respectively,
\begin{align}
    X_{Res} &\coloneq A^T\cdot X_x + X_q,\\
    Z_{Res} &\coloneq Z_q + B^T\cdot Z_z.
\end{align}
These are exactly the errors on the message qubits for which we want to develop approximations, in order to reduce their size.

Ultimately, we will have to take care of four problems: reducing residual errors of $X$ and $Z$ type, via both sequential and parallel algorithms, but let us start with the most simple case, which is the problem of reducing the $X$ errors via a sequential algorithm (tackled in Section~\ref{sec:sequential_X_reduction_algos}). The reason that the $X$ errors are particularly simple to handle is that the syndrome and residual error of this problem are both entirely in terms of matrices over which we have direct control, that is, $A$, $B$ and $D$. Examining the syndrome in this case, which is $\sigma_z$, we note that it takes the form of a syndrome for a classical error reduction problem. In this analogy, $X_x$ and $X_q$ may be thought of as message bit errors, and $X_z$ as check bit errors. 

To make this analogy more concrete, suppose we were to define a new matrix $B' \coloneq (D|B)$, which has $m$ rows and $n'$ columns, where $n' \coloneq m+n$. Then, we suppose $(B'|I)$ is a parity-check matrix for a classical error-reduction code. This means that, given a message bit error $V \in \mathbb{F}_2^{n'}$ and a check bit error $T \in \mathbb{F}_2^m$, where $|V| = v$ and $|T| = t$, we measure a syndrome $B' \cdot V + T$. Since $(B'|I)$ is the parity-check matrix for a classical error-reduction code, there is then an algorithm allowing us to develop an approximation to $V$, up to an error of size $\leq \epsilon\cdot t$.

Returning to our problem of reducing the residual $X$ error, what we aim to do, therefore, is let $H_Z = (D|B|I)$ be a parity-check matrix for a classical error-reduction code (with $n+m$ message bits and $m$ check bits), at which point we will be able to develop approximations to $X_x$ and $X_q$, that are correct up to an error of size $\leq \epsilon \cdot |X_z|$. Of course, the error that we are actually aiming to approximate is not $X_x$ or $X_q$, but $X_{Res}$. We can use our approximations for $X_x$ and $X_q$ to calculate an approximation for $X_{Res}$. One might be worried, however, that because $X_x$ is multiplied by the matrix $A^T$ in the expression for $X_{Res}$, that having a good approximation for $X_x$ and $X_q$ may not result in a good approximation to $X_{Res}$, that is, the error in our approximation to $X_x$ could ``blow up'' so that the final approximation to $X_{Res}$ in not a good one. In order to handle this, we will have to ensure that the amount of error reduction (called $\epsilon$ above) achieved on the error $X_x$ is small with respect to the maximum column weight of $A^T$, thus ensuring that the final approximation to $X_{Res}$ is still a good one.

At this point, we must explain how we can ensure that the factor of error reduction on $X_x$ (the $\epsilon$) is sufficient, and in order to do this, we must explain the actual construction of the classical error-reduction code by which the matrix $(D|B|I)$ gets defined. Our inspiration comes from the randomised construction of classical error-reduction codes in Spielman~\cite{spielman1995linear}. There, a classical error-reduction code with parity-check matrix $(B'|I)$ gets defined by considering a one-sided lossless expander graph (which Spielman obtains randomly), with $n'$ ``left'' vertices, corresponding to message bits, and $m$ ``right'' vertices, corresponding to check bits. The matrix $B'$ is defined as the adjacency matrix of the graph. Spielman showed that the lossless expansion may enable the error reduction, where the error reduction itself takes place via some simple small set flip algorithm.

We take a similar approach here, by letting $(D|B|I)$ be a parity-check matrix for a classical error-reduction code, which is achieved again via a lossless expander graph. Indeed, we consider a bipartite graph with two groups of ``left'' vertices, of sizes $m$ and $n$, and one group of ``right'' vertices, of size $m$. This must expand in the sense that small sets of the $m$ and $n$ left-hand vertices must jointly have a large number of neighbours amongst the right-hand vertices. We show that a similar small-set flip algorithm, inputted with the syndrome $\sigma_z = D\cdot X_x + B\cdot X_q + X_z$, may produce approximations to $X_x$ and $X_q$, such that the remaining error has size that is small in $X_z$. In particular, supposing that in our bipartite graph, all the degrees of the left-hand set of $m$ nodes are some constant $\Delta_2 \in \mathbb{N}$, the remaining error on $X_x$ will be shown to have size $\lesssim \frac{|X_z|}{\Delta_2}$. We see that, by taking this degree $\Delta_2$ to be some very large constant, a large amount of error reduction may be achieved. In particular, if we suppose the (column) sparsity of the matrix $A^T$ is $\sim \Delta_1$, then as long as $\Delta_2 \gg \Delta_1$, our final calculated approximation to $X_{Res}$ will still be a good one. It is interesting to comment at this point that, in our case, we have to be very careful about the amount of error reduction that we achieve, such as this factor $\Delta_2$. In Spielman's work, smaller factors, such as a simple error reduction by a factor of $2$, turn out to be enough.

To summarise this part, the reduction of the $X$ errors is made possible by taking the matrix $H_Z = (D|B|I)$ to be the parity-check matrix for a classical error-reduction code, which we achieve by letting $(D|B)$ be the adjacency matrix for some bipartite graph with one-sided expansion. As long as the degree $\Delta_2$ of the nodes corresponding to the columns of $D$ is much larger than the (column) sparsity $\sim \Delta_1$ of $A^T$, then we will be able to achieve a reduction in the size of $X_{Res}$.

Since we have already constrained the matrices $A, B$ and $D$ so much, it is not clear at this point how we will be able to handle the reduction of the $Z$ errors, either with sequential or parallel algorithms. However, we may proceed as follows (see Section~\ref{sec:sequential_Z_reduction_algos} for the sequential algorithm). We re-write the $X$ syndrome via
\begin{align}
    \sigma_x &= Z_x + A\cdot Z_q + C\cdot Z_z\\
    &= Z_x + A \cdot Z_q + (A\cdot B^T+D^T)\cdot Z_z\\
    &=Z_x + A\cdot (Z_q + B^T\cdot Z_z) + D^T\cdot Z_z\\
    &= Z_x + A\cdot Z_{Res} + D^T\cdot Z_z.
\end{align}
We see now that the exact residual error of $Z$ type we wish to reduce is sitting under the matrix $A$. Moreover, the syndrome has again taken the form of that in a classical error reduction problem. By taking $(I|A|D^T)$ to be the parity-check matrix for a classical error-reduction code, we may aim to directly reduce the error $Z_{Res}$, rather than developing approximations to two errors and calculating an approximation to the residual error, as we did for the $X$ errors. In turn, we wish to now have $(A|D^T)$ be the adjacency matrix for a one-sided lossless expander.

In total, our requirements for the matrices $A, B$ and $D$ are as follows:
\begin{enumerate}
    \item The matrix $(D|B)$ exhibits small-set expansion;
    \item The matrix $(A|D^T)$ exhibits small-set expansion;
    \item The (column) sparsity of the matrix $D$ must be much greater than the (column) sparsity of the matrix $A^T$. 
\end{enumerate}
Each of these requirements may be simultaneously satisfied by taking the three matrices $A, B$ and $D$ to be according to a two-way expanding structure, which we term a lossless $Z$-graph, described as follows, and depicted in Figure~\ref{fig:lossless_Z_graph}. The lossless $Z$-graph is a bipartite graph, with two sets of ``left'' vertices called $L_1$ and $L_2$, of sizes $n$ and $m$, respectively, and similarly two sets of ``right'' vertices called $R_1$ and $R_2$, of sizes $n$ and $m$, respectively. The graph obtained by restriction to the vertices $(L_1, R_2)$ is $(\Delta_1, \Delta_1')$-biregular, where $\Delta_1' = \frac{n}{m}\Delta_1$. Similarly, the graph obtained by restriction to $(R_1, L_2)$ is $(\Delta_1, \Delta_1')$-biregular. Finally, the graph obtained by restriction to the vertices $(L_2, R_2)$ is $\Delta_2$-regular. There are no edges between the vertices $L_1$ and $R_1$, and so the graph appears in the form of the letter $Z$, as shown. We let the matrix $A$ be the adjacency matrix of the graph obtained by restriction to the vertices $(L_1, R_2)$, the matrix $B$ be the adjacency matrix of the graph obtained by restriction to the vertices $(R_1, L_2)$, and finally the matrix $D$ be the adjacency matrix of the graph obtained by restriction to the vertices $(R_2, L_2)$. Equivalently, the matrix $D^T$ may be defined as the adjacency matrix of the graph obtained by restriction to the vertices $(L_2, R_2)$.
\begin{figure}[ht]
    \centering
    \includegraphics[width=0.5\linewidth]{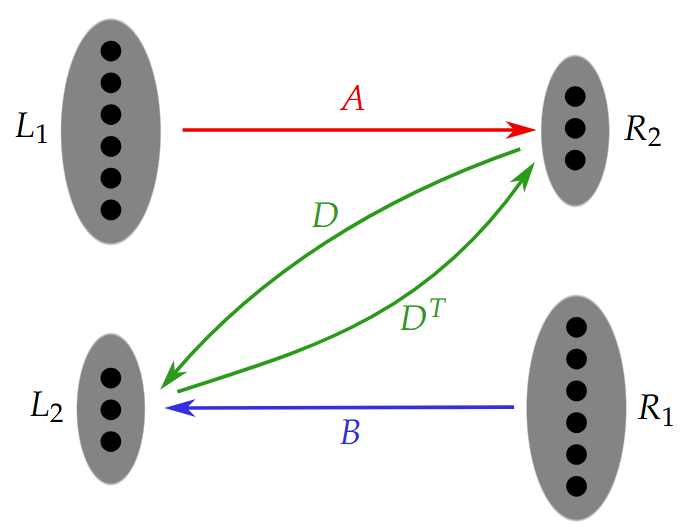}
    \caption{Definition of the matrices $A, B$ and $D$ which define our quantum error-reduction code, from the lossless $Z$-graph. Vertices in $L_1$ and $R_1$ all have degree $\Delta_1$. Vertices in $R_2$ and $L_2$ all have $\Delta_1' = \frac{n}{m}\Delta_1$ edges joining them to vertices in $L_1$ and $R_1$, respectively, and $\Delta_2$ edges joining them to each other.}
    \label{fig:lossless_Z_graph}
\end{figure}

When we call the graph a \textit{lossless} $Z$-graph, we mean that it has the following expansion properties. Suppose we take subsets $S_1 \subseteq L_1$ and $S_2 \subseteq L_2$ of some small constant-fractional size. Then, $S_1$ and $S_2$ must collectively have a large number of neighbours in $R_2$. Let us denote the neighbours of $S_1 \cup S_2$ in $R_2$ as $N_{R_2}(S_1 \cup L_2)$. Noting that $|N_{R_2}(S_1 \cup S_2)| \leq \Delta_1|S_1| + \Delta_2|S_2|$, the graph is a lossless $Z$-graph only if
\begin{equation}\label{eq:rough_lossless_statement}
    |N_{R_2}(S_1 \cup S_2)| \geq (1-\epsilon_1)\Delta_1|S_1| + (1-\epsilon_2)\Delta_2|S_2|
\end{equation}
for some small $\epsilon_1, \epsilon_2$. We also require that subsets $S_1 \subseteq R_1$ and $S_2 \subseteq R_2$ of some small constant-fractional size have a similarly large number of neighbours in $L_2$. At a high level, the lossless $Z$-graph is a two-way lossless expander of mixed degree, but where the two-way expansion is only required to go ``through'' the edges connecting $L_2$ and $R_2$.

By taking $\Delta_1, \Delta_2$ to be constants with $\Delta_2 \gg \Delta_1$, we will have obtained all of the requirements described above, and constructed a quantum error-reduction code. To conclude this section, we finally remark that our parallel error-reduction algorithms for $X$ errors and $Z$ errors, described in Sections~\ref{sec:parallel_erc_algos}, come with some additional considerations, which we discussion in detail there. For now, we comment that the amount of error reduction required in our case is so great that our parallel error-reduction algorithms require especially strong lossless $Z$-graphs, which we can obtain randomly, but we are unable to obtain explicitly. Indeed, the parallel error-reduction algorithms will require lossless $Z$-graphs for which $\epsilon_1 = \mathcal{O}(1/\Delta_1)$ and $\epsilon_2 = \mathcal{O}(1/\Delta_2)$ (see Equation~\eqref{eq:rough_lossless_statement}), which we are unable to obtain via the explicit construction. Moreover, our parallel error-reduction algorithms will require the quantity $\frac{\Delta_1^2}{\Delta_2}$ to be small, which can be obtained randomly (since the random construction allows any two constant integers $\Delta_1, \Delta_2$), but the explicit construction only allows the quantity $\frac{\Delta_1}{\Delta_2}$ to be made small. These two considerations explain why parallel error-reduction algorithms for the explicit quantum codes are not obtained in this work.

\subsubsection{Constructions of Lossless \texorpdfstring{$Z$}{}-Graphs}\label{sec:overview_lossless_construction}

We briefly remark on the construction of the lossless $Z$-graphs themselves, which is handled in Section~\ref{sec:lossless_Z_graphs}.

First, for the random construction, we sample a graph from the natural ensemble, and show that it satisfies the expansion constraints with high probability. That is, we consider fixed sets of vertices $L_1, L_2, R_1, R_2$ with the appropriate number of half-edges going towards the correct end-points, and join them to each other uniformly at random. Showing that the expansion properties hold uses mostly standard techniques, see for example~\cite{hoory2006expander}. One complication will be that showing the joint expansion of (for example) $S_1 \subseteq L_1$ and $S_2 \subseteq L_2$ is not possible directly when the sizes $|S_1|$ and $|S_2|$ are very imbalanced, that is, the quantity $\frac{|S_1|}{|S_2|}$ is very large or very small. However, in such cases, we show that the joint expansion may be reduced to the usual lossless expansion. For example, when $\frac{|S_1|}{|S_2|}$ is very large, the joint expansion may be demonstrated via the expansion of vertices $L_2$ going into $R_2$. Similarly, when $\frac{|S_1|}{|S_2|}$ is very small, the joint expansion may be demonstrated via the expansion of vertices $L_2$ going into $R_2$. See Section~\ref{sec:random_lossless_Z_graphs} for details.

Our explicit construction proceeds via a white-box modification of the recent construction of two-sided lossless expanders by~\cite{HLMRZ25}. At first glance, it might seem that one could use their construction as a black box, by letting the graph between $(L_2, R_2)$ be a two-sided lossless expander, and taking $(L_1, R_2)$ and $(R_1, L_2)$ to both be one-sided lossless expanders. However, even if $S_1 \subset L_1$ and $S_2 \subset L_2$ are individually expanding, the lossless expansion guarantees do not prohibit their neighborhoods from matching exactly and canceling each other out. Thus, we need to open up their construction and modify the individual components.

At a high level, the two-sided expanders of~\cite{HLMRZ25} are constructed via a \emph{local-to-global} framework. Specifically, they overlay a constant-sized ``gadget'' lossless expander many times according to some blueprint, and argue that the properties of the blueprint imply that the global graph will also experience lossless expansion. For our case, we will choose our constant-sized gadget graph to itself be a lossless Z-graph. Then, by overlaying it many times according to the blueprint in~\cite{HLMRZ25}, we are able to argue that the global graph will be a lossless Z-graph as well. See Section~\ref{sec:explicit_lossless_Z_graphs} for details.

\subsection{Future Directions}

In this work, we have constructed the first quantum error-correcting codes with encoding, unencoding and decoding algorithms of low depth and complexity. In particular, we have presented asymptotically good codes with sequential encoding, unencoding, and decoding algorithms of linear complexity, as well as corresponding logarithmic-depth parallel algorithms, with complexities $\mathcal{O}(n)$, $\mathcal{O}(n)$, and $\mathcal{O}(n \log n)$, respectively. For our explicit construction, the same is obtained, without the parallel classical decoding algorithm. The first natural direction for future work is therefore as follows.
\\

\noindent\textbf{Open Problem 1:} Can we construct explicit and asymptotically good quantum error-correcting codes with low-depth parallel classical decoding algorithms?
\\

As discussed above, our explicit lossless $Z$-graphs fail to meet the conditions required of our parallel classical decoding algorithm on two fronts, namely that we require \textit{very} small $\epsilon_1, \epsilon_2$, namely $\mathcal{O}(1/\Delta_1)$, $\mathcal{O}(1/\Delta_2)$, respectively, as well as the fact that we require $\Delta_2$ to be quadratically larger than $\Delta_1$, rather than simply a constant multiple larger. We note that the former of these seems to be a particular obstruction to the current approach, since the explicit two-sided lossless expander~\cite{HLMRZ25}, are not yet known with $\epsilon = \mathcal{O}(1/\Delta)$.
\\

\noindent\textbf{Open Problem 2:} Can we understand the exact rate/distance tradeoff of explicit or randomised asymptotically good quantum codes with low-depth and low-complexity encoders and decoder?
\\

In this work, we have only showed that our quantum codes are asymptotically good, and we also show that our construction can achieve rate close to $1$. However, the exact rate/distance tradeoff is left to future work. Interestingly, on the classical side, following Spielman's work, there has been a fruitful line of work studying the corresponding problem~\cite{brown2005repeat,guruswami2005linear,druk2014linear,narayanan2019subquadratic,brehm2025linear}. Notably, the recent work~\cite{brehm2025linear} constructs linear-time classical codes with parameters achieving the Gilbert-Varshamov bound, whose duals have the same properties. Given the importance of properties of the dual code when forming quantum CSS codes, it could be interesting to explore the connection of the codes in that work to quantum codes with fast encoders.
\\

\noindent\textbf{Open Problem 3:} Can we construct (asymptotically good) quantum codes with low-depth and complexity \textit{fault-tolerant} encoders?
\\

As discussed, this paper considers only the code capacity setting which is relevant in communication problems and information theory. However, it is also interesting to ask whether codes with low-depth and complexity encoders exist in a fault-tolerant setting. Relatedly, it is not known whether there exist LDPC codes with low-depth or complexity encoders, even on the classical side. Must there be necessarily be a tradeoff between the parameters of a code, the depth or complexity or its encoder, and its maximum check weight?

\section{Outline of the Paper}

In Section~\ref{sec:prelims}, we describe the necessary preliminary material for this paper in order to make it accessible to a broad theoretical computer science audience; this may be skipped by those already familiar with the essentials of quantum error correction. In Section~\ref{sec:erc_to_ecc}, we describe how quantum error-correcting codes may be constructed from quantum error-reduction codes. In Section~\ref{sec:quantum_error_reduction_codes}, we describe how quantum error-reduction codes may be constructed from ``lossless $Z$-graphs''. In Section~\ref{sec:lossless_Z_graphs}, we describe the construction of lossless $Z$-graphs, both randomly and explicitly. Finally, in Section~\ref{sec:all_together}, we put all the ideas together to conclude the proofs of our main results, Theorems~\ref{thm:main_randomised} and~\ref{thm:main_explicit}.

\section{Preliminaries}\label{sec:prelims}

In order to make the paper accessible to a general theoretical computer science audience, we present in this section the essential information on quantum error-correcting codes, in particular CSS codes. Those already familiar with CSS codes can safely skip this section.

The state of a single qubit is specified by a vector in $\mathbb{C}^2$ of unit norm. The standard basis, also called the computational basis, is denoted by two vectors $\ket{0}$ and $\ket{1}$. Other important states are $\ket{+} = (\ket{0}+\ket{1})/\sqrt{2}$ and $\ket{-} = (\ket{0}-\ket{1})/\sqrt{2}$. The state of $n$ qubits is specified by a vector in $(\mathbb{C}^2)^{\otimes n}$ of unit norm. Tensor product symbols are often dropped for readability, for example, a computational basis state on $n$ qubits is commonly denoted $\ket{x}$ for $x \in \{0,1\}^n$.

Quantum stabiliser codes are by far the most commonly studied class of quantum codes, of which quantum CSS codes are a commonly-studied sub-class. To introduce quantum stabiliser codes, we begin with the single-qubit Pauli operators,
\begin{equation}
    I = \begin{pmatrix}
        1 & 0\\
        0 & 1
    \end{pmatrix}\hspace*{1cm} X = \begin{pmatrix}
        0 & 1\\
        1 & 0
    \end{pmatrix}\hspace*{1cm} Y = \begin{pmatrix}
        0 & -i\\
        i & 0
    \end{pmatrix}\hspace*{1cm} Z = \begin{pmatrix}
        1 & 0\\
        0 & -1
    \end{pmatrix}.
\end{equation}
Up to an overall phase, the Pauli operators for $n$ qubits are constructed by taking tensor products of the above operators. Given a set $\mathcal{S}$ of $n$-qubit Pauli operators which commute, the stabiliser code corresponding to $\mathcal{S}$ is the set of $n$-qubit states with eigenvalue $+1$ under all operators in $\mathcal{S}$ (the set of states \textit{stabilised} by $\mathcal{S})$, that is,
\begin{equation}
    \mathcal{Q} = \left\{\ket{\psi} \in (\mathbb{C}^2)^{\otimes n}: P\ket{\psi} = \ket{\psi} \text{ for all } P \in \mathcal{S}\right\}.
\end{equation}
Different $\mathcal{S}$ may result in the same code $\mathcal{Q}$. A quantum CSS code $\mathcal{Q}$ is a stabiliser code which has an $\mathcal{S}$ for which every Pauli in $\mathcal{S}$ is either a tensor product of only $I$ and $X$ (an $X$-type stabiliser), or a tensor product of only $I$ and $Z$ (a $Z$-type stabiliser). Thus, a quantum CSS code may be specified by two binary matrices, $H_X$ and $H_Z$, each with $n$ columns. $H_X$ specifies a list of $X$-Pauli operators in the code's stabiliser, known as the $X$ stabilisers or $X$ checks, and $H_Z$ specifies a list of $Z$-Pauli operators in the code's stabiliser, known as $Z$ stabilisers or $Z$ checks. Each row of $H_X$ specifies an $X$ check, where a $0/1$ represents an $I/X$ in the corresponding stabiliser, and similarly for $H_Z$. In order to form a valid stabiliser code, all of the code's stabilisers must commute. Trivially, all of the $X$ checks commute amongst themselves, and all of the $Z$ checks commute amongst themselves. To form a quantum CSS code from the matrices $H_X$ and $H_Z$, it is therefore necessary and sufficient that each $X$ check commutes with each $Z$ check. Since $XZ = -ZX$, this is true if and only if every row of $H_X$ has an even overlap with every row of $H_Z$, which is to say that $H_X\cdot H_Z^T = 0$, with arithmetic performed in binary. 

Given a valid CSS code, the number of logical qubits that it encodes is $n - \mathsf{rank} \;H_X - \mathsf{rank} \;H_Z$. While quantum errors are in principle continuous, the principle of error discretisation means that we need only consider a discrete set of errors. In particular, we may consider all errors to be Pauli operators. The distance of the code is the minimum number of qubits on which an error must occur in order that the error is undetectable and can change the logical state of the code. For a CSS code, there is a notion of $X$-distance and $Z$-distance, respectively:
\begin{align}
    d_X &= \min\{|v|: v \in \ker H_X\setminus \im H_Z^T\}\\
    d_Z &= \min\{|v|: v \in \ker H_Z\setminus \im H_X^T\}.
\end{align}
This is the minimum weight of an $X$-type Pauli and a $Z$-type Pauli, respectively, that can undetectably change the logical information. The distance of the code is then
\begin{equation}
    d = \min\{d_X, d_Z\}.
\end{equation}

The remaining quantum operations that will be required for our exploration will be the $\CNOT$ gate, and the notion of single-qubit Pauli measurements. The $\CNOT$ gate is a two-qubit gate acting as follows,
\begin{equation}
    \CNOT\ket{00} = \ket{00}, \hspace*{1cm} \CNOT\ket{01} = \ket{01}, \hspace*{1cm} \CNOT\ket{10} = \ket{11}, \hspace*{1cm} \CNOT\ket{11} = \ket{10},
\end{equation}
where the first qubit is referred to as the control, and the second as the target. If we wish to emphasise the control and target, then we may write $\mathsf{CNOT}_{ab}$, where qubit $a$ is the control and qubit $b$ is the target.

For single-qubit Pauli measurements, the only necessary information for our exploration will be the fact that a measurement of a qubit in the $X/Z$ basis returns the outcome $y \in \{0,1\}$ with probability $1$ if the qubit is an eigenstate of the $X/Z$ operator with eigenvalue $(-1)^y$.

It will be useful for us to draw circuit diagrams, which depict quantum computations, and are read from left to right. Here, a single wire denotes a single qubit, whereas a slashed wire denotes multiple qubits (where the number may be labelled). Preparation of one or multiple qubits in a particular state may be indicated by writing that state before the wire. For example, the following diagrams depict, respectively, the preparation of a single qubit in the state $\ket{0}$, the preparation of a single qubit in the state $\ket{+}$, and the preparation of $k$ qubits in the state $\ket{\psi}$.
\[
\begin{quantikz}
\lstick{$\ket{0}$} & \qw
\end{quantikz}
\qquad
\begin{quantikz}
\lstick{$\ket{+}$} & \qw
\end{quantikz}
\qquad
\begin{quantikz}
\lstick{$\ket{\psi}$} & \qw \qwbundle{k}
\end{quantikz}
\]
The $\CNOT$ operation (controlled on the top wire and targeting the lower wire), and $X/Z$ measurements are depicted as follows.
\[
\begin{quantikz}
\qw & \ctrl{1} & \qw \\
\qw & \targ{}  & \qw
\end{quantikz}
\qquad
\begin{quantikz}
\qw  & \meter{X} \\
\end{quantikz}
\qquad
\begin{quantikz}
\qw & \meter{Z} \\
\end{quantikz}
\]
The following commutation relations will be useful for us, and may be readily verified:
\begin{alignat}{3}\label{eq:CNOT_commutation_rules}
    \mathsf{CNOT}_{12}X_1 &= X_1X_2\mathsf{CNOT}_{12}, \hspace*{1cm}&&\mathsf{CNOT}_{12}X_2&&=X_2\mathsf{CNOT}_{12}\\\mathsf{CNOT}_{12}Z_1 &= Z_1\mathsf{CNOT}_{12}, \hspace*{1cm}&&\mathsf{CNOT}_{12}Z_2 &&= Z_1Z_2\mathsf{CNOT}_{12},
\end{alignat}
where, as always, tensor product symbols are dropped and subscripts are used to denote the qubit on which the given gate acts. Graphically, we have
\[
\begin{quantikz}
\qw&\gate[wires=1]{X}&\ctrl{1} & \qw & \qw\\
\qw & \qw & \targ{}  & \qw & \qw 
\end{quantikz} = \begin{quantikz}
\qw&\qw&\ctrl{1} & \gate[wires=1]{X} & \qw\\
\qw &\qw & \targ{}  & \gate[wires=1]{X} & \qw 
\end{quantikz}, \hspace*{1cm} \begin{quantikz}
\qw&\qw & \ctrl{1} & \qw & \qw\\
\qw & \gate[wires=1]{X} & \targ{}  & \qw & \qw 
\end{quantikz} = \begin{quantikz}
\qw&\qw&\ctrl{1} & \qw & \qw\\
\qw &\qw & \targ{}  & \gate[wires=1]{X} & \qw 
\end{quantikz}
\]
and
\[
\begin{quantikz}
\qw&\gate[wires=1]{Z}&\ctrl{1} & \qw & \qw\\
\qw & \qw & \targ{}  & \qw & \qw 
\end{quantikz} = \begin{quantikz}
\qw&\qw&\ctrl{1} & \gate[wires=1]{Z} & \qw\\
\qw &\qw & \targ{}  & \qw & \qw 
\end{quantikz}, \hspace*{1cm} \begin{quantikz}
\qw&\qw & \ctrl{1} & \qw & \qw\\
\qw & \gate[wires=1]{Z} & \targ{}  & \qw & \qw 
\end{quantikz} = \begin{quantikz}
\qw&\qw&\ctrl{1} & \gate[wires=1]{Z} & \qw\\
\qw &\qw & \targ{}  & \gate[wires=1]{Z} & \qw 
\end{quantikz}.
\]

\section{From Quantum Error Reduction to Quantum Error Correction}\label{sec:erc_to_ecc}

In this work, we will construct linear-time encodable and decodable quantum error-correcting codes, where we take inspiration from the original construction of linear-time encodable and decodable classical error-correcting codes due to Spielman~\cite{spielman1995linear}. It is interesting that certain elements of Spielman's construction will prove very amenable to quantisation, whereas others will prove much more difficult.

At a high level, Spielman's construction proceeds by concatenating a series of ``classical error-reduction codes'' in a particular structure, so that the whole thing forms an error-correcting code. The low complexity of encoding and decoding at each level of the concatenation translates to a low complexity encoding and decoding of the whole code. We will go on to find that the global picture, that of concatenating error-reduction codes to form error-correcting codes, translates mostly smoothly from the classical to the quantum case. There will be a minor alteration in the definition of error reduction. In particular, in the classical case, one talks about message and check bits, whereas in the quantum case, we will talk about message qubits, $X$-check qubits, and $Z$-check qubits, as we shortly introduce. Additionally, proving the translation from quantum error reduction for the parallel algorithms will require particular care, although proving the statement for the sequential algorithms will be a direct translation of the corresponding proof in Spielman~\cite{spielman1995linear}. The majority of the work will be in the actual construction of the quantum error-reduction codes, which is handled in Sections~\ref{sec:quantum_error_reduction_codes} and~\ref{sec:lossless_Z_graphs}. 

To prepare us for defining quantum error-reduction codes, we will now describe the encoding, unencoding and decoding of quantum CSS codes.

\subsection{Encoding, Unencoding and Decoding of Quantum CSS Codes}\label{sec:enc_unenc_dec_css}

We will now discuss the encoding, unencoding and decoding of quantum CSS codes in a way that is suitable for this context. As we work in the channel capacity/communication setting, the process of encoding quantum information, noise, and unencoding, using a  general quantum CSS code, may be described as follows (see also Figure~\ref{fig:encoding_noise_unencoding}).
\begin{figure}[ht]
    \centering
    
\scalebox{1.1}{
\begin{quantikz}[row sep=0.6cm, column sep=0.5cm]
\lstick{X-Check Qubits \hspace*{0.4cm}$\ket{+}^{\otimes m_X}$}
  & \qwbundle{m_X}
  & \gate[wires=3]{\text{Encoding}}
  & \qw
  & \gate[wires=3]{\text{Noise}}
  & \qw
  & \gate[wires=3]{\text{Unencoding}}
  & \qw
  & \meter{X} \\
\lstick{Message Qubits\hspace*{1.1cm} $\ket{\psi}$}
  & \qwbundle{n}
  & \qw
  & \qw
  & 
  & \qw
  & \qw
  & \qw
  & \qw \\
\lstick{Z-Check Qubits \hspace*{0.525cm}$\ket{0}^{\otimes m_Z}$}
  & \qwbundle{m_Z}
  & \qw
  & \qw
  & 
  & \qw
  & \qw
  & \qw
  & \meter{Z}
\end{quantikz}
}
\caption{}
    \label{fig:encoding_noise_unencoding}
\end{figure}

In words, we begin with $N$ qubits, where $N$ is the code's block length. The first $m_X = \mathsf{rank}\;H_X$ of these qubits each begin in the state $\ket{+}$, and are called $X$-check qubits, and the last $m_Z = \mathsf{rank}\;H_Z$ of the qubits each begin in the state $\ket{0}$, and are called $Z$-check qubits. We refer to the $X$-check qubits and $Z$-check qubits collectively as check qubits. The remaining $n$ qubits are referred to as the message qubits, and are initialised in whichever state we wish to encode and send through the noise channel; here we depict this state as some general $n$-qubit state $\ket{\psi}$. The encoding operation is a quantum operation responsible for preparing the logical state $\overline{\ket{\psi}}$. As we work in the channel capacity setting, this operation occurs without noise or faults. The important point is that, after the encoding operation, but before the noise occurs, the $n$ qubits are in a quantum state that is a simultaneous $+1$ eigenstate of all the code's stabilisers. In particular, we may always take the encoding operation to consist entirely of $\mathsf{CNOT}$ gates, and then $m_X$ is the number of the code's $X$ checks, and $m_Z$ is the number of the code's $Z$ checks. One notes that, before the encoding operation, the $n$ qubits are in a $+1$ eigenstate of every $X$ operator acting on one of the first $m_X$ qubits, and a $+1$ eigenstate of every $Z$ operator acting on one of the last $m_Z$ qubits. After the encoding operation, the $N$ qubits are in a $+1$ eigenstate of each of these operations commuted through the encoding operation. Since we take the encoding operation to consist entirely of $\CNOT$ operations, the resulting stabilisers, both $X$ checks and $Z$ checks, may be calculated via Equation~\eqref{eq:CNOT_commutation_rules}. In summary, the encoding operation causes the $X$ checks and $Z$ checks to spread out across the $N$ qubits.

The noise operation may then be taken to be some tensor product of Paulis. The Pauli error occurring on the $X$-check qubits may be written $X_xZ_x$, where $X_x$ and $Z_x$ are respectively Paulis of pure $X$ and $Z$ type. Note that one may think of both of $X_x$ and $Z_x$ as bit strings of length $N$ if desired, denoting the positions in which they are non-trivial. Similarly, we denote the Pauli errors occurring on the message qubits as $X_qZ_q$ and $X_zZ_z$, respectively. The situation appears as follows.
\begin{figure}[ht]
    \centering
\scalebox{1.1}{
\begin{quantikz}[row sep=0.4cm, column sep=0.5cm]
\lstick{X-Check Qubits \hspace*{0.4cm}$\ket{+}^{\otimes m_X}$}
  & \qwbundle{m_X}
  & \gate[wires=3]{\text{Encoding}}
  & \qw
  & \gate[wires=1]{X_xZ_x}
  & \qw
  & \gate[wires=3]{\text{Unencoding}}
  & \qw
  & \meter{X} \\
\lstick{Message Qubits\hspace*{1.1cm} $\ket{\psi}$}
  & \qwbundle{n}
  & \qw
  & \qw
  & \gate[wires=1]{X_qZ_q}
  & \qw
  & \qw
  & \qw
  & \qw \\
\lstick{Z-Check Qubits \hspace*{0.525cm}$\ket{0}^{\otimes m_Z}$}
  & \qwbundle{m_Z}
  & \qw
  & \qw
  & \gate[wires=1]{X_zZ_z}
  & \qw
  & \qw
  & \qw
  & \meter{Z}
\end{quantikz}
}
\caption{}
    \label{fig:encoding_noise_unencoding_labelled_Paulis}
\end{figure}

Again, since we work in the channel capacity setting, the unencoding operation is a quantum operation that is noiseless and without faults. It may be taken to be the inverse of the encoding operation. Because the $\mathsf{CNOT}$ gate is self-inverse, it may just be taken to be the same sequence of $\mathsf{CNOT}$s as the encoding operation with the order reversed. In the last step, every $X$-check qubit is measured in the $X$ basis, and every $Z$-check qubit is measured in the $Z$ basis (note our slight abuse of notation in using single-qubit measurement symbols in Figure~\ref{fig:encoding_noise_unencoding_labelled_Paulis}, where we are in fact denoting the measurement of $m_X$ qubits, each in the $X$ basis, and $m_Z$ qubits, each in the $Z$ basis).

At this point, at the end of the quantum operations, there is some Pauli error remaining on the message qubits that we wish to correct. Our attempt to correct this error is informed by our measurement outcomes; the measurement outcomes are turned into a candidate correction via a classical decoding algorithm.

When $Z$ errors occur in the noise step, immediately after that step, they may cause the $N$ qubits to no longer be in a $+1$ eigenstate of certain $X$ checks. In turn, this will cause certain $X$ measurements at the end of the circuit to return the value $1$ rather than $0$. Likewise, when $X$ errors occur in the noise step, immediately after that step, they may cause the $N$ qubits to no longer be in a $+1$ eigenstate of certain $Z$ checks any longer, which will in turn cause some of the $Z$ measurements at the end of the circuit to return the value $1$ rather than $0$ (which they would if there were no noise). To make this precise, we give more notation to our quantum code. That is, we write our $X$- and $Z$-check matrices as
\begin{align}
    H_X &= \begin{pmatrix}
        H_{xx} & H_{xq} & H_{xz}
    \end{pmatrix} \in \mathbb{F}_2^{m_X \times n}\\
    H_Z &= \begin{pmatrix}
        H_{zx} & H_{zq} & H_{zz}
    \end{pmatrix} \in \mathbb{F}_2^{m_Z \times n}.
\end{align}
Given $a \in \{x,z\}$, $H_{ax}$, $H_{aq}$ and $H_{az}$ denote the support of the $a$-type checks on the $X$-check qubits, message qubits and $Z$-check qubits, respectively. Then, one finds that the $X$-syndrome, that is, the collection of measurement outcomes of the $X$-basis measurements, written as a vector in $\mathbb{F}_2^{m_X}$, is
\begin{equation}
    \sigma_x = H_{xx}Z_x + H_{xq}Z_q + H_{xz}Z_z,
\end{equation}
where arithmetic takes place in binary, and similarly the $Z$-syndrome is
\begin{equation}
    \sigma_z = H_{zx}X_x + H_{zq}X_q + H_{zz}X_z
\end{equation}
as a vector in $\mathbb{F}_2^{m_Z}$. Notice that we are conflating the meaning of the symbols $X_x, Z_x$, and so on, between Paulis, and binary vectors denoting their support.

Finally, the decoding process is a classical process that takes as input the syndromes $\sigma_x$ and $\sigma_z$ and attempts to compute the best correction to make to fix the residual Pauli error on the message qubits.

Before moving on, there are two remarks we wish to make. The first is that we are only interested in fixing the residual error on the message qubits. Typically in quantum error correction, the information remains permanently inside the error-correcting code, and measurements make use of ancillary qubits. In that setting, one wishes to correct all errors on all $N$ qubits. Here, we are encoding into and unencoding out of the quantum code. $X$-check and $Z$-check qubits are measured and discarded, and we are only interested in correcting the residual error on the message qubits. Our second remark is to emphasise one interesting difference in this quantum setup with respect to the classical setup considered by Spielman~\cite{spielman1995linear}. In the classical setup, the unencoding and decoding operations, the operations that take your information out of the codespace, and attempt to find a correction, respectively, are rolled into one operation that is simply called ``decoding''. In the quantum setting, it is interesting that these two operations must be separated, so that in particular the former is a quantum operation (a circuit of $\mathsf{CNOT}$s), whereas the latter is a classical algorithm.

\subsection{Definition of Quantum Error-Reduction Codes}

Following the corresponding classical definition of Spielman~\cite{spielman1995linear}, our definition of quantum error-reduction codes is as follows.
\begin{definition}[Quantum Error-Reduction Code]\label{def:qerc}

Consider a quantum CSS of block length $N$. Suppose, as in Figure~\ref{fig:encoding_noise_unencoding_labelled_Paulis}, we perform perfect unitary encoding and unencoding, sandwiching Pauli errors, and at the end we measure $X$ and $Z$ stabilisers. The code is called a quantum error-reduction code of rate $r$, error reduction $\epsilon$, and reducible distance $\delta$, if it has $rN$ message qubits, and a classical ``error reduction'' algorithm that does the following. If there are Pauli errors on $v$ message qubits and $t$ check qubits, where $v \leq \delta N$, and $t \leq \delta N$, then the algorithm, taking the stabiliser measurement outcomes as its input, outputs a Pauli correction for the message qubits, after which there will be a Pauli error of weight at most $\epsilon t$ on the message qubits.
    
\end{definition}
\begin{remark}
    From the point of view of quantum computing, it is interesting to note that the notion of a (quantum) error-reduction code is reminiscent of the well-studied notion of single-shot quantum error correction~\cite{bombin2015single, campbell2019theory, gu2024single}. In quantum computing, measurements are known to be faulty, and so one cannot generally trust a single round of stabiliser measurements. The canonical approach to overcome this challenge is to repeat the measurements a number of times that scales with the distance to gain confidence in one's measurement outcomes, although doing so introduces an unfortunate time overhead. It has been shown that for some codes, a single round (or a constant-number of rounds) of check measurement is sufficient, not to perfectly correct the error, but to control the error and prevent it from growing, and a single more reliable measurement is deferred to the end of the computation. We say that the code admits single-shot quantum error correction. In this analogy, errors on our check (qu)bits correspond to faults in the quantum measurements, and error reduction corresponds to controlling the error on the message (qu)bits. While we are not claiming a formal connection, the analogy is rather striking.
\end{remark}

\subsection{Concatenation Structure}\label{sec:concat_structure}

We now show how quantum error-reduction codes may be concatenated to form quantum error-correcting codes. As mentioned, a direct quantisation of Spielman's structure~\cite{spielman1995linear} is sufficient. The analysis of the sequential algorithms also passes across nicely, although the analysis of the parallel algorithms is slightly more involved.

We define our family of quantum error-correcting codes $\left(\mathcal{Q}_k\right)_{k=0}^\infty$ as follows. Let $r^{(1)}$ and $r^{(2)}$ be numbers in $(0,1)$ such that
\begin{equation}\label{eq:qerc_to_qecc_rate}
    R \coloneq 1+\frac{1}{r^{(2)}}-\frac{1}{r^{(1)}r^{(2)}} > 0.
\end{equation}
Let $\mathcal{Q}_0$ be a quantum error-correcting code with $n_0$ message qubits and rate $R$, so that Pauli errors occurring on at most a $\delta_0$ fraction of the qubits can be corrected. $\mathcal{Q}_0$ is imagined to be some constant-sized code which may be constructed, for example, via a randomised procedure~\cite{calderbank1996good}. For $k \geq 1$, let $\mathcal{R}_k^{(1)}$ and $\mathcal{R}_k^{(2)}$ be families of quantum error-reduction codes such that $\mathcal{R}_k^{(j)}$ has rate $r^{(j)}$, error reduction $\epsilon^{(j)}$, and reducible distance $\delta^{(j)}$ for $j = 1,2$. In addition, let $\mathcal{R}_k^{(1)}$ have
\begin{equation}
    n_k \coloneq n_0\left(\frac{r^{(1)}}{1-r^{(1)}}\right)^k
\end{equation}
message qubits, and let $\mathcal{R}_k^{(2)}$ have $n_k\left(\frac{1-r^{(1)}}{r^{(1)}}\right)\cdot\frac{1}{R} = \frac{n_{k-1}}{R}$ message qubits.\footnote{Note that the condition $R>0$ enforces that $r^{(1)} > 1/2$, and thus that $n_k \to \infty$ as $k \to \infty$.}

For $k \geq 1$, $\mathcal{Q}_k$ will be a quantum error-correcting code with $n_k$ message qubits. Describing its encoding circuit is enough to describe the code itself, which we do as follows. The $n_k$ message qubits of $\mathcal{Q}_k$ are collectively called $M_k$, and are first encoded into the code $\mathcal{R}_k^{(1)}$ to produce $n_k\left(\frac{1-r^{(1)}}{r^{(1)}}\right) = n_{k-1}$ check qubits; we call this set of check qubits $A_k$. Then, the qubits in $A_k$ are encoded into the quantum error-correcting code $\mathcal{Q}_{k-1}$, to produce a further set of check qubits called $B_k$ of size $n_{k-1}\left(\frac{1-R}{R}\right)$. Finally, the qubits in $A_k \cup B_k$ are together encoded into the code $\mathcal{R}_k^{(2)}$ to produce $\frac{n_{k-1}}{R}\left(\frac{1-r^{(2)}}{r^{(2)}}\right)$ further check qubits: a set which we call $C_k$.

The following may be readily verified from the construction.

\begin{proposition}\label{prop:ecc_code_rate}
    For all $k \geq 0$, $\mathcal{Q}_k$ is a code with block length $\frac{n_k}{R}$ and rate $R$.
\end{proposition}

\subsubsection{Sequential Algorithms}

Forming sequential error correction algorithms for $\mathcal{Q}_k$ from the sequential error-reduction algorithms of $\mathcal{R}_k^{(1)}$ and $\mathcal{R}_k^{(2)}$ is a straightforward quantisation of the corresponding argument in Spielman~\cite{spielman1995linear}, as long as one is careful about the ordering of classical and quantum operations, as we now show.

\begin{lemma}\label{lem:seq_erc_to_ecc}
    As long as
    \begin{equation}
        \epsilon^{(2)} \leq \frac{1-r^{(1)}}{r^{(1)}},
    \end{equation}
    $\mathcal{Q}_k$ can correct Pauli errors on a $\Delta$-fraction of the qubits, where
    \begin{equation}
        \Delta = \min\left(\delta^{(1)}\frac{R}{r^{(1)}}, \;\delta^{(2)}(1-R), \;\delta_0\right).
    \end{equation}
    Suppose that the codes $\mathcal{R}_k^{(1)}$ and $\mathcal{R}_k^{(2)}$ are each families with linear-time quantum encoding circuits, linear-time quantum unencoding circuits, and linear-time classical error-reduction algorithms. Then, $\mathcal{Q}_k$ can be encoded and correctly decoded from a $\Delta$-fraction of Pauli errors in linear time.
\end{lemma}
\begin{proof}
   Let us show that $\mathcal{Q}_k$ can be encoded using a linear number of (quantum) gates. Suppose that $\mathcal{R}_k^{(1)}$ can be encoded using $c^{(1)}n_k$ gates and $\mathcal{R}_k^{(2)}$ can be encoded using $c^{(2)}\frac{n_{k-1}}{R}$ gates, for constants $c^{(1)}$ and $c^{(2)}$. Let $T_k$ denote the number of gates required to encode $\mathcal{Q}_k$ for $k \geq 0$. Then, for $k \geq 1$,
    \begin{equation}
        T_k = c^{(1)}n_k + c^{(2)}\frac{n_{k-1}}{R}+T_{k-1},
    \end{equation}
    or,
    \begin{equation}
        T_k-T_{k-1} = \left(\frac{r^{(1)}}{1-r^{(1)}}\right)^{k-1}\left[c^{(1)}m_0\left(\frac{r^{(1)}}{1-r^{(1)}}\right)+\frac{c^{(2)}m_0}{R}\right].
    \end{equation}
    Solving the recurrence relation yields $T_k = \Theta(n_k)$, as required. Note that the quantum unencoding circuit uses the same number of quantum gates as the quantum encoding circuit.

    Next, let us suppose that a $\Delta$-fraction of the qubits of $\mathcal{Q}_k$ have been affected by Pauli errors. We will show by induction on $k$ that these errors can be corrected. The base case is trivial since $\Delta \leq \delta_0$.
    
    For the inductive step, we begin by using the error-reduction algorithm of $\mathcal{R}_k^{(2)}$ to reduce the errors on $A_k \cup B_k$ using the checks in $C_k$. Specifically, the code $\mathcal{R}_k^{(2)}$ is quantumly unencoded, its checks (qubits in $C_k$) are measured, we perform its classical error-reduction algorithm, and a Pauli correction is applied on $A_k \cup B_k$. Initially, there are at most $\frac{n_k}{R}\Delta$ errors in total, and so in particular at most this many on either $A_k \cup B_k$ or $C_k$. We have $\Delta \leq \delta^{(2)}(1-R)$, and so
    \begin{equation}
        \frac{n_k}{R}\Delta \leq \delta^{(2)}n_k\left(\frac{1-R}{R}\right),
    \end{equation}
    and since $|A_k \cup B_k \cup C_k| = n_k\left(\frac{1-R}{R}\right)$, we can perform error reduction to reduce the size of the Pauli errors in $A_k \cup B_k$ to at most $\epsilon^{(2)}\frac{n_k}{R}\Delta$.

    Now, the assumption
    \begin{equation}
        \epsilon^{(2)} \leq \frac{1-r^{(1)}}{r^{(1)}}
    \end{equation}
    implies that the number of errors on $A_k \cup B_k$ is at most
    \begin{equation}
        \epsilon^{(2)}\frac{n_k}{R}\Delta \leq \Delta \frac{n_{k-1}}{R},
    \end{equation}
    and so using the inductive hypothesis these errors can be corrected perfectly. Indeed, the quantum code $\mathcal{Q}_k$ is unencoded, its checks (qubits in $B_k$) measured, and the error on $A_k$ is removed by a Pauli correction. At this point, there are no errors in $A_k$, and at most
    \begin{equation}
        \frac{n_k}{R}\Delta \leq \delta^{(1)}\frac{n_k}{r^{(1)}}
    \end{equation}
    errors on $M_k$. Since the block length of $\mathcal{R}_k^{(1)}$ is $\frac{n_k}{r^{(1)}}$, we can unencode $\mathcal{R}_k^{(1)}$, and use its error-reduction algorithm to perfectly correct the errors on $M_k$. The complexity of the decoding algorithm can be argued in the same way as the encoding algorithm.
\end{proof}
\begin{remark}
    While the local-to-global conversion of error reduction to error correction in this sequential case works in fundamentally the same way as in Spielman's case~\cite{spielman1995linear}, it is interesting to note that this decoding algorithm proceeds via alternating layers of quantum and classical computation.
\end{remark}

\subsubsection{Parallel Algorithms}

\begin{lemma}\label{lem:parallel_erc_to_ecc}
    Suppose that, for $j \in \{1,2\}$, $\mathcal{R}_k^{(j)}$ is a family of codes with linear-time quantum encoding and unencoding circuits of constant depth. Moreover, suppose that it has a linear-time and constant-depth classical error-reduction algorithm that, on input a stabiliser measurement resulting from Pauli errors on $v$ message qubits and $t$ check qubits, where $v,t \leq \delta^{(j)} N$, produces a Pauli correction resulting in a Pauli error on the message qubits of weight at most $\max(\epsilon^{(j)} \cdot v, \epsilon^{(j)} \cdot t)$. If
    \begin{equation}
        \epsilon^{(1)}, \epsilon^{(2)} \leq \frac{1-r^{(1)}}{r^{(1)}},
    \end{equation}
    then $\mathcal{Q}_k$ can be encoded and unencoded using quantum circuits of logarithmic depth, and using a linear number of quantum gates. Moreover, there is a classical decoding algorithm running in logarithmic depth, using a number of classical gates $O(n_k \log n_k)$, which can correct from a $\Delta$-fraction of Pauli errors, where
    \begin{equation}
        \Delta = \min\left(\delta^{(1)}\frac{R}{r^{(1)}}, \delta^{(2)}(1-R), \delta_0\right).
    \end{equation}
\end{lemma}
\begin{remark}
    As in Spielman~\cite{spielman1995linear}, the error-reduction algorithms for the parallel case reduce the error to a size $\max\left(\epsilon^{(j)}\cdot v, \epsilon^{(j)}\cdot t\right)$, making it weaker than the sequential algorithms, which lead to an error of weight at most $\epsilon^{(j)}\cdot t$.
\end{remark}
\begin{proof}[Proof of Lemma~\ref{lem:parallel_erc_to_ecc}]
    The depth of the quantum encoding and unencoding circuits follows by the construction of $\mathcal{Q}_k$, and the total size of the circuits follows by the same considerations as those in the proof of Lemma~\ref{lem:seq_erc_to_ecc}. The decoding algorithm must be treated a little more carefully.

    The first part of the algorithm begins in the same way as in Lemma~\ref{lem:seq_erc_to_ecc}. That is, the quantum unencoding is performed on the code $\mathcal{R}_k^{(2)}$, and the stabilisers measured, corresponding to measuring the qubits in $C_k$. A Pauli correction is then calculated and performed for the qubits in $A_k \cup B_k$. Moving to the next level, we have $A_k = M_{k-1}$ and $B_k = A_{k-1} \cup B_{k-1} \cup C_{k-1}$. The quantum unencoding of the code $\mathcal{R}_{k-1}^{(2)}$ then proceeds, followed by the measuring of its stabilisers and Pauli correction, and so on. Eventually, we reach the code $\mathcal{Q}_0$, whose code block is the qubits $A_1 \cup B_1$, which may then be corrected by brute force in constant time. At this point, the only qubits remaining in the code state are those of $M_i$ for $i = 0, \ldots, k$ where, recall, $M_{i-1}$ are the check qubits resulting from the encoding of the message qubits $M_i$ into the code $\mathcal{R}_i^{(1)}$ for each $i = 1, \ldots, k$. Moreover, using the same reasoning as that in the proof of Lemma~\ref{lem:seq_erc_to_ecc}, there are at this stage at most $\frac{n_i}{R}\Delta$ Pauli errors in the block $M_i$ for $i = 1, \ldots, k$, and there are no Pauli errors in the block $M_0$. Moreover, up to this point, the algorithm has run in linear time and logarithmic depth, both in classical and quantum gates.

    As in Spielman's proof of the parallel error correction algorithm, we must be a little careful at this stage, because a naive application of our error-reduction algorithms would lead to a $O(\log^2 n_k)$ depth, rather than $O(\log n_k)$. Spielman's solution to this is to perform error reduction on every $M_i$ for $i = 1, \ldots, k$ in parallel to give the desired error correction. We will be able to do this, as long as we are careful about how the code is unencoded. In particular, we will perform rounds of error reduction as we unencode the $\mathcal{R}_i^{(1)}$, before performing the parallelised error reduction as Spielman does.

    We first explain the next step at a high level before going into detail. We unencode the code $\mathcal{R}_1^{(1)}$, and measure its stabilisers (the qubits in $M_0$). We can use this sydrome to perform error reduction on the Pauli error in $M_1$. Next, we unencode the code $\mathcal{R}_2^{(1)}$, and measure its stabilisers (the qubits in $M_1$), using this syndrome to reduce the Pauli error in $M_2$. We continue until we unencode the code $\mathcal{R}_k^{(1)}$ and measure out its stabiliser (corresponding to the qubits in $M_{k-1}$). These steps can be performed in linear time and logarithmic depth, both in quantum and classical gates, and at this stage we are left with only message qubits in $M_k$, and the qubits in $M_i$ for $i = 0, \ldots, k-1$, have turned into classical syndromes.
    
    We now go into details. We begin this step with qubits in $M_0, \ldots, M_k$ still encoded in the code, where $M_{i-1}$ form the check qubits for message qubits in $M_i$ according to the code $\mathcal{R}_i^{(1)}$ for $i = 1, \ldots, k$. We emphasise that the qubits in $M_i$ are both the message qubits for the code $\mathcal{R}_i^{(1)}$ and the check qubits for the code $\mathcal{R}_{i+1}^{(1)}$, for each $i = 1, \ldots, k-1$. The qubits in $M_k$ are only the message qubits for the code $\mathcal{R}_k^{(1)}$, whereas the qubits in $M_0$ are only the check qubits for the code $\mathcal{R}_1^{(1)}$. We know that, at this stage, the weight of the Pauli errors in $M_i$ is at most $\frac{n_i}{R}\Delta$ for $i = 1, \ldots, k$, and that there are no Pauli errors in $M_0$. We denote the Pauli errors in $M_i$ as $X_iZ_i$ for $i = 1, \ldots, k$, where $X_i$ and $Z_i$ may be thought of as bit strings for the corresponding $X$- and $Z$-type errors. For $i = 1, \ldots, k-1$, we will also consider those qubits in $M_i$ that are $X$-check qubits for the code $\mathcal{R}_{i+1}^{(1)}$, and those that are $Z$-check qubits for the code $\mathcal{R}_{i+1}^{(1)}$. For $i = 1, \ldots, k-1$, we denote the Pauli errors on the $X$-check qubits in $M_i$ as $X_i^{(x)}Z_i^{(x)}$, where $X_i^{(x)}$ and $Z_i^{(x)}$ may be thought of as bit strings for the corresponding pure $X$ and pure $Z$-type errors. Similarly, we denote the Pauli errors on the $Z$-check qubits in $M_i$ as $X_i^{(z)}Z_i^{(z)}$. One may note that the bit string $X_i$ is the concatenation of the bit strings $X_i^{(x)}$ and $X_i^{(z)}$, while $Z_i$ is the concatenation of $Z_i^{(x)}$ and $Z_i^{(z)}$, for $i = 1, \ldots, k-1$. For our final piece of notation for this stage, we denote the parity-check matrices of the code $\mathcal{R}_i^{(1)}$ as
    \begin{align}
        H_{X,i} &= \begin{pmatrix}
            H_{xx,i} & H_{xq,i} & H_{xz,i}
            \end{pmatrix}\\
        H_{Z,i} &= \begin{pmatrix}
            H_{zx,i} & H_{zq,i} & H_{zz,i},
        \end{pmatrix}
    \end{align}
    which denote the supports of the $X$-type and $Z$-type stabilisers on the code's $X$-check qubits, message qubits, and $Z$-check qubits, respectively.

    Now, when we unencode the code $\mathcal{R}_1^{(1)}$ and measure its stabilisers (the qubits in $M_0$), we obtain syndromes
    \begin{align}
        \sigma_{x,1} &= H_{xq,1}Z_1\\
        \sigma_{z,1} &= H_{zq,1}X_1
    \end{align}
    since there are no errors in $M_0$. Given the size of the errors, we may perform error reduction and reduce the size of $Z_1$ and $X_1$ to at most $\epsilon^{(1)}\frac{n_1}{R}\Delta \leq \frac{n_0}{R}\Delta$.

    Next, when we unencode the code $\mathcal{R}_2^{(1)}$ and measure its stabilisers (the qubits in $M_1$), we obtain syndromes
    \begin{align}
        \sigma_{x,2} &= H_{xq,2}Z_2 + H_{xx,2}Z_1^{(x)} + H_{xz,2}Z_1^{(z)}\\
        \sigma_{z,2} &= H_{zq,2}X_2 + H_{zx,2}X_1^{(x)} + H_{zz,2}X_1^{(z)}.
    \end{align}
    Since the size of the error $Z_1  X_1$ (the number of non-trivial errors in the Pauli, or equivalently the Hamming weight of the bit-wise union of the bit strings) is now at most $\frac{n_0}{R}\Delta$, and the size of $X_2Z_2$ is at most $\frac{n_2}{R}\Delta$, we may perform error reduction, and after a correction, the size of the Pauli error on $M_2$, $X_2Z_2$, is reduced to at most $\epsilon^{(1)}\frac{n_2}{R}\Delta \leq \frac{n_1}{R}\Delta$. We continue in this way, unencoding the code $\mathcal{R}_i^{(1)}$, performing error reduction and a Pauli correction on the qubits in $M_i$ for every $i = 1, \ldots, k$. Inductively, we find that at the end, all of the qubits in $M_i$, for $i = 0, \ldots, k-1$, have been measured out, leaving only the qubits in $M_k$ unmeasured. There is a Pauli error $X_kZ_k$ remaining on the qubits in $M_k$, and we also hold syndromes as classical bit strings
    \begin{align}\label{eq:X_syndrome_i}
        \sigma_{x,i} &= H_{xq,i}Z_i + H_{xx,i} Z_{i-1}^{(x)} + H_{xz,i}Z_{i-1}^{(z)}\\
        \sigma_{z,i} &= H_{zq,i}X_i + H_{zx,i} X_{i-1}^{(x)} + H_{zz,i}X_{i-1}^{(z)},\label{eq:Z_syndrome_i}
    \end{align}
    for $i = 1, \ldots, k$. Since the error reduction as been performed, we find that the weight of $X_iZ_i$ (which is a Pauli error for $i = k$, or a union of bit strings for $i = 1, \ldots, k-1$) is at most $\frac{n_{i-1}}{R}\Delta$. At a high level, what we have done is to unencode each of the codes $\mathcal{R}_i^{(1)}$ in turn, for $i = 1, \ldots, k$, but we have interspersed the quantum unencoding with the error reduction. This is a necessity in the quantum case because, without doing the error reduction in between the quantum unencodings, errors in $M_i$ would spread and grow into sets $M_j$ for $j > i$. Up until now, the circuit has used a linear number of quantum and classical gates, and has run in a logarithmic depth.

    With the classical syndromes in hand, we may now proceed as Spielman does, reducing the errors in $M_i$ for $i = 1, \ldots, k$ \textit{simultaneously}, for a logarithmic number of rounds.\footnote{Note that this might seem strange, because the use of the error-reduction algorithms here is not proceeded by a quantum measurement, but is simply input some classical bit strings $\sigma_{x,i}$ and $\sigma_{z,i}$ of the form of Equation~\eqref{eq:X_syndrome_i} and~\eqref{eq:Z_syndrome_i}. However, the error-reduction algorithm is exactly a classical algorithm taking input of this type and, given guarantees on the size of the bit strings $Z_k, Z_{i-1}^{(x)}, Z_{i-1}^{(x)}, X_i, X_{i-1}^{(x)}$ and $X_{i-1}^{(z)}$, outputs approximations to $Z_i$ and $X_i$ that are guaranteed to reduce their size.} For every $j$ and $i$ satisfying $1 \leq j \leq i \leq k$, after $j$ \textit{total} rounds of error reduction, including the initial one, the size of the error in $M_i$ is at most $\frac{n_{i-j}}{R}\Delta$, so that there are no errors remaining after a logarithmic number of rounds of this simultaneous error reduction. There are no further quantum gates used in this process, barring Pauli correction on $M_k$ (of which there are a linear number, and they may be executed in constant depth). Each round of parallelised error reduction uses a linear number of classical gates, and the logarithmic number of rounds leads to a total classical circuit size $O(n_k\log n_k)$.  
\end{proof}

\section{Quantum Error-Reduction Codes from Lossless Z-Graphs}\label{sec:quantum_error_reduction_codes}

In this section, we construct quantum error-reduction codes and their algorithms, reducing the problem to the construction of a certain structure that we call a lossless Z-graph. In turn, we construct lossless Z-graphs both randomly and explicitly in Section~\ref{sec:lossless_Z_graphs}. 

In Section~\ref{sec:construct_quantum_error_reduction_codes}, we present our construction of quantum error-reduction codes, and explain how to encode (and unencode) them with circuits of low depth and number of gates. In Section~\ref{sec:seq_error_reduction_algorithms}, we describe sequential error-reduction algorithms for our codes, proving their complexity and that they work. Finally, in Section~\ref{sec:parallel_erc_algos}, we do the same for our parallel error-reduction algorithms.

\subsection{Construction of Quantum Error-Reduction Codes and their (Un)encoding}\label{sec:construct_quantum_error_reduction_codes}

Our quantum error-reduction codes are quantum CSS codes with parity-check matrices of the form
\begin{alignat}{4}
    H_X &= (\,I&&|A&&|C&&)\label{eq:error_reduction_X_parity}\\
    H_Z &= (D&&|B&&|\,I&&)\label{eq:error_reduction_Z_parity}.
\end{alignat}
As always, $H_X$ and $H_Z$ are binary matrices whose rows denote the support of $X$ checks and $Z$ checks, respectively, and whose columns correspond to the physical qubits of the code. We let $H_X$ and $H_Z$ both have $m$ rows (so the number of $X$ checks always equals the number of $Z$ checks, $m_X = m_Z = m$). The physical qubits are grouped into $m$ $X$-check qubits, $n$ message qubits, and $m$ $Z$-check qubits, respectively. $I$ denotes the $m \times m$ identity matrix. Matrices $A$ and $C$ correspond to the support of the $X$ stabilisers on message qubits and $Z$-check qubits, respectively, whereas matrices $D$ and $B$ correspond to the support of $Z$ stabilisers on the $X$-check qubits and message qubits, respectively.

As always, these parity-check matrices define a valid CSS code if and only if $H_X\cdot H_Z^T = 0$, with arithmetic performed over $\mathbb{F}_2$. One can check that this is the case if and only if
\begin{equation}
    C = AB^T+D^T.
\end{equation}
Our strategy in building quantum error-reduction codes will be to specify matrices $A,B$ and $D$, and then to simply let $C$ be according to this equation.

\subsubsection{Encoding Circuit}

Let us consider the encoding of a quantum CSS code of this form. It turns out that any quantum code of this form can be encoded in constant depth and linear time if the matrices $A, B$ and $D$ are chosen to be sparse.

A circuit that encodes the code is described as follows.
\begin{enumerate}
    \item Begin with $m$ $X$-check qubits in the state $\ket{+}$, $m$ $Z$-check qubits in the state $\ket{0}$, and $n$ message qubits in the logical state that we wish to encode into the code;\label{step:encoding_init}
    \item For every $(i,j) \in [m] \times [n]$ such that $A_{i,j} = 1$, perform a $\CNOT$ with the $i$'th $X$-check qubit as control and the $j$'th message qubit as target;\label{step:encoding_A_matrix}
    \item For every $(i,j) \in [m] \times [n]$ such that $B_{i,j} = 1$, perform a $\CNOT$ with the $j$'th message qubit as control and the $i$'th $Z$-check qubit as target;\label{step:encoding_B_matrix}
    \item For every $(i,j) \in [m] \times [m]\,\,\,$ such that $D_{i,j} = 1$, perform a $\CNOT$ with the $j$'th $X$-check qubit as control and the $i$'th $Z$-check qubit as target.\label{step:encoding_D_matrix}
\end{enumerate}
\begin{figure}
    \centering
    \begin{quantikz}[row sep=0.4cm, column sep=0.5cm]
\lstick{X-Check Qubits $\ket{+}^{\otimes m}$}
  & \ctrl{1}
  & \qw
  & \ctrl{2} 
  & \qw \\
\lstick{Message Qubits\hspace*{0.5cm} $\ket{\psi}$}
  & \gate[wires=1]{\text{CNOTs } A}
  & \ctrl{1}
  & \qw 
  & \qw \\
\lstick{Z-Check Qubits $\ket{0}^{\otimes m}$}
  & \qw
  & \gate[wires=1]{\text{CNOTs } B^T}
  & \gate[wires=1]{\text{CNOTs }D^T}
  & \qw 
\end{quantikz}
    \caption{Depiction of a simple encoding circuit for a quantum CSS code of the form of Equations~\eqref{eq:error_reduction_X_parity} and~\eqref{eq:error_reduction_Z_parity}. This is intended for illustrative purposes only.}
    \label{fig:encoding_circuit_diagram}
\end{figure}
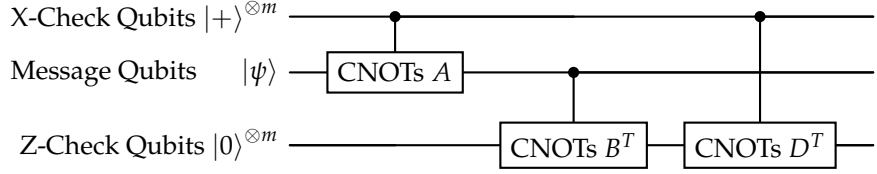
This encoding circuit is sketched in Figure~\ref{fig:encoding_circuit_diagram}. Let us show why it is an encoding circuit for the code. It is sufficient to show that, by the end of the circuit, the qubits are collectively in a $+1$-eigenstate of all (are stabilised by all) $X$ checks and $Z$ checks of the code, as specified by equations~\eqref{eq:error_reduction_X_parity} and~\eqref{eq:error_reduction_Z_parity}. To do this, one notes that, after Step~\ref{step:encoding_init}, the qubits are stabilised by every single-qubit $X$-operator acting on any $X$-check qubit, and any single-qubit $Z$-operator acting on any $Z$-check qubit. That is, after Step~\ref{step:encoding_init}, the stabilisers of the code may be denoted
\begin{alignat}{4}
    H_X &= (I&&|0&&|0&&)\\
    H_Z &= (0&&|0&&|I&&).
\end{alignat}
It is sufficient to show that these stabilisers commute through the encoding circuit to the desired form of Equations~\eqref{eq:error_reduction_X_parity} and~\eqref{eq:error_reduction_Z_parity}.

Step~\ref{step:encoding_A_matrix} causes $X$ stabilisers to spread onto the message qubits. In particular, after Step~\ref{step:encoding_A_matrix}, one can see using the commutation of $X$ operators through $\CNOT$ gates (Equation~\eqref{eq:CNOT_commutation_rules}) that the $i$-th $X$ check has support on the message qubits corresponding to the $i$-th row of $A$. That is, at this point, the qubits have stabilisers corresponding to
\begin{alignat}{4}
    H_X &= (I&&|A&&|0&&)\\
    H_Z &= (0&&|\,0&&|I&&).
\end{alignat}
Next, Step~\ref{step:encoding_B_matrix} causes $X$ stabilisers to spread from the message qubits to the $Z$-check qubits, while simultaneously causing $Z$ stabilisers to spread from the $Z$-check qubits onto the message qubits. In particular, the $i$'th $Z$ stabiliser gains support on the message qubits corresponding to the $i$'th row of $B$. Simultaneously, any single-qubit $X$ operator on the message qubits in the support of an $X$ stabiliser spreads to the $Z$-check qubits according to the columns of $B$. Thus, after this step, the qubits are stabilised by operators described as
\begin{alignat}{4}
    H_X &= (I&&|A&&|AB^T&&)\\
    H_Z &= (0&&|B&&|\,\,\,\,I&&),
\end{alignat}
where arithmetic is performed in binary because $X^2 = I$.

Finally, in Step~\ref{step:encoding_D_matrix}, $Z$ stabilisers are spread from the $Z$-check qubits to the $X$-check qubits, and $X$ stabilisers are spread from the $X$-check qubits to the $Z$-check qubits according to the rows and columns of $D$, respectively. Thus, at the end of the circuit, the qubits are stabilised by operators described as
\begin{alignat}{4}
    H_X &= (\,I&&|A&&|AB^T + D^T&&)\\
    H_Z &= (D&&|B&&|\,\,\,\,\,\,\,\,\,\,\,\,I&&),
\end{alignat}
as required, since $C = AB^T+D^T$.

The complexity of the encoding circuit is exactly the number of $1$'s in the matrices $A, B$ and $D$, which is linear for sparse $A, B$ and $D$. Furthermore, for such matrices, the $\CNOT$s can be organised into a constant number of layers, making the encoding circuit constant depth.

\subsubsection{Unencoding, Stabiliser Measurement and Error Spreading}

After the qubits of the code pass through the noise channel, they are unencoded and the stabilisers measured. This is performed by a circuit that inverts the encoding circuit, which may simply be achieved by running the $\CNOT$s of the encoding circuit in the opposite order, before measuring each of the $X$-check qubits in the $X$ basis, and each of the $Z$-check qubits in the $Z$ basis. This is graphically depicted in Figure~\ref{fig:encode_noise_unencode_measure}. As in Figure~\ref{fig:encoding_noise_unencoding_labelled_Paulis}, we assign symbols to the Paulis occurring in the noise channel, so that the syndromes, i.e., the outcomes of the stabiliser measurements, are
\begin{align}
    \sigma_x &= Z_x + A\cdot Z_q + C\cdot Z_z\\
    \sigma_z &= D\cdot X_x + B \cdot X_q + X_z,
\end{align}
where it is understood that $Z_x$ is both a pure-type $Z$-Pauli operator, and a bit string denoting the support of the Pauli operator.
\begin{figure}
    \centering
    \scalebox{0.75}{\begin{quantikz}[row sep=0.4cm, column sep=0.5cm]
\lstick{X-Check Qubits $\ket{+}^{\otimes m_X}$}
  & \ctrl{1}
  & \qw
  & \ctrl{2}
  & \gate[wires=1]{X_xZ_x}
  & \ctrl{2}
  & \qw
  & \ctrl{1}
  & \meter{X} \\
\lstick{Message Qubits\hspace*{0.5cm} $\ket{\psi}$}
  & \gate[wires=1]{\text{CNOTs } A}
  & \ctrl{1}
  & \qw
  & \gate[wires=1]{X_qZ_q}
  & \qw
  & \ctrl{1}
  & \gate[wires=1]{\text{CNOTs } A}
  & \qw \\
\lstick{Z-Check Qubits $\ket{0}^{\otimes m_Z}$}
  & \qw
  & \gate[wires=1]{\text{CNOTs } B^T}
  & \gate[wires=1]{\text{CNOTs }D^T}
  & \gate[wires=1]{X_zZ_z}
  & \gate[wires=1]{\text{CNOTs }D^T}
  & \gate[wires=1]{\text{CNOTs } B^T}
  & \qw
  & \meter{Z}
\end{quantikz}}
    \caption{}
    \label{fig:encode_noise_unencode_measure}
\end{figure}
At this point, we must address an issue that is a distinctly non-classical consideration: the issue of errors spreading during the unencoding circuit. This is an issue that results from the possibility of $X$ errors occurring on the $X$-check qubits, and $Z$ errors occurring on the $Z$-check qubits. The problem is that, after unencoding and stabiliser measurement, the errors remaining on the message qubits are not just the errors that occurred on them, $X_qZ_q$, but errors that can spread to them in the unencoding circuit. In particular, one may check that the residual $X$ error on the message qubits after this circuit is
\begin{equation}
    X_{Res} \coloneq A^T\cdot X_x + X_q,
\end{equation}
rather than simply $X_q$, and the residual $Z$ error on the message qubits after this circuit is
\begin{equation}
    Z_{Res} \coloneq Z_q + B^T\cdot Z_z.
\end{equation}
These are the errors that our error-reduction algorithms must ultimately find reduced versions of, in order to reduce the weight of the Pauli errors on the message qubits.

\subsubsection{Constructing the Parity-Check Matrices from Lossless Z-Graphs}

One of Spielman's constructions of error-reduction codes~\cite{spielman1995linear} was based on lossless expanders (obtained via a random construction). Motivated by the need to construct quantum error-reduction codes which can reduce both $X$-type and $Z$-type errors, while handling the problem of error spreading, we make the following definition of a particular type of bipartite graph.
\begin{definition}\label{def:lossless_Z-graphs}[Lossless Z-Graphs]
    An $(n, m, \Delta_1, \Delta_2, \eta_1, \eta_2, \epsilon_1, \epsilon_2)$ lossless-Z graph is a bipartite graph with the following properties. The left vertices are partitioned into two sets called $L_1$ and $L_2$, which have size $n$ and $m$, respectively. Similarly, the right vertices are partitioned into sets of size $R_1$ and $R_2$, which have sizes $n$ and $m$, respectively. The subgraphs obtained by restriction to $(L_1, R_2)$ and $(R_1, L_2)$ are $(\Delta_1, \Delta_1')$-biregular graphs, where $\Delta_1' = \frac{n}{m}\Delta_1$, and the subgraph obtained by restriction to $(L_2, R_2)$ is a $\Delta_2$-regular (bipartite) graph. There are no edges between vertices in $L_1$ and $R_1$.

    Given subsets $S_1 \subseteq L_1$ and $S_2 \subseteq L_2$ of size $|S_1| \leq \eta_1 n$ and $|S_2| \leq \eta_2 m$, $S_1 \cup S_2$ has a large number of neighbours in $R_2$, in particular,
    \begin{equation}
        |N_{R_2}(S_1\cup S_2)| \geq (1-\epsilon_1)\Delta_1\cdot |S_1| + (1-\epsilon_2)\Delta_2\cdot |S_2|. 
    \end{equation}
    Similarly, given subsets $S_1 \subseteq R_1$ and $S_2 \subseteq R_2$ of size $|S_1| \leq \eta_1 n$ and $|S_2| \leq \eta_2m$, $S_1 \cup S_2$ has a large number of neighbours in $L_2$, namely,
    \begin{equation}
        |N_{L_2}(S_1\cup S_2)| \geq (1-\epsilon_1)\Delta_1\cdot |S_1| + (1-\epsilon_2)\Delta_2\cdot |S_2|.
    \end{equation}
\end{definition}
Given a lossless $Z$-graph $G$, we construct the \textit{quantum error-reduction code corresponding to $G$} as follows. The code has the usual parity-check matrices we take for a quantum error-reduction code,
\begin{alignat}{4}
    H_X &= (\,I&&|A&&|C&&)\\
    H_Z &= (D&&|B&&|\,I&&),
\end{alignat}
where $C = AB^T+D^T$, and where
\begin{itemize}
    \item $A$ is a binary $m \times n$ matrix, and is taken as the adjacency matrix of the subgraph of $G$ obtained by restriction to $(L_1, R_2)$;
    \item $B$ is a binary $m \times n$ matrix, and is taken as the adjacency matrix of the subgraph of $G$ obtained by restriction to $(R_1, L_2)$;
    \item $D$ is a binary $m \times m$ matrix, and is taken as the adjacency matrix of the subgraph of $G$ obtained by restriction to $(R_2, L_2)$.
\end{itemize}
For the avoidance of doubt, the rows and columns of $D$ correspond to the nodes in $R_2$ and $L_2$, respectively. A row in $D$ corresponding to a node $V \in R_2$ has support corresponding to the neighbours of $V$ in $L_2$. Notice that we could have equivalently defined $D^T$ as the adjacency matrix of the subgraph of $G$ obtained by restriction to $(L_2, R_2)$.

\subsection{Sequential Error-Reduction Algorithms}\label{sec:seq_error_reduction_algorithms}

Given a quantum error-reduction code constructed from a lossless Z-graph as above, we now describe how to perform error reduction on the $X$ errors in Section~\ref{sec:sequential_X_reduction_algos}, and then the $Z$ errors in Section~\ref{sec:sequential_Z_reduction_algos}.

\subsubsection{Reducing the X Errors}\label{sec:sequential_X_reduction_algos}

We now give some intuition before presenting the construction formally.

We begin by noting that the $X$ errors are relatively easy to handle. Indeed, the input to the problem of reducing the $X$ errors is the $Z$-check syndrome
\begin{equation}
    \sigma_z = D\cdot X_x + B\cdot X_q + X_z,
\end{equation}
and we aim to reduce the residual $X$ error on the message qubits,
\begin{equation}
    X_{Res} \coloneq A^T\cdot X_x+X_q.
\end{equation}
In particular, we aim to flip bits in $X_{Res}$ so that its Hamming weight $|X_{Res}|$ becomes small as a function of $X_z$.

At this point, we can notice that this just has the form of a classical error reduction problem. We have full control over our choices of $D, B$ and $A$, which are all the matrices appearing in the problem. The syndrome $\sigma_z$ takes the form of a checked message bit error, $D \cdot X_x + B \cdot X_q$ added to a check bit error $X_z$. We can therefore get good approximations of the errors $X_x$ and $X_q$ (as a function of $X_z$) by simply taking
\begin{equation}
    H_Z = (D|B|I)
\end{equation}
to be the parity-check matrix for a classical error-reduction code. Indeed, when building our parity-check matrices for the quantum error-reduction code, we have taken the graph corresponding to $(D|B)$ to be a lossless expander, just with the added detail that nodes on one side of the graph can have two different degrees. It will be possible to perform error reduction of the $X$ errors using this, in a similar manner to Spielman's use of a (one-sided) lossless expander to construct a classical error-reduction code in~\cite{spielman1995linear}.

It will be important, however, to take care of the issue that the degrees on one side of the graph are mixed. The reason for this will be as follows. Doing the initial error reduction will give us something like $|X_x| \lesssim \frac{|X_z|}{\Delta_2}$ and $|X_q| \lesssim \frac{|X_z|}{\Delta_1}$. Reducing $X_x$ and $X_z$ then allows us to reduce $X_{Res}$, but the presence of $A^T$ in the expression for $X_{Res}$, which results from the spreading of quantum errors in the unencoding circuit, will then give us (something like) $|X_{Res}| \lesssim \Delta_1'\frac{|X_z|}{\Delta_2} + \frac{|X_z|}{\Delta_1}$, where we recall that $\Delta_1' = \frac{n}{m}\Delta_1$. This will be okay, as long as we eventually choose $\Delta_2 \gg \Delta_1' \sim \Delta_1$. To be clear, both $\Delta_1$ and $\Delta_2$ are constants as the length of the code grows, but $\Delta_2$ will be chosen to be a much larger constant than $\Delta_1$.

We now describe the reduction of the $X$ errors more formally. Our sequential error-reduction algorithm for $X$ errors is given in Algorithm~\ref{alg:X_sequential_reduction}. In words, we imagine the errors $X_x$, $X_q$ and $X_z$ as subsets of the vertices $R_2, R_1$ and $L_2$, respectively, connected by the graph $G$. Because we are tackling what is fundamentally a classical error reduction problem, we refer to $X_x$ and $X_q$ as message bits and $X_z$ as check bits for the purpose of discussing this algorithm.\footnote{This is not to be confused with the language used more broadly in the paper, where the qubits on which $X_q$, $X_x$ and $X_z$ are supported are called, respectively, message qubits, $X$-check qubits, and $Z$-check qubits.} We also imagine the syndrome $\sigma_z$ as a subset of the vertices $L_2$, and refer to these as the syndrome bits. The algorithm then proceeds by flipping message bits (bits in $X_x$ or $X_q$) sequentially, where at each time step we flip a bit if it is in contact with more violated syndrome bits than satisfied syndrome bits, so that the number of violated syndrome bits reduces at each time step. We continue until there are no violated syndrome bits. This allows us to construct approximations $\tilde{X}_x$ and $\tilde{X}_q$ to the bit strings $X_x$ and $X_q$. We then finish by computing the approximation $\tilde{X}_{Res} = A^T\cdot \tilde{X}_x + \tilde{X}_q$ to $X_{Res}$. We now analyse the runtime and performance of the algorithm.

\begin{algorithm}[t]
\caption{Sequential Error-Reduction Algorithm for $X$ Errors}\label{alg:X_sequential_reduction}
\begin{algorithmic}[1] 

\Require Lossless $Z$-Graph $G$; \hspace*{0.1cm} Syndrome $\sigma_z = D\cdot X_x+B\cdot X_q + X_z$
\Ensure Approximation to the residual $X$ error, $\tilde{X}_{Res}$
\State $\tilde{X}_x \gets 0$
\State $\tilde{X}_q \gets 0$
\State $\tilde{\sigma}_z \gets 0$
\While{$\tilde{\sigma}_z \neq \sigma_z$}

    \State $\hat{\sigma}_z \gets \sigma_z + \tilde{\sigma}_z$
    \State If a bit in $\tilde{X}_x$ (associated to the vertices $R_2$) or in $\tilde{X}_q$ (associated to the vertices $R_1$) is in contact (via the graph $G$) with more $1$'s than $0$'s in $\hat{\sigma}_z$ (associated to the vertices $L_2$), then flip it.
    \State $\tilde{\sigma}_z \gets D\cdot \tilde{X}_x + B\cdot \tilde{X}_q$
\EndWhile
\State $\tilde{X}_{Res} \gets A^T\cdot \tilde{X}_x + \tilde{X}_q$

\end{algorithmic}
\end{algorithm}

\begin{proposition}
    The sequential error-reduction algorithm for $X$ errors terminates, and uses a linear number of classical gates.
\end{proposition}
\begin{proof}
    Since the Hamming weight of the syndrome decreases at each step, the algorithm must terminate, and there must be at most a linear number of bit flips. In addition, the sparsity of the graph $G$ implies that each bit flip requires a constant amount of computation (see Lemma 9 of~\cite{sipser1996expander} for details).
\end{proof}
\begin{lemma}\label{lem:X_seq_reduction}
    Suppose that we have a quantum error-reduction code built from an $(n,m,\Delta_1, \Delta_2, \eta_1, \eta_2, \epsilon_1, \epsilon_2)$ lossless $Z$-graph, where $\epsilon = \epsilon_1 = \epsilon_2 \leq \frac{1}{8}$. Suppose that the $X$ errors are of weight
    \begin{align}
        |X_q| &< \alpha n\\
        |X_x| &< \beta m\\
        |X_z| &< \gamma m.
    \end{align}
    Then, there exist small enough constants $\alpha, \beta, \gamma$ such that the sequential error-reduction algorithm for $X$ errors produces an approximation to the residual $X$ error, $\tilde{X}_{Res}$, such that
    \begin{equation}
        |X_{Res} + \tilde{X}_{Res}| \leq \frac{n}{m}\Delta_1\frac{8|X_z|}{\Delta_2}+\frac{8|X_z|}{\Delta_1}.
    \end{equation}
\end{lemma}
\begin{proof}
    The initial set of corrupt message bits is given by the vectors $X_q$ and $X_x$, and the set of corrupt check bits is given by the vector $X_z$. Over the course of the algorithm, we flip message bits, and the set of corrupt message bits changes as we do so. In doing so, we develop approximations $\tilde{X}_q$ and $\tilde{X}_x$ to the true sets of corrupt message bits, $X_q$ and $X_x$. At a given step of the algorithm, we let the sets of remaining corrupt message bits correspond to vectors $\hat{X}_x \coloneq X_x +  \tilde{X}_x$ and $\hat{X}_q \coloneq X_q + \tilde{X}_q$. We let the set of unsatisfied checks at a given step of the algorithm be $\hat{\sigma}_z \coloneq D\cdot\hat{X}_x + B\cdot \hat{X}_q + X_z$. Finally, at a given step of the algorithm, we let $u$ be the number of unsatisfied checks, and we let $s$ be the number of satisfied checks neighbouring a corrupt message bit.

    We may always choose $\alpha \leq \eta_1$ and $\beta \leq \eta_2$. We begin by showing that if, at some step of the algorithm,
    \begin{align}
        |\hat{X}_q| &< \eta_1n\label{eq:X_seq_step_assumption_1}\\
        |\hat{X}_x| &< \eta_2m\label{eq:X_seq_step_assumption_2}\\
        |\hat{X}_q|\left(\frac{1}{4}-\epsilon\right)\Delta_1+|\hat{X}_x|\left(\frac{1}{4}-\epsilon\right)\Delta_2 &>|X_z|.\label{eq:X_seq_step_assumption_3}
    \end{align}then there is a message bit for the algorithm to flip, i.e., the algorithm has not finished. Using the first two of these conditions, we have, using the expansion of the graph $G$, that
    \begin{equation}
        u+s \geq (1-\epsilon)\Delta_1\cdot|\hat{X}_q| + (1-\epsilon)\Delta_2\cdot|\hat{X}_x|.
    \end{equation}
    On the other hand, satisfied checks with corrupt message bit neighbours must neighbour at least two corrupt message bits, or themselves be corrupt. Thus,
    \begin{equation}
        \Delta_1|\hat{X}_q| + \Delta_2|\hat{X}_x| + |X_z| \geq u+2s.
    \end{equation}
    Together, these equations yield
    \begin{align}\label{eq:X_seq_u_bound}
        u &\geq (1-2\epsilon)\Delta_1\cdot|\hat{X}_q| + (1-2\epsilon)\Delta_2\cdot|\hat{X}_x| - |X_z|.
    \end{align}
    Combining this with Equation~\eqref{eq:X_seq_step_assumption_3} gives
    \begin{equation}
        u > \frac{\Delta_1}{2}|\hat{X}_q| + \frac{\Delta_2}{2}|\hat{X}_x| + |X_z|.
    \end{equation}
    This equations tell us that there is some message bit neighbouring more unsatisfied check bits than satisfied check bits, which the algorithm could then flip.

    Now, since the number of unsatisfied checks decreases at every step, the algorithm must terminate, and so there must come a point when there is no message bit for the algorithm to flip. This means that at some step, at least one of the assumptions of Equations~\eqref{eq:X_seq_step_assumption_1} to~\eqref{eq:X_seq_step_assumption_3} must not hold. We aim to show that it is the last of the three. For a contradiction, suppose it is the first. Since $|\hat{X}_q|$ changes in steps of at most one over the course of the algorithm, for Equation~\eqref{eq:X_seq_step_assumption_1} to be violated, it must be the case at some step that $|\hat{X}_q| = \eta_1n$. Noting that we can still apply the expansion properties of the graph, Equation~\eqref{eq:X_seq_u_bound} still holds, and we have
    \begin{align}
        u &\geq (1-2\epsilon)\Delta_1\eta_1n + (1-2\epsilon)\Delta_2\cdot|\hat{X}_x| - |X_z|\\
        &\geq (1-2\epsilon)\Delta_1\eta_1n - \gamma m.\label{eq:X_seq_contra_1}
    \end{align}
    On the other hand, initially, we have
    \begin{align}
        u &\leq \Delta_1|X_q| + \Delta_2|X_x| + |X_z|\\
        &< \Delta_1\alpha n + \Delta_2\beta m + \gamma m,\label{eq:X_seq_reduction_u_upper}
    \end{align}
    and because the number of unsatisfied checks decreases at each step, we have $u < \Delta_1\alpha n + \Delta_2\beta m + \gamma m$ at every step. We have a contradiction with Equation~\eqref{eq:X_seq_contra_1} by choosing small enough constants $\alpha, \beta, \gamma$.

    One can similarly show that the Equation~\eqref{eq:X_seq_step_assumption_2} cannot be violated, because it implies a step at which $|\hat{X}_x| = \eta_2m$, where one would have
    \begin{align}
        u &\geq (1-2\epsilon)\Delta_1\eta_1n + (1-2\epsilon)\Delta_2\cdot |\hat{X}_x| - |X_z|\\
        &\geq (1-2\epsilon)\Delta_2\eta_2m - \gamma m,
    \end{align}
    as well as Equation~\eqref{eq:X_seq_reduction_u_upper}, from which we derive a contradiction given small enough constants $\alpha, \beta, \gamma$.

    We have found that the algorithm must terminate, and when it does, the Equation~\eqref{eq:X_seq_step_assumption_3} must be violated. At this point, the remaining message corruptions have size
    \begin{align}
        |\hat{X}_q| &\leq \frac{|X_z|}{(1/4-\epsilon)\Delta_1}\label{eq:X_message_remaining_size}\\
        |\hat{X}_x| &\leq \frac{|X_z|}{(1/4-\epsilon)\Delta_2}.
    \end{align}
    The size of the error on our estimate of $X_{Res}$ is $|X_{Res} + \tilde{X}_{Res}|$, which, by a triangle inequality, is at most
    \begin{equation}
        |A^T\cdot \hat{X}_x| + |\hat{X}_q| \leq \frac{n}{m}\Delta_1\frac{|X_z|}{(1/4-\epsilon)\Delta_2}+\frac{|X_z|}{(1/4-\epsilon)\Delta_1} \leq\frac{n}{m}\Delta_1\frac{8|X_z|}{\Delta_2} + \frac{8|X_z|}{\Delta_1}.
    \end{equation}   

\end{proof}

\subsubsection{Reducing the Z Errors}\label{sec:sequential_Z_reduction_algos}

Initially, reducing the $Z$ errors seems to be more problematic than reducing the $X$ errors. The input to the $Z$ error reduction is the $X$ syndrome
\begin{equation}
    \sigma_x = Z_x + A\cdot Z_q + C\cdot Z_z,
\end{equation}
and we aim to reduce the residual $Z$ error
\begin{equation}
    Z_{Res} = Z_q + B^T\cdot Z_z,
\end{equation}
where we have
\begin{equation}
    C = AB^T+D^T.
\end{equation}
We have direct control over the matrices $A, B$ and $D$, but are obliged to set $C$ according to this equation. Thus, a direct reduction of the $Z$ errors $Z_q$ and $Z_z$ (in order to calculate a reduction of $Z_{Res}$) using the matrices $A$ and $C$ (in analogy to what we do for the $X$ errors) seems difficult. However, we are able to treat the $Z$ errors in a fundamentally different way to the $X$ errors. We re-write
\begin{align}
    \sigma_x &= Z_x + A\cdot Z_q + (AB^T+D^T)\cdot Z_z\\
    &= Z_x + A\cdot (Z_q + B^T\cdot Z_z) + D^T Z_z\\
    &= Z_x + A\cdot Z_{Res} + D^T\cdot Z_z.
\end{align}
We find that the error that we ultimately want to reduce, $Z_{Res}$, in this case sits directly under the matrix $A$ in the syndrome $\sigma_x$. This means that the reduction of the $Z$ errors is also writeable in the form of a classical error reduction problem. The difference in this case is that we will directly reduce the errors $Z_{Res}$ and $Z_z$, instead of reducing $Z_q$ and $Z_z$ before calculating a reduced $Z_{Res}$. Of course, this requires the matrices $A$ and $D^T$ to correspond to a classical error-reduction code, such as a lossless expander, but this is exactly what is given to us by the structure of the lossless $Z$-graph.

\begin{algorithm}[t]
\caption{Sequential Error-Reduction Algorithm for $Z$ Errors}\label{alg:Z_sequential_reduction}
\begin{algorithmic}[1] 

\Require Lossless $Z$-Graph $G$; \hspace*{0.1cm} Syndrome $\sigma_x = Z_x + A\cdot Z_q + C\cdot Z_z = Z_x + A\cdot Z_{Res} + D^T\cdot Z_z$
\Ensure Approximation to the residual $Z$ error, $\tilde{Z}_{Res}$
\State $\tilde{Z}_z \gets 0$
\State $\tilde{Z}_{Res} \gets 0$
\State $\tilde{\sigma}_x \gets 0$
\While{$\tilde{\sigma}_x \neq \sigma_x$}

    \State $\hat{\sigma}_x \gets \sigma_x + \tilde{\sigma}_x$
    \State If a bit in $\tilde{Z}_z$ (associated to the vertices $L_2$) or in $\tilde{Z}_{Res}$ (associated to the vertices $L_1$) is in contact (via the graph $G$) with more $1$'s than $0$'s in $\hat{\sigma}_x$ (associated to the vertices $R_2$), then flip it.
    \State $\tilde{\sigma}_x \gets A\cdot \tilde{Z}_{Res} + D^T\cdot \tilde{Z}_z$
\EndWhile

\end{algorithmic}
\end{algorithm}

In Algorithm~\ref{alg:Z_sequential_reduction}, we formally present the sequential error-reduction algorithm for the $Z$ errors, although it is the same bit-flipping procedure as Algorithm~\ref{alg:X_sequential_reduction}, albeit applied to different vectors, and without the final calculation step. The proof of runtime and correctness of the algorithm are also essentially the same as for the $X$ errors, and so we omit and abridge them as follows.
\begin{proposition}
    The sequential error-reduction algorithm for $Z$ errors terminates, and runs in linear time.
\end{proposition}
\begin{lemma}\label{lem:Z_seq_reduction}
     Suppose that we have a quantum error-reduction code built from an $(n,m,\Delta_1, \Delta_2, \eta_1, \eta_2, \epsilon_1, \epsilon_2)$ lossless $Z$-graph, where $\epsilon = \epsilon_1 = \epsilon_2 \leq\frac{1}{8}$. Suppose that the $Z$ errors are such that
    \begin{align}
        |Z_q| &< \alpha n\\
        |Z_z| &< \beta m\\
        |Z_x| &< \gamma m.
    \end{align}
    Then, there exist small enough constants $\alpha, \beta, \gamma$ such that sequential error-reduction algorithm for $Z$ errors produces an approximation to the residual $Z$ error, $\tilde{Z}_{Res}$, such that
    \begin{equation}
        |Z_{Res} + \tilde{Z}_{Res}| \leq \frac{|Z_x|}{(1/4-\epsilon)\Delta_1} \leq\frac{8|Z_x|}{\Delta_1}.
    \end{equation}
\end{lemma}
\begin{proof}
    By a triangle inequality, the initial residual $Z$ error has size $|Z_{Res}| \leq |Z_q| + \Delta_1\frac{n}{m}|Z_z|$. Therefore, by taking small enough constants $\alpha, \beta, \gamma$, the proof of Lemma~\ref{lem:X_seq_reduction} applies with the analogy explained above. Specifically, $Z_{Res}$ here plays the role of $X_q$ there, $Z_z$ here plays the role of $X_x$ there, and $Z_x$ here plays the role of $X_z$ there. When the algorithm terminates, the remaining residual $Z$ error has the claimed size; see Equation~\eqref{eq:X_message_remaining_size}.
\end{proof}

\subsection{Parallel Error-Reduction Algorithms}\label{sec:parallel_erc_algos}

We now turn to our parallel error-reduction algorithms for our quantum error-reduction codes constructed from lossless $Z$-graphs.
\subsubsection{Parallel Reduction of X Errors}

Just as with the quantum encoding circuits, the quantum unencoding circuits may be performed via a constant-depth quantum circuit with a linear number of quantum gates. Given $X$ errors as before, the $Z$-check syndrome that is measured is
\begin{equation}
    \sigma_z = D\cdot X_x + B\cdot X_q + X_z,
\end{equation}
and the residual $X$ error that we aim to reduce is
\begin{equation}
    X_{Res} = A^T\cdot X_x + X_q.
\end{equation}
Given that the graph corresponding to the matrix $(D|B)$ is a lossless expander, we may use a parallel small-set flip decoding algorithm, as Spielman does~\cite{spielman1995linear}, but here to reduce the errors $X_x$ and $X_q$, before calculating a reduction of $X_{Res}$. We will face the added complication that our lossless expander must be particularly strong in order to get a great enough reduction to overcome the problem of error spreading; in particular, we will require our lossless $Z$-graph to have $\epsilon_i = \mathcal{O}(1/\Delta_i)$. This can be obtained in the randomised construction, although it is not known how to obtain explicit two-sided expanders with this property~\cite{HLMRZ25}, and so we also do not obtain explicit lossless $Z$-graphs with this property. This is one key reason that our parallel algorithms are limited to the randomised construction. 

Our parallel error-reduction algorithm for $X$ errors is described formally in Algorithm~\ref{alg:X_parallel_reduction_algorithm}.

\begin{algorithm}[t]
\caption{Parallel Error-Reduction Algorithm for $X$ Errors}\label{alg:X_parallel_reduction_algorithm}
\begin{algorithmic}[1] 

\Require Lossless $Z$-Graph $G$; \hspace*{0.1cm} Syndrome $\sigma_z = D \cdot X_x + B \cdot X_q + X_z$; \hspace*{0.1cm}
\Ensure $\tilde{X}_{Res}$: approximation to residual $X$ error $X_{Res} = A^T\cdot X_x + X_q$
\State $\tilde{X}_x \gets 0$
\State $\tilde{X}_q \gets 0$
\State In parallel, for every bit in $\tilde{X}_x$ (associated to the vertices $R_2$) or in $\tilde{X}_q$ (associated to the vertices $R_1$), if the bit is in contact (via the graph $G$) with more $1$'s than $0$'s in $\sigma_z$ (associated to the vertices $L_2$), then flip it.
\State $\tilde{X}_{Res} \gets A^T\cdot \tilde{X}_x + \tilde{X}_q$

\end{algorithmic}
\end{algorithm}

\begin{proposition}
    The parallel error-reduction algorithm for $X$ errors may be run in constant depth and with a linear total number of gates.
\end{proposition}

\begin{lemma}\label{lem:parallel_X_reduction}
    Suppose that we have a quantum error-reduction code built from an $(n,m,\Delta_1, \Delta_2, \eta_1, \eta_2, \epsilon_1, \epsilon_2)$ lossless $Z$-graph, where $\epsilon_1, \epsilon_2 \leq \frac{1}{8}, \epsilon_1< \frac{2}{\Delta_1}$ and $\epsilon_2< \frac{2}{\Delta_2}$. Suppose that 
    \begin{align}
        |X_q| &\leq \alpha n\\
        |X_x| &\leq \beta m\\
        |X_z| &\leq \gamma m.
    \end{align}
    Then, there exist small enough constants $\alpha, \beta$ and $\gamma$ such that at the end of the parallel error-reduction algorithm for $X$ errors, we have the size of remaining errors
    \begin{align}\label{eq:X_q_parallel_reduction}
        |X_q + \tilde{X}_q| &< \frac{32}{\Delta_1}\max\left(|X_q|, |X_x| + |X_z|\right)\\
        |X_x + \tilde{X}_x| &< \frac{32}{\Delta_2}\max\left(|X_q|, |X_x| + |X_z|\right).\label{eq:X_x_parallel_reduction}
    \end{align}
    In addition, the size of the remaining residual $X$ error is
    \begin{equation}
        |X_{Res} + \tilde{X}_{Res}| < 32\left(\frac{n}{m}\frac{\Delta_1}{\Delta_2} + \frac{1}{\Delta_1}\right)\max\left(|X_q|, |X_x| + |X_z|\right).
    \end{equation}
\end{lemma}
\begin{proof}
    Throughout, we imagine errorful bits as subsets of $R_1, R_2$ and $L_2$, and the remaining syndrome as a subset of $L_2$.  We denote the sets of errorful bits as $V_1 \subseteq R_1$ (the support of $X_q$), $V_2 \subseteq R_2$ (the support of $X_x$), and $T \subseteq L_2$ (the support of $X_z$). We further denote the sets $F_1 \subseteq V_1$ and $F_2 \subseteq V_2$ as the sets of errorful bits that fail to flip in the algorithm, as well as the sets $C_1 \subseteq R_1$ and $C_2 \subseteq R_2$ as the sets of bits that are not errorful, but get flipped at the end of the algorithm. The aim then simply becomes to upper bound $|C_1 \cup F_1|$, the number of errorful bits in $R_1$ at the end of the algorithm, and $|C_2 \cup F_2|$, the number of errorful bits in $R_2$ at the end of the algorithm.

    We may take $\alpha < \eta_1$ and $\beta < \eta_2$. We will start by showing that $|V_1 \cup C_1| < \eta_1 n$ and $|V_2 \cup C_2| < \eta_2 m$, assuming that $\alpha, \beta$ and $\gamma$ are small enough constants. First, suppose that $|V_1 \cup C_1| \geq \eta_1 n$ for a contradiction. Pick a subset $C_1' \subseteq C_1$ such that $|V_1 \cup C_1'| = \eta_1 n$. We know that more than $\frac{\Delta_1}{2}$ of the edges leaving every vertex in $C_1'$ land on a corrupt check bit (a bit in $T$) or a bit in $L_2$ taking input from a bit in $V_1$ or $V_2$. This gives us
    \begin{align}
        |N_{L_2}(V_1 \cup C_1' \cup V_2)| < |N_{L_2}(V_1 \cup V_2)| + |T| + \frac{\Delta_1}{2}|C_1'| &= |N_{L_2}(V_1 \cup V_2)| + |T| + \frac{\Delta_1}{2}(\eta_1n - |V_1|)\\
        &\leq \Delta_1|V_1| + \Delta_2|V_2| + |T| + \frac{\Delta_1}{2}(\eta_1n-|V_1|).
    \end{align}
    On the other hand, using the expansion property, we have
    \begin{align}
        |N_{L_2}(V_1 \cup C_1' \cup V_2)| &\geq (1-\epsilon_1)\Delta_1|V_1 \cup C_1'| + (1-\epsilon_2)\Delta_2|V_2|\\
        &=(1-\epsilon_1)\Delta_1\eta_1n + (1-\epsilon_2)\Delta_2|V_2|.
    \end{align}
    Putting these together gives us
    \begin{equation}
        \left(\frac{1}{2}-\epsilon_1\right)\Delta_1\eta_1n < \frac{\Delta_1}{2}|V_1| + \epsilon_2\Delta_2|V_2| + |T| \leq \frac{\Delta_1}{2}\alpha n + \epsilon_2\Delta_2\beta m + \gamma m.
    \end{equation}
    One may choose small enough constants $\alpha, \beta, \gamma$ such that this equation yields a contradiction. Similar methods show that if $|V_2 \cup C_2| \geq \eta_2m$, then one can derive a contradiction given small enough $\alpha, \beta, \gamma$. Given that we now have $|V_1 \cup C_1| < \eta_1n$ and $|V_2 \cup C_2| < \eta_2m$, we have
    \begin{multline}\label{eq:VC_bound}
        (1-\epsilon_1)\Delta_1|V_1 \cup C_1| + (1-\epsilon_2)\Delta_2|V_2 \cup C_2| \leq |N_{L_2}(V_1 \cup V_2 \cup C_1 \cup C_2)| < \\|N_{L_2}(V_1 \cup V_2)| + |T| + \frac{\Delta_1}{2}|C_1| + \frac{\Delta_2}{2}|C_2|,
    \end{multline}
    where the latter inequality follows from the same considerations as above on the neighbours of $C_i$.

    We next turn to the sets $F_1$ and $F_2$. Since $|F_1| < \eta_1n$ and $|F_2| < \eta_2m$, we can apply the expansion to $F_1 \cup F_2$, and furthermore it is quick to show that the set of unique neighbours of $F_1 \cup F_2$ in $L_2$, that is, the set of nodes in $L_2$ with exactly one neighbour in $F_1 \cup F_2$, denoted $N_{L_2}^*(F_1 \cup F_2)$, has size
    \begin{equation}
        |N_{L_2}^*(F_1 \cup F_2)| \geq (1-2\epsilon_1)\Delta_1|F_1| + (1-2\epsilon_2)\Delta_2|F_2|.
    \end{equation}
    By definition of $F_1$ and $F_2$, we must have that at least $\frac{\Delta_1}{2}|F_1| + \frac{\Delta_2}{2}|F_2|$ edges leaving $F_1 \cup F_2$ land on satisfied checks. Therefore, there are at least
    \begin{equation}
        \left(\frac{1}{2}-2\epsilon_1\right)\Delta_1|F_1| + \left(\frac{1}{2}-2\epsilon_2\right)\Delta_2|F_2| - |T|
    \end{equation}
    edges leaving $F_1 \cup F_2$ landing on satisfied checks, which have no other neighbours in $F_1 \cup F_2$, and that are not themselves corrupt. Each of these must therefore have another neighbour in $V_1 \cup V_2$. This implies a lower bound on the number of collisions in $L_2$ of edges leaving $V_1 \cup V_2$, and thus that
    \begin{equation}
        |N_{L_2}(V_1 \cup V_2)| \leq \Delta_1|V_1| + \Delta_2|V_2| - \left(\frac{1}{4}-\epsilon_1\right)\Delta_1|F_1| - \left(\frac{1}{4}-\epsilon_2\right)\Delta_2|F_2| + \frac{|T|}{2}.
    \end{equation}
    Combining this with Equation~\eqref{eq:VC_bound} gives us
    \begin{equation}
        \left(\frac{1}{2}-\epsilon_1\right)\Delta_1|C_1| + \left(\frac{1}{2}-\epsilon_2\right)\Delta_2|C_2| + \left(\frac{1}{4}-\epsilon_1\right)\Delta_1|F_1| + \left(\frac{1}{4}-\epsilon_2\right)\Delta_2|F_2| < \epsilon_1\Delta_1|V_1| + \epsilon_2\Delta_2|V_2| + \frac{3}{2}|T|,
    \end{equation}
    which amounts to an upper bound on the number of errorful message bits at the end of the algorithm in terms of the number of errorful message and check bits at the beginning of the algorithm.

   Now, at the end of the algorithm, the number of errors remaining in $R_1$ is
    \begin{align}
        |C_1| + |F_1| < \frac{\epsilon_1\Delta_1|X_q| + \epsilon_2\Delta_2|X_x| + (3/2)|X_z|}{(1/4-\epsilon_1)\Delta_1} &\leq \frac{8\epsilon_1\Delta_1|X_q| + 8\epsilon_2\Delta_2|X_x| + 12|X_z|}{\Delta_1}\\
        &<\frac{16|X_q| + 16(|X_x| + |X_z|)}{\Delta_1}.
    \end{align}
    Similarly, the number of errors remaining in $R_2$ is
    \begin{align}
        |C_2| + |F_2| < \frac{\epsilon_1\Delta_1|X_q| + \epsilon_2\Delta_2|X_x| + (3/2)|X_z|}{(1/4-\epsilon_2)\Delta_2} &\leq \frac{8\epsilon_1\Delta_1|X_q| +8\epsilon_2\Delta_2|X_x| + 12|X_z|}{\Delta_2}\\
        &< \frac{16|X_q| + 16(|X_x| + |X_z|)}{\Delta_2}.
    \end{align}
    One can see that, at the end of the round, the number of errors remaining in $R_i$ is less than
    \begin{equation}
        \frac{32}{\Delta_i}\max\left(|X_q|, |X_x| + |X_z|\right)
    \end{equation}
    for $i = 1, 2$. Concretely, this means that
    \begin{align}
        |X_q + \tilde{X}_q| &< \frac{32}{\Delta_1}\max\left(|X_q|, |X_x| + |X_z|\right)\\
        |X_x + \tilde{X}_x| &< \frac{32}{\Delta_2}\max\left(|X_q|, |X_x| + |X_z|\right).
    \end{align}
    Using that $|A^Tv| \leq \Delta_1\frac{n}{m}v$ for any vector $v$, and by a triangle inequality, we have the claimed bound on the remaining residual $X$ error $X_{Res} + \tilde{X}_{Res}$.
\end{proof}

\subsubsection{Parallel Reduction of Z Errors}

There is a further complication present for the parallel reduction of $Z$ errors. As before, we have the measured $X$-syndrome
\begin{equation}
    \sigma_x = Z_x + A\cdot Z_q + C\cdot Z_z = Z_x + A\cdot Z_{Res} + D^T\cdot Z_z
\end{equation}
and we aim to reduce the residual $Z$ error
\begin{equation}
    Z_{Res} = Z_q + B^T\cdot Z_z.
\end{equation}
Naively, we would attempt to perform the same parallel reduction to directly calculate an approximation to $Z_{Res}$. The problem with this is that, if one does so, one ends up with (loosely speaking)
\begin{equation}
    |Z_{Res} + \tilde{Z}_{Res}| < \frac{32}{\Delta_1}\max\left(|Z_{Res}|, |Z_z| + |Z_x|\right).
\end{equation}
This is problematic because we can only estimate $|Z_{Res}| \leq |Z_q| + \Delta_1\frac{n}{m}|Z_z|$, and in the worst case, our estimate would look something like
\begin{equation}
    |Z_{Res} + \tilde{Z}_{Res}| < 32\frac{n}{m}|Z_z|,
\end{equation}
which is not giving us error reduction at all.

To remedy this, we perform two rounds of the parallel error reduction. First, we attempt only to get an approximation to the error $Z_z$, and in doing so we get a reduction of this error by a factor $\sim \Delta_2$. We can use this to get an initial approximation to $Z_{Res}$ via $B^T\cdot Z_z$. Using this, we can then perform parallel error reduction on $Z_{Res}$. Formally, the algorithm for reducing the $Z$ errors in parallel is written in Algorithm~\ref{alg:Z_parallel_reduction}.

\begin{algorithm}[ht]
\caption{Parallel Error-Reduction Algorithm for $Z$ Errors}\label{alg:Z_parallel_reduction}
\begin{algorithmic}[1] 

\Require Lossless $Z$-Graph $G$; \hspace*{0.1cm} Syndrome $\sigma_x = Z_x + A\cdot Z_q + C\cdot Z_z = Z_x + A\cdot Z_{Res} + D^T\cdot Z_z$
\Ensure $\tilde{Z}_{Res}$: approximation to the residual $Z$ error $Z_{Res} = Z_q + B^T\cdot Z_z$
\State $\tilde{Z}_z \gets 0$
\State In parallel, for every bit in $\tilde{Z}_z$ (associated to the vertices in $L_2$), if the bit is in contact (via the graph $G$) with more $1$'s than $0$'s in $\sigma_x$ (associated to the vertices $R_2$), then flip it.
\State $\tilde{Z}_{Res} \gets B^T \cdot \tilde{Z}_z$
\State $\hat{\sigma}_x \gets \sigma_x + AB^T\cdot \tilde{Z}_z + D^T\cdot \tilde{Z}_z = \sigma_x + C\cdot \tilde{Z}_z$
\State In parallel, for every bit in $\tilde{Z}_{Res}$ (associated to the vertices in $L_1$), if the bit is in contact (via the graph $G$) with more $1$'s than $0$'s in $\hat{\sigma}_x$ (associated to the vertices $R_2$), then flip it.

\end{algorithmic}
\end{algorithm}
\begin{proposition}
    The parallel error-reduction algorithm for $Z$ errors may be run in constant depth and with a linear total number of gates.
\end{proposition}
\begin{lemma}\label{lem:parallel_Z_reduction}
    Suppose that we have a quantum error-reduction code built from an $(n,m,\Delta_1, \Delta_2, \eta_1, \eta_2, \epsilon_1, \epsilon_2)$ lossless $Z$-graph, where $\epsilon_1, \epsilon_2 \leq \frac{1}{8}, \epsilon_1< \frac{2}{\Delta_1}$ and $\epsilon_2< \frac{2}{\Delta_2}$, and where
    \begin{equation}\label{eq:imbalanced_Deltas_condition}
        \frac{128\Delta_1'^2}{\Delta_2} \leq 1,
    \end{equation}
    for $\Delta_1' = \frac{n}{m}\Delta_1$. Suppose that
    \begin{align}
        |Z_q| &\leq \alpha n\\
        |Z_z| &\leq \beta m\\
        |Z_x| &\leq \gamma m.
    \end{align}
    Then, there exist small enough constants $\alpha, \beta$ and $\gamma$ such that at the at the end of the parallel error-reduction algorithm for $Z$ errors, the size of the remaining residual $Z$ error is
    \begin{equation}
        |Z_{Res} + \tilde{Z}_{Res}| <\frac{128}{\Delta_1}\max\left(|Z_q|, |Z_x|+|Z_z|\right).
    \end{equation}
\end{lemma}
We comment that, for the parallel reduction of $Z$ errors, we require two properties of the lossless $Z$-graph that we may obtain from our randomised construction but not from our explicit construction. The first is the same as in the parallel reduction of $X$ errors, that is, that the expansion of strong enough to achieve $\epsilon_i = O(1/\Delta_i)$. The second is that, in the randomised construction, we may obtain the lossless $Z$-graphs for any pair of integers $\Delta_1, \Delta_2$, which will allow us to fulfil the condition of Equation~\eqref{eq:imbalanced_Deltas_condition}. In our explicit construction of lossless $Z$-graphs, $\Delta_1$ and $\Delta_2$ are chosen to have some constant imbalance, and then taken to be large enough numbers, thus making this sort of condition hard to satisfy.
\begin{proof}[Proof of Lemma~\ref{lem:parallel_Z_reduction}]
    Given small enough constants $\alpha, \beta$ and $\gamma$, the proof of Lemma~\ref{lem:parallel_X_reduction} applies to the first round of reduction with $Z_z$ in place of $X_x$, $Z_{Res}$ in place of $X_q$ and $Z_x$ in place of $X_z$. However, we emphasise that in this first round we are only attempting to reduce the error $Z_z$. By Lemma~\ref{lem:parallel_X_reduction}, in particular Equation~\eqref{eq:X_x_parallel_reduction} at the end of the first round of reduction, we have an approximation $\tilde{Z}_z$ to the error $Z_z$ such that the remaining error $\hat{Z}_z = Z_z + \tilde{Z}_z$ satisfies
    \begin{equation}
        |\hat{Z}_z| < \frac{32}{\Delta_2}\max\left(|Z_{Res}|, |Z_z| + |Z_x|\right).
    \end{equation}
    The second round of reduction then runs the same algorithm but with $Z_z$ replaced with $\hat{Z}_z$, $Z_{Res}$ replaced with $\hat{Z}_{Res} = Z_q + B^T\cdot \hat{Z}_z$, and the syndrome adjusted accordingly. Lemma~\ref{lem:parallel_X_reduction} applies once more, assuming $\alpha, \beta$ and $\gamma$ are sufficiently small, and by Equation~\eqref{eq:X_q_parallel_reduction}, the resultant correction reduces the residual $Z$ error to a size less than
    \begin{align}
        \frac{32}{\Delta_1}\max\left(|\hat{Z}_{Res}|, |\hat{Z}_z| + |Z_x|\right)
        &= \frac{32}{\Delta_1}\max\left(|Z_q + B^T\cdot (Z_z+\tilde{Z}_z)|, |\hat{Z}_z| + |Z_x|\right)\\
        &\leq \frac{32}{\Delta_1}\max\left(|Z_q| + \Delta_1'|Z_z+\tilde{Z}_z|, |Z_z+\tilde{Z}_z| + |Z_x|\right)\\
        &\leq \frac{128}{\Delta_1}\max\left(|Z_q|, |Z_x|, \Delta_1'|Z_z + \tilde{Z}_z|\right)\\
        &<\frac{128}{\Delta_1}\max\left(|Z_q|, |Z_x|, \frac{32\Delta_1'}{\Delta_2}\max\left(|Z_{Res}|, |Z_z| + |Z_x|\right)\right)\\
        &\leq\frac{128}{\Delta_1}\max\left(|Z_q|, |Z_x|, \frac{32\Delta_1'}{\Delta_2}\max\left(|Z_q| + \Delta_1'|Z_z|, |Z_z| + |Z_x|\right)\right)\\
        &\leq\frac{128}{\Delta_1}\max\left(|Z_q|, |Z_x|, \frac{128\Delta_1'}{\Delta_2}\max\left(|Z_q|, \Delta_1'|Z_z|, |Z_x|\right)\right)\\
        &=\frac{128}{\Delta_1}\max\left(|Z_q|, |Z_x|, \frac{128\Delta_1'^2}{\Delta_2}|Z_z|\right)\\
        &\leq\frac{128}{\Delta_1}\max\left(|Z_q|, |Z_x|, |Z_z|\right)\\
        &\leq\frac{128}{\Delta_1}\max\left(|Z_q|, |Z_x|+|Z_z|\right).
    \end{align}
\end{proof}

\section{Lossless Z-Graphs}\label{sec:lossless_Z_graphs}

We now state our definition of lossless $Z$-graphs. In Section~\ref{sec:random_lossless_Z_graphs}, we will prove their existence via a randomised procedure. In Section~\ref{sec:explicit_lossless_Z_graphs}, we will construct them explicitly.

\begin{definition}\label{def:main_lossless_Z-graphs}[Lossless $Z$-Graphs]
    An $(n, m, \Delta_1, \Delta_2, \eta_1, \eta_2, \epsilon_1, \epsilon_2)$ lossless $Z$-graph is a bipartite graph with the following properties. The left vertices are partitioned into two sets called $L_1$ and $L_2$, which have sizes $n$ and $m$, respectively. Similarly, the right vertices are partitioned into sets of size $R_1$ and $R_2$, which have sizes $n$ and $m$, respectively. The subgraphs obtained by restriction to $(L_1, R_2)$ and $(R_1, L_2)$ are $(\Delta_1, \Delta_1')$-biregular graphs, where $\Delta_1' = \frac{n}{m}\Delta_1$, and the subgraph obtained by restriction to $(L_2, R_2)$ is a $\Delta_2$-regular (bipartite) graph. There are no edges between vertices in $L_1$ and $R_1$.

    Given subsets $S_1 \subseteq L_1$ and $S_2 \subseteq L_2$ of size $|S_1| \leq \eta_1 n$ and $|S_2| \leq \eta_2 m$, $S_1 \cup S_2$ has a large number of neighbours in $R_2$, in particular,
    \begin{equation}
        |N_{R_2}(S_1\cup S_2)| \geq (1-\epsilon_1)\Delta_1\cdot |S_1| + (1-\epsilon_2)\Delta_2\cdot |S_2|. 
    \end{equation}
    Similarly, given subsets $S_1 \subseteq R_1$ and $S_2 \subseteq R_2$ of size $|S_1| \leq \eta_1 n$ and $|S_2| \leq \eta_2m$, $S_1 \cup S_2$ has a large number of neighbours in $L_2$, namely,
    \begin{equation}
        |N_{L_2}(S_1\cup S_2)| \geq (1-\epsilon_1)\Delta_1\cdot |S_1| + (1-\epsilon_2)\Delta_2\cdot |S_2|.
    \end{equation}
\end{definition}

\subsection{Random Construction}\label{sec:random_lossless_Z_graphs}

We now prove the existence of lossless $Z$-graphs via a simple randomised procedure.

\begin{theorem}\label{thm:random_lossless_Z_graphs}

Fix $0 < \delta < 1$, integers $\Delta_1, \Delta_2 > 0$ and $C > 0$. Then, there exist $\eta_1, \eta_2 > 0$ such that, for all $n$ large enough, there exists an $(n,m,\Delta_1, \Delta_2, \eta_1, \eta_2, \epsilon_1, \epsilon_2)$ lossless $Z$-graph, where $\frac{n}{m} = C$, and $\epsilon_i = \frac{1+\delta}{\Delta_i}$ for $i = 1,2$.
    
\end{theorem}
\begin{proof}
    We sample a lossless $Z$-graph uniformly at random from the natural ensemble. That is, we consider the set of left vertices, split into $L_1$ of size $n$ and $L_2$ of size $m$. We also consider the right vertices, split into $R_2$ of size $m$ and $R_1$ of size $n$. All vertices in $L_1$ and $R_1$ start with $\Delta_1$ half edges going towards the vertices in $R_2$ and $L_2$, respectively. The vertices in $L_2$ and $R_2$ start with $\Delta_1'$ half-edges going towards the vertices in $R_1$ and $L_1$, respectively, as well as $\Delta_2$ half-edges going towards each other. The half-edges are then joined to each other at random.

    We will show that there exist $\eta_1, \eta_2 > 0$ such that, with high probability, given subsets $S_1 \subseteq R_1$ and $S_2 \subseteq R_2$ of sizes $s_1 \leq \eta_1n$ and $s_2 \leq \eta_2m$, $N_{L_2}(S_1 \cup S_2)$ is large. By symmetry, and by a union bound, we will obtain the same expansion statement for subsets $S_1 \subseteq L_1$ and $S_2 \subseteq L_2$. Concretely, we will show that
    \begin{equation}
        |N_{L_2}(S_1 \cup S_2)| \geq (\Delta_1-1-\delta)s_1 + (\Delta_2-1-\delta)s_2.
    \end{equation}

    Let $\tilde{\delta} = \frac{\delta}{2}$. For the benefit of the latter part of the proof, we will begin by simply showing that given $S_2 \subseteq R_2$ of size $s_2 \leq \tilde{\eta}_2m$, where $\tilde{\eta}_2$ is a constant depending only on $\Delta_2$ and $\delta$, the number of neighbours of $S_2$ in $L_2$, i.e., $N_{L_2}(S_2)$, is large. In particular, we will show that $|N_{L_2}(S_2)| \geq \tilde{\beta}_2 s_2$, where $\tilde{\beta}_2 = \Delta_2-1-\tilde{\delta}$, which is a standard statement of lossless expansion. Indeed, we have $|N_{L_2}(S_2)| \leq t$, where $t = \tilde{\beta}_2s_2$, if and only if $N_{L_2}(S_2) \subseteq T$ for $T \subseteq L_2$ some set of size $t$. Fixing sets $S_2$ and $T$ of these sizes, we have $N_{L_2}(S_2) \subseteq T$ with probability
    \begin{equation}
        \frac{(t\Delta_2)(t\Delta_2-1)\ldots(t\Delta_2-s_2\Delta_2+1)}{(m\Delta_2)(m\Delta_2-1)\ldots(m\Delta_2-s_2\Delta_2+1)} \leq \left(\frac{t}{m}\right)^{s_2\Delta_2}.
    \end{equation}
    By a union bound over sets $S_2$ of size $s_2$, over sets $T$ of size $t = \tilde{\beta}_2s_2$, and finally over sizes $s_2 = 1, \ldots, \tilde{\eta}_2m$, the probability of having a bad set $S_2$ is at most
    \begin{align}
        \sum_{s_2=1}^{\tilde{\eta}_2m}\begin{pmatrix}
            m\\ s_2
        \end{pmatrix}\begin{pmatrix}
            m\\ t
        \end{pmatrix}\left(\frac{t}{m}\right)^{s_2\Delta_2} &\leq \sum_{s_2=1}^{\tilde{\eta}_2m}\left(\frac{em}{s_2}\right)^{s_2}\left(\frac{em}{t}\right)^t\left(\frac{t}{m}\right)^{s_2\Delta_2}\\
        &=\sum_{s_2=1}^{\tilde{\eta}_2m}\left[e^{1+\tilde{\beta}_2}\tilde{\beta}_2\left(\frac{t}{m}\right)^{\tilde{\delta}}\right]^{s_2}\\
        &\leq\sum_{s_2=1}^{\tilde{\eta}_2m}\left[e^{1+\tilde{\beta}_2}\tilde{\beta}_2^{1+\tilde{\delta}}\tilde{\eta}_2^{\tilde{\delta}}\right]^{s_2}
    \end{align}
    For any $\Delta_2$ and $\delta$, we may take $\tilde{\eta}_2$ to be a small enough constant that the quantity in square brackets is smaller than $q$, for any $0 < q < 1$. The probability of a bad set $S_2$ is then upper bounded by $\sum_{s_2=1}^\infty q^{s_2}$, which may be made arbitrarily small.

    The same considerations show that, with high probability, given any set $S_1 \subseteq R_1$ of size $s_1 \leq \tilde{\eta}_1n$, where $\tilde{\eta}_1$ is a constant depending only on $\Delta_1$, $\delta$ and $\frac{n}{m}$, $|N_{L_2}(S_1)| \geq \tilde{\beta}_1s_1$, where $\tilde{\beta}_1 = \Delta_1-1-\tilde{\delta}$.

    We now turn to the main statement on the joint expansion of sets $S_1 \subseteq R_1$ and $S_2 \subseteq R_2$ of sizes $s_1 \leq \eta_1n$ and $s_2 \leq \eta_2m$, respectively. We will show this in three cases, which will then all hold simultaneously by a union bound. The three cases will be
    \begin{enumerate}
        \item $\frac{s_1}{s_2} > \gamma^+$;
        \item $\frac{s_1}{s_2} < \gamma^-$;
        \item $\gamma^- \leq \frac{s_1}{s_2} \leq \gamma^+$.
    \end{enumerate}
    Here, $\gamma^\pm$ are constants defined via the equations
    \begin{equation}
        \frac{1}{\gamma^-\left(\frac{\tilde{\beta}_1}{\tilde{\beta}_2}\right)+1} = \frac{1}{\frac{1}{\gamma^+}\left(\frac{\tilde{\beta}_2}{\tilde{\beta}_1}\right)+1} = 1-\frac{\tilde\delta}{\max(\Delta_1, \Delta_2)-1},
    \end{equation}
    where one may check that $\gamma^- < \gamma^+$.

    For the first case of $\frac{s_1}{s_2} > \gamma^+$, we may use the expansion of sets $S_1 \subseteq R_1$ alone to show the expansion of the set $S_1 \cup S_2$. Indeed, the set $S_1 \cup S_2$ has at least as many neighbours in $L_2$ as $S_1$ alone does, which we know to be at least $\tilde{\beta}_1s_1$ (we may take $\eta_1 \leq \tilde{\eta}_1$). We then have, for any $\lambda \in (0,1)$,
    \begin{align}
        \tilde{\beta}_1s_1 &= \lambda\tilde{\beta}_1s_1 + (1-\lambda)\tilde{\beta}_1s_1\\
        &> \lambda\tilde{\beta}_1s_1 + (1-\lambda)\frac{\tilde{\beta}_1}{\tilde{\beta}_2}\gamma^+\tilde{\beta}_2s_2.
    \end{align}
    We may now choose $\lambda$ to be the solution of $\lambda = (1-\lambda)\frac{\tilde{\beta}_1}{\tilde{\beta}_2}\gamma^+$, which is
    \begin{equation}
        \lambda = \frac{1}{\frac{1}{\gamma^+}\left(\frac{\tilde{\beta}_2}{\tilde{\beta}_1}\right)+1} = 1-\frac{\tilde{\delta}}{\max(\Delta_1,\Delta_2)-1}.
    \end{equation}
    We find that the number of neighbours of $S_1 \cup S_2$ in $L_2$ is at least
    \begin{equation}
        \left(1-\frac{\tilde{\delta}}{\max(\Delta_1, \Delta_2)-1}\right)(\tilde{\beta}_1s_1+\tilde{\beta}_2s_2),
    \end{equation}which is turn is at least
    \begin{multline}
        \left(1-\frac{\tilde{\delta}}{\Delta_1-1}\right)\left(1-\frac{\tilde{\delta}}{\Delta_1-1}\right)(\Delta_1-1)s_1+\left(1-\frac{\tilde{\delta}}{\Delta_2-1}\right)\left(1-\frac{\tilde{\delta}}{\Delta_2-1}\right)(\Delta_2-1)s_2
        >       \beta_1s_1 + \beta_2s_2,
    \end{multline}
    where $\beta_1 = \Delta_1-1-\delta$ and $\beta_2 = \Delta_2-1-\delta$.

    For the second case, that of $\frac{s_1}{s_2} < \gamma^-$, the expansion of sets $S_2 \subseteq R_2$ alone provides the desired statement, via essentialy the same argument that we have just given. 
    
    It remains to treat the third case, that of $\gamma^- \leq \frac{s_1}{s_2} \leq \gamma^+$. Let $t = \beta_1s_1 + \beta_2s_2$ and fix sets $S_1 \subseteq R_1$, $S_2 \subseteq R_2$ and $T \subseteq L_2$ of sizes $s_1, s_2$ and $t$, respectively. We have that $N_{L_2}(S_1 \cup S_2) \subseteq T$ with probability
    \begin{equation}
        \frac{(t\Delta_2)(t\Delta_2-1)\ldots (t\Delta_2-s_2\Delta_1+1)}{(m\Delta_2)(m\Delta_2-1)\ldots(m\Delta_2-s_2\Delta_2+1)}\cdot\frac{(t\Delta_1')(t\Delta_1'-1)\ldots (t\Delta_1'-s_1\Delta_1+1)}{(m\Delta_1')(m\Delta_1'-1)\ldots(m\Delta_1'-s_1\Delta_1+1)} \leq \left(\frac{t}{m}\right)^{s_1\Delta_1+s_2\Delta_2}.
    \end{equation}
    We now union bound over sets $T$ of size $t$, $S_1$ of size $s_1$, and $S_2$ of size $s_2$, and finally over positive integers $s_1$ and $s_2$ such that $s_1 \leq \eta_1n$, $s_2 \leq \eta_2m$ and $\gamma^- \leq \frac{s_1}{s_2}\leq \gamma^+$. One finds that the probability of bad sets $S_1$ and $S_2$ is at most
    \begin{equation}
        \sum_{s_1, s_2}\left[e^{1+\beta_1}\frac{n}{m}\frac{t}{s_1}\left(\frac{t}{m}\right)^\delta\right]^{s_1}\left[e^{1+\beta_2}\frac{t}{s_2}\left(\frac{t}{m}\right)^\delta\right]^{s_2},
    \end{equation}
    where the range of the summation is as described. We may upper bound this quantity as
    \begin{equation}
        \sum_{s_1, s_2}\left[e^{1+\beta_1}\frac{n}{m}\left(\beta_1+\frac{\beta_2}{\gamma^-}\right)\left(\frac{t}{m}\right)^\delta\right]^{s_1}\left[e^{1+\beta_2}\left(\beta_1\gamma^++\beta_2\right)\left(\frac{t}{m}\right)^\delta\right]^{s_2},
    \end{equation}
    which we may further upper bound as
    \begin{equation}
        \sum_{s_1=1}^{\eta_1n}\sum_{s_2=1}^{\eta_2m}\left[e^{1+\beta_1}\frac{n}{m}\left(\beta_1+\frac{\beta_2}{\gamma^-}\right)\left(\frac{t}{m}\right)^\delta\right]^{s_1}\left[e^{1+\beta_2}\left(\beta_1\gamma^++\beta_2\right)\left(\frac{t}{m}\right)^\delta\right]^{s_2}.
    \end{equation}
    Recalling that $t = \beta_1s_2+\beta_2s_2$, we see that by taking $\eta_1, \eta_2$ to be sufficiently small, the quantities in both the square brackets may be made small, and summing the geometric series as before leads to a small quantity.
\end{proof}

\subsection{Explicit Construction}\label{sec:explicit_lossless_Z_graphs}

We now prove the existence of explicit constructions of lossless $Z$-graphs with slightly weaker properties than the randomized construction of Theorem~\ref{thm:random_lossless_Z_graphs}.

\begin{theorem} \label{thm:Z-lossless}
    For any $\eps > 0$, $\alpha > 0$, and $\beta_1, \beta_2 \in \bbN$, there exist $k = k(\eps) \in \bbN$ and $d_0 = d_0(\eps, \alpha, \beta_1, \beta_2)$ such that for all integers $d_1, d'_1, d_2 \ge d_0$ with $\beta_1 d_1 = \beta_2 d'_1$ and $d_2 \in (1 \pm 0.001) \alpha \cdot d_1$, there are $\eta = \eta(k, d_1, d'_1, d_2) > 0$, $q_0, D' \in \bbN$, and an explicit family $\{ \calG_q \}_{q \in \calQ}$ of $(n_q, m_q, \Delta_1, \Delta_2, \eta, \eta, \eps, \eps)$ lossless Z-graphs indexed by $\calQ = \{ q : q \text{ is prime}, q \equiv 1 \text{ mod } 4, q > q_0 \}$, where $n_q = \beta_1 \cdot \frac{q(q^2-1)D'}{2}$ and $m_q = \beta_2 \cdot \frac{q(q^2-1)D'}{2}$, $\Delta_1 := kd_1$, and $\Delta_2 := kd_2$. Further, there is an algorithm that takes as input $q \in \calQ$ and in time $\mathrm{poly}(q)$ outputs $\calG_q$. 
\end{theorem}

The remainder of this section will focus on proving \Cref{thm:Z-lossless}.

\subsubsection{Components}

Our Z-graph will be constructed by combining a \emph{base graph} and a constant size \emph{Z-gadget graph}. The following was shown in~\cite{HLMRZ25}.

\begin{definition}[Structured bipartite graph]
\label{def:structured-bipartite-graph}
    A $(k,D)$-biregular bipartite graph $G$ between vertex sets $V$ and $M$ is a \emph{structured bipartite graph} if:
    \begin{enumerate}[(1)]
        \item The set $M$ can be expressed as a disjoint union $\sqcup_{a\in[k]}M_a$ such that each $v\in V$ has exactly one neighbor in each $M_a$.
        \item \label{property:nbr}
        For each vertex $u\in M$, there is an injective function $\mathrm{Nbr}_u:[D]\to V$ that specifies an ordering of the $D$ neighbors of $u$.
        \item \label{property:special-sets}
        There is an $s\in \bbN$ such that the following holds:
        for each pair of distinct $a, b\in [k]$, there are $r(a,b)$ \emph{special sets} $\{Q_i^{a,b} \subseteq [D]\}_{i\in[r(a,b)]}$ that partition $[D]$ (abbreviated to $r$ and $Q_i$), each $|Q_i| \in [\frac{D}{2s}, \frac{2D}{s}]$,
        such that for every $u \in M_a$, there are distinct $v_1,\dots,v_r \in M_b$ with $N(u)\cap N(v_i) = \Nbr_u(Q_i)$ for each $i\in[r]$ and $N(u) \cap N(v') = \varnothing$ for all other $v'\in M$.
    \end{enumerate}
\end{definition}

\begin{lemma}[\cite{HLMRZ25}] \label{lem:HLMRZ-base-graph}
    For every $k$ that is a power of $2$, and large enough $D\in\bbN$, there is an algorithm that takes in large enough prime $q\equiv 1\text{ mod }4$ and constructs vertex sets $V, M$ such that $|M| = k \cdot \frac{q(q^2-1)}{2}$, $|V| = D \cdot \frac{q(q^2-1)}{2}$ along with structured $(k, D)$-biregular bipartite graph $G$ on $(V,M)$, with the following properties:
    \begin{itemize}
        \item $s = \Theta(\sqrt{D})$ for the special set structure.
        \item $G$ is a $O\parens*{D^{5/8}}$-small-set $2\sqrt{k}$-neighbor expander.
        \item $G$ is a $O\parens*{D^{1/4}}$-small-set skeleton expander.
    \end{itemize}
\end{lemma}

The above lemma references notions of ``small set neighbor expanders'' and ``small set skeleton expander.'' We will not need to know what these mean, though the curious reader may refer to \cite{HLMRZ25} for details. 

Let $\beta \in \bbN$ and let $G = (V, M)$ be a bipartite graph. We will be concerned with properties of the \emph{$\beta$-duplication of $G$}, which is the graph on vertex sets $V' = V_1 \cup \dots \cup V_\beta$ and $M'$, where $M' = M$, each $V_i$ is a copy of $V$, and each graph $G'[V_i, M]$ is a copy of $G$.

\begin{proposition} \label{prop:dup-exp-defs}
    Let $\beta \in \bbN$ and suppose $G = (V, M)$ is a $\tau$-small-set $j$-neighbor expander (resp. $\lambda$-small-set skeleton expander). Then, the $\beta$-duplication of $G$ is a $\beta\tau$-small-set $j$-neighbor expander (resp. $\beta\lambda$-small-set skeleton expander).
\end{proposition}

The proof of \Cref{prop:dup-exp-defs} is straightforward from the definitions.

\begin{proposition} \label{prop:dup-structured-bipartite}
    Let $\beta \in \bbN$, and let $G = (V,M)$ be a $(k,D)$-biregular structured bipartite graph. Then, the $\beta$-duplication of $G$ is a structured bipartite graph.
\end{proposition}

\begin{proof}
    The three properties in \Cref{def:structured-bipartite-graph} follow straightforwardly from the fact that each $(V_j, M)$ is a structured bipartite graph. To be explicit, for \Cref{property:nbr}, the function $\Nbr_u : [\beta D] \cong [\beta] \times [D] \rightarrow V_1 \cup \dots \cup V_\beta$ sends $(i, j)$ to the vertex $\Nbr_u(j)$ in $V_i$, and for \Cref{property:special-sets}, the special set $(\calQ')_i^{(a,b)}$ is defined to be $\calQ_i^{(a,b)} \times [\beta]$, i.e. it is the union of the $\beta$ special sets in each of $V_j$, $j \in [\beta_2]$. 
\end{proof}

\paragraph{The base graph.} We are now ready to define our base graph. For $k$ a power of $2$ and sufficiently large $D' \in \bbN$, let $G' = (V', M')$ be the structured $(k',D')$-biregular bipartite graph guaranteed in \Cref{lem:HLMRZ-base-graph}. Our base graph $G$ consists of five vertex sets, $L_1, L_2, M, R_1, R_2$. The graphs $G[L_1, M]$ and $G[R_1, M]$ are $\beta_1$-duplications of $G'$, and the graphs $G[L_2, M]$ and $G[R_2, M]$ are $\beta_2$-duplications of $G'$. That is, $L_1 = L_{1,1} \cup \dots \cup L_{1,\beta_1}$ and $R_1 = R_{1,1} \cup \dots \cup R_{1,\beta_1}$ each consist of $\beta_1$ copies of $V'$, where each $G[L_{1,i}, M]$ and $G[R_{1,i}, M]$ is a copy of $G'$, and $L_2 = L_{2,1} \cup \dots \cup L_{2,\beta_2}$ and $R_2 = R_{2,1} \cup \dots \cup R_{2,\beta_2}$ each consist of $\beta_2$ copies of $V'$, where each $G[L_{2,i}, M]$ and $G[R_{2,i}, M]$ is a copy of $G'$.

The following is a corollary of \Cref{lem:HLMRZ-base-graph}, \Cref{prop:dup-exp-defs}, and \Cref{prop:dup-structured-bipartite}.

\begin{lemma} \label{lem:our-base-graph}
    For every $k$ that is a power of $2$, integers $\beta_1, \beta_2 \in \bbN$, and large enough $D'\in\bbN$, letting $D_1 := \beta_1 D'$, $D_2 := \beta_2 D'$, and $D := D_1 + D_2$, there is an algorithm that takes in large enough prime $q \equiv 1 \text{ mod }4$ and constructs vertex sets $L_1, L_2, M, R_1, R_2$ such that $|M| = k \cdot \frac{q(q^2-1)}{2}$, $|L_1| = |R_1| = \beta_1 \cdot \frac{q(q^2-1) D'}{2}$, $|L_2| = |R_2| = \beta_2 \cdot \frac{q(q^2-1)D'}{2}$ and a graph $G$ on vertex set $L_1 \cup L_2 \cup M \cup R_1 \cup R_2$, satisfying the following properties:
    \begin{itemize}
        \item $G[L_1, M]$ and $G[R_1, M]$ are $(k, D_1)$-biregular structured bipartite graphs,
        \item $G[L_2, M]$ and $G[R_2, M]$ are $(k, D_2)$-biregular structured bipartite graphs,
        \item $s = \Theta(\sqrt{D})$ for the special set structure.
        \item $G[L_1, M], G[L_2, M], G[R_1,M], G[R_2, M]$ are $O\parens*{D^{5/8}}$-small-set $2\sqrt{k}$-neighbor expanders.
        \item $G[L_2, M]$ and $G[R_2, M]$ are a $O\parens*{D^{1/4}}$-small-set skeleton expanders.
    \end{itemize}
\end{lemma}

\paragraph{The Z-gadget graph.} To go with our base graph, we will want a gadget graph $\calH$ on vertex sets $Y_1 \cup Y_2 \cup Z_1 \cup Z_2$, where $Y_1 = Z_1 = [D_1]$ and $Y_2 = Z_2 = [D_2]$, that has edges between only parts $Z_1$ and $Y_2$, $Y_2$ and $Z_2$, and $Z_2$ and $Y_1$. We will also need $\calH$ to respect the special sets of the base graph. The existence of such a gadget graph is given in the following lemma, proven in \Cref{sec:proof-Zgadget}. 

\begin{restatable}[Z-gadget graph]{lemma}{Zgadgetz}
\label{lem:Zgadget}
    Let $D_1, D_2, d_1, d'_1, d_2,k$ be integers such that $D_1 \cdot d_1 = D_2 \cdot d'_1$, and $k \le D^{0.1} \le d_1, d'_1, d_2 \le o(D)$, where $D := D_1 + D_2$. Let $Y_1, Z_1$ be sets of size $D_1$, and let $Y_2, Z_2$ be sets of size $D_2$. Also suppose that for some $s \in \bbN$, we have for all distinct $a,b \in [k]$ an integer $r(a,b)$ and partitions $(\calY_i^{a,b})_{i \in [r(a,b)]}$ of $Y_2$ and $(\calZ_i^{a,b})_{i \in [r(a,b)]}$ of $Z_2$ where each part has size within $\left[ \frac{D_2}{2s}, \frac{2D_2}{s} \right]$. 
    
    Then, there exists a graph $\calH$ on vertex set $(Y_1, Y_2, Z_1, Z_2)$ where $Y_1 = Z_1 = [D_1], Y_2 = Z_2 = [D_2]$, such that the subgraphs $(Y_1, Z_2)$ and $(Z_1, Y_2)$ are $(d_1, d'_1)$-biregular graphs, the subgraph $(Y_2, Z_2)$ is a $(d_2,d_2)$-biregular graph, and there are no edges between $Y_1$ and $Z_1$, also satisfying the following properties:
    \begin{itemize}
    \item {\bf Expansion from $Y_1 \cup Y_2$ to $Z_2$:} 
        For any $\mu = o_D(1)$, there is $\eps = o_D(1)$ such that for any $A_1 \in [Y_1]$ and $A_2 \in [Y_2]$ with $d_1 |A_1| + d_2 |A_2| \le \mu \cdot D_2$, we have 
        \[
            |N_{Z_2}(A_1 \cup A_2)| \ge (1 - \eps) \left( d_1 |A_1| + d |A_2| \right).
        \]
    \item {\bf Expansion from $Z_1 \cup Z_2$ to $Y_2$:}
        For any $\mu = o_D(1)$, there is $\eps = o_D(1)$ such that for any $B_1 \in [Z_1]$ and $B_2 \in [Z_2]$ with $d_1 |B_1| + d_2 |B_2| \le \mu \cdot D_2$, we have 
        \[
            |N_{Y_2}(B_1 \cup B_2)| \ge (1 - \eps) \left( d_1 |B_1| + d_2 |B_2| \right).
        \]
    \item {\bf Spread w.r.t. $(\calZ_{i}^{a,b})_{i \in [r(a,b)]}$:}
        For any $A_1 \subseteq [Y_1]$ and $A_2 \subseteq [Y_2]$, for any distinct $a,b \in [k]$, and for any $W \subseteq [s]$ with $|W| \ge \frac{s \log D}{\min\{ d_1, d_2 \}}$, 
        \[
            \sum_{i \in W} |N_{Z_2}(A_1 \cup A_2) \cap \calZ_i^{a,b}| \le 64 |W| \cdot \max \left\{ \frac{d_1 |A_1| + d |A_2|}{s}, \log D \right\}. 
        \]
    \item {\bf Spread w.r.t. $(\calY_{i}^{a,b})_{i \in [r(a,b)]}$}
        For any $B_1 \subseteq [Z_1]$ and $B_2 \subseteq [Z_2]$, for any distinct $a,b \in [k]$, and for any $W \subseteq [s]$ with $|W| \ge \frac{s \log D}{\min\{ d_1, d_2 \}}$, 
        \[
            \sum_{i \in W} |N_{Y_2}(B_1 \cup B_2) \cap \calZ_i^{a,b}| \le 64 |W| \cdot \max \left\{ \frac{d_1 |B_1| + d |B_2|}{s}, \log D \right\}. 
        \]
    \end{itemize}
\end{restatable}

\subsubsection{Our Expanding Z-Graph}

For $\eps \in (0, 1)$, integers $\beta_1, \beta_2 \in \bbN$, and constant $\alpha > 0$, let us pick $k \ge 16/\eps^2$ to be a power of $2$, and let us choose $\delta \in (0, 1)$ and large enough $d_1, d'_1, d_2, D_1, D_2 \in \bbN$ such that $\beta_1 d_1 = \beta_2 d'_1$, $d_2 \in (1 \pm 0.001) \alpha \cdot d_1$, $D^{-1/16} \le \delta \le o_D(1) \cdot \frac{1}{k^2}$, and $\frac{D^{1/4} \log^2 D}{\delta} \le d_1, d'_1, d_2 \le \frac{\delta D^{3/8}}{\log D}$, where $D := D_1 + D_2$. 

Let $G$ be a base graph on vertex set $(L_1, L_2, M, R_1, R_2)$ as given in \Cref{lem:our-base-graph} such that $G[L_1, M]$ and $G[R_1, M]$ are $(k, D_1)$-biregular structured bipartite graphs, and $G[L_1, M]$ and $G[R_2, M]$ are $(k, D_2)$-biregular structured bipartite graphs. We have that $G[L_1, M], G[L_2, M], G[R_1, M], G[R_2, M]$ are $O(D^{5/8})$-small-set $2\sqrt{k}$-neighbor expanders. Let $\calH$ be a Z-gadget graph on vertex sets $(Y_1, Y_2, Z_1, Z_2)$ as given in \Cref{lem:Zgadget} such that $Y_1 = Z_1 = [D_1]$, $Y_2 = Z_2 = [D_2]$, and the subgraphs $\calH[Y_1, Z_2]$ and $\calH[Z_1, Y_2]$ are $(d_1, d'_1)$-biregular graphs, and the subgraph $\calH[Y_2, Z_2]$ is a $(d_2, d_2)$-biregular graph. $\calH$ also satisfies spread w.r.t. the special sets of $G[R_2, M]$ and $G[L_2, M]$. 


We define our graph $\calG$ as follows:
\begin{itemize}
    \item $\calG$ is on vertex sets $(L_1 \cup L_2, R_1, \cup R_2)$.
    \item For each $u \in M$, we place a copy of $\calH = (Y_1, Y_2, Z_1, Z_2; E_H)$ on the vertex sets $(N^G_{L_1}(u), N^G_{L_2}(u), N^G_{R_1}(u),$ $ N^G_{R_2}(u))$, where the bijection of the sets $Y_i$ to $N^G_{L_i}(u)$ and $Z_i$ to $N^G_{R_i}(u)$ are given by the orderings of the neighbors of $u$ (\Cref{property:nbr}). We let this copy of $\calH$ on the neighbors of $u \in M$ be denoted $\calH_u$. The edges of $\calG$ are all the edges from the union of all the $\calH_u$, $u \in M$. 
\end{itemize}
Note that each vertex in $L_1 \cup L_2 \cup R_1 \cup R_2$ has exactly $k$ neighbors in $M$ in $G$, so they each partake in exactly $k$ different copies of $H$. Thus $\calG[Y_1, Z_2]$ and $\calG[Z_1, Y_2]$ are $(\Delta_1, \Delta'_1)$-biregular bipartite graphs, where $\Delta_1 = k d_1$ and $\Delta'_1 = k d'_1$, and $\calG[Y_2, Z_2]$ is a $(\Delta_2, \Delta_2)$-biregular bipartite graph, where $\Delta_2 = k d_2$. We also let $n = |L_1| = |R_1| = \beta_1 \cdot \frac{q(q^2-1)D'}{2}$ and $m = |L_2| = |R_2| = \beta_2 \cdot \frac{q(q^2-1)D'}{2}$. Notice that $\frac{n}{m} = \frac{\beta_1}{\beta_2}$. 

\subsubsection{Analysis}

The proof that $\calG$ satisfies the properties in \Cref{thm:Z-lossless} follows via a slight (and simple) modification of the proof of lossless expansion in \cite{HLMRZ25}. In what follows, we give a proof sketch of \Cref{thm:Z-lossless}, referring the reader to \cite{HLMRZ25} for proofs of the claims.

Let us focus on showing expanion from subsets of $L_1 \cup L_2$ to $R_2$, as the $R_1 \cup R_2 \rightarrow L_2$ case is identical. Fix $S_1 \subseteq L_1$, $S_2 \cup L_2$, so that $|S_1|, |S_2| \le \eta D$. We will show that $|N_{R_2}(S_1 \cup S_2)| \ge (1-\eps) (\Delta_1 |S_1| + \Delta_2 |S_2|)$. 

Let $U \subseteq M$ be the neighbors of $S_1$ and $S_2$ in the graph $G$. We split $U$ into its ``high degree'' part $U_h := \{ u \in U : \deg_{G[S_1 \cup S_2, U]}(u) \ge \frac{\tau}{\delta} \}$ and its ``low degree'' part $U_\ell := U \backslash U_h$. Here, $\tau = \tau_1 + \tau_2$, where $G[L_1, M]$ and $G[L_2, M]$ are $\tau_1$ and $\tau_2$-small-set $2\sqrt{k}$-neighbor expanders, respectively, and $\tau_1, \tau_2 = O(D^{5/8})$. It happens that $\frac{\tau_1}{\tau_2} = \frac{\beta_1}{\beta_2}$, so the precise setting of $\tau$ is not important, only that it captures the asymptotic size of $\tau_1, \tau_2$.

The first step is to show that most of the edges from $S_1$ and $S_2$ to $U$ in fact point to $U_\ell$.

\begin{claim}
    For each $i \in \{ 1, 2 \}$, the number of edges in $G[S_i, U]$ incident to $U_\ell$ is at least $\parens*{1-\sqrt{\delta} - 2k^{-1/2}}\cdot k |S_i|$.
\end{claim}

The proof of this claim is identical to Claim 2.10 in \cite{HLMRZ25} so we omit it. It relies only on the fact that $G[L_i, M]$ is a $\tau_i$-small-set $2\sqrt{k}$-neighbor expander. Essentially, this claim says that most edges in $G$ coming out of $S_1$ and $S_2$ point to low-degree middle vertices, whose gadgets experience good expansion into $R_2$. 

\begin{definition}
    For $S_1 \subseteq L_1$, $S_2 \subseteq L_2$, and $U = N_{G}(S_1 \cup S_2) \subseteq M$, if a vertex $v \in R_2$ is a neighbor of $S$ in the final product due to connections from the gadget $H_u$ for $u\in U$, then we color the edge $(u, v)$ red.
    The red edges form a subgraph of $G[R_2, M]$, which we denote as $\RED(S_1 \cup S_2)$.
\end{definition}

By the choice of the threshold, we have $\frac{\tau}{\delta} \le o_D(1) \cdot \frac{D_2}{d_1 + d_2}$, and hence if we inspect the gadget $\calH_u$ around $u \in U_\ell$, letting $S_1(u) := \Nbr^{G[L_1, M]}_u(N^G_{S_1}(u)) \subseteq Y_1$ and $S_2(u) := \Nbr^{G[L_2, M]}_u (N^G_{S_2}(u)) \subseteq Y_2$, it holds that $|N^{\calH_u}_{Z_2}(S_1(u) \cup S_2(u))| \ge (1 - o_D(1)) (d_1 |S_1(u)| + d_2 |S_2(u)|)$. In particular, this implies that
\begin{align*}
    e(\RED(S_1 \cup S_2)) 
    &\ge \sum_{u \in U_\ell} (1 - o_D(1)) (d_1 |S_1(u)| + d_2 |S_2(u)|) \\
    &= (1 - o_D(1)) \cdot \left( d_1 \cdot e_{G}(S_1, U_\ell) + d_2 \cdot e_G(S_2, U_\ell) \right) \\
    &\ge (1 - o_D(1)) (1 - \sqrt{\delta} - 2k^{-1/2}) \cdot (d_1 \cdot e_G(S_1, U) + d_2 \cdot e_G(S_2, U)).
\end{align*}

The remainder of the argument is to show that there are very few collisions between the neighborhoods of different gadgets, so that most of the $\RED$ edges are actually going to distinct elements of $R_2$. 

We construct the \emph{collision graph} $C$ --- the multi-graph $C$ on vertex set $U \subseteq M$ by placing a copy of the edge $\{u,v\}$ for each $u \neq v \in U$, and $r\in R_2$ such that $\{u,r\}$ and $\{v,r\}$ are red edges in $\RED$.
The number of neighbors of $S_1 \cup S_2$ in $R_2$ in the final product $\calG$ is at least
\begin{align*}
    e(\RED) - e(C),
\end{align*}
since a vertex $v \in R_2$ with degree $d_v$ in $\RED$ contributes one neighbor, but it is counted $d_v$ times in $e(\RED)$ and $\binom{d_v}{2}$ times in $e(C)$, and $d_v - \binom{d_v}{2} \leq 1$ for all $d_v \in \bbN$.

\begin{claim}[\cite{HLMRZ25}] \label{claim:middle-to-right}
    Suppose $k\delta^2 \leq o_D(1)$, $\lambda \leq s \delta$, and $d_1, d_2 \geq \frac{1}{\delta} \max\{\lambda, \sqrt{s}\} \log D$.
    Then, 
    \[
        e(C) \leq o_D(1) \cdot k(d_1 |S_1| + d_2 |S_2|) = o_D(1) \cdot (\Delta_1 |S_1| + \Delta_2 |S_2|).
    \]
\end{claim}

The proof of this claim is analogous to the proof of Claim 2.13 in \cite{HLMRZ25}, so we omit it here. It relies on the graph $G[R_2, M]$ being a $O(D^{1/4})$-small-set skeleton expander. 

Now, to finish up the proof of \Cref{thm:Z-lossless}, we have that for $\delta \le o_D(1) \cdot \frac{1}{k^2}$ and $k \ge 16/\eps^2$, 
\begin{align*}
    e(\RED(S_1 \cup S_2)
    &\ge (1 - o_D(1)) (1 - \sqrt{\delta} - 2k^{-1/2}) \cdot (d_1 \cdot e_G(S_1, U) + d_2 \cdot e_G(S_2, U)) \\
    &\ge (1-\eps/2) \cdot k(d_1 |S_1| + d_2|S_2|) \\
    &= (1-\eps/2) \cdot (\Delta_1 |S_1| + \Delta_2 |S_2|).
\end{align*}
The number of neighbors of $S_1 \cup S_2$ in the final graph $\calG$ is at least $e(\RED(S_1 \cup S_2)) - e(C)$, and from \Cref{claim:middle-to-right} we have that $e(C) \le o_D(1) \cdot (\Delta_1 |S_1| + \Delta_2 |S_2|)$. Thus, choosing $D_1, D_2$ large enough, we have that 
\[
    |N_{R_2}(S_1 \cup S_2)| \ge (1 - \eps) \cdot (\Delta_1 |S_1| + \Delta_2 |S_2|).
\]
The proof for expansion from $R_1 \cup R_2$ to $L_2$ is identical.

\subsubsection{The Existence of Good Z-Gadget Graphs} \label{sec:proof-Zgadget}

In this section, we prove \Cref{lem:Zgadget}, restated below.

\Zgadgetz*

\noindent
Define a distribution $\calG_{D_1,D_2,d_1,d'_1,d_2}$ of random graphs as follows. 
\distbox{$\calG_{D_1,D_2,d_1,d'_1,d_2}$}{
\begin{itemize}
    \item The vertex set consists of $Y_1 \cup Y_2 \cup Z_1 \cup Z_2$, where $|Y_1| = |Z_1| = D_1$ and $|Y_2| = |Z_2| = D_2$.
    \item Sample a random $(d_1, d'_1)$-biregular graph on $D_1 + D_2$ vertices and place it on vertex sets $(Y_1, Z_2)$ and $(Z_1, Y_2)$.
    \item Sample a random $(d_2, d_2)$-biregular graph on $D_2 + D_2$ vertices and place it on vertex sets $(Y_2, Z_2)$.
\end{itemize}
}
\noindent
We will show that a random graph sampled from $\calG_{D_1,D_2,d_1,d'_1,d_2}$ will with high probability satisfy both expansion and spread as stated in \Cref{lem:Zgadget}. We start with spread.

\paragraph{Spread.}
Showing that a random graph sampled from $\calG_{D_1,D_2,d_1,d'_1,d_2}$ satisfies spread is a simple corollary of the following lemma from~\cite{HLMOZ25}.

\begin{lemma}[\cite{HLMOZ25}, Lemma 5.3] \label{lem:HLMOZ-spread}
    Let $H$ be a random $(d_1, d_2)$-biregular graph on vertex sets $Y \cup Z$, and let $n_1 = |Y|$, $n_2 = |Z|$, and $n = n_1 + n_2$. Let $Z = Z_1 \cup \dots \cup Z_s$ be a fixed partition of $Z$ such that $\frac{|Z|}{2s} \le |Z_i| \le \frac{2|Z|}{s}$ for each $i \in [s]$. Then, with probability $1 - \Theta(1/n)$, we have that for all $A \subseteq Y$ and all $W \subseteq [s]$ with $|W| \ge \frac{s \log n}{d_1}$, 
    \[
        \sum_{i \in W} |N_Z(A) \cap Z_i| \le 32 |W| \cdot \max \left\{ \frac{d_1}{s} |A|, \log n \right\} 
    \]
\end{lemma}

Let us apply \Cref{lem:HLMOZ-spread} to the random graphs $(Y_1, Z_2)$ and $(Y_2, Z_2)$ sampled in $\calG_{D_1,D_2,d_1,d'_1,d_2}$ and to the partitionings $(\calZ_i^{a,b})_{i \in [s(a,b)]}$. We obtain that with probability $1 - \Theta(k^2/D)$, for all $A_1 \subseteq Y_1$ and $A_2 \subseteq Y_2$ and all $W \subseteq [s]$ with $|W| \ge \frac{s \log D}{\min \{ d_1, d_2 \}}$,
\begin{align*}
    \sum_{i \in W} |N_{Z_2}(A_1 \cup A_2) \cap \calZ_i^{a,b}|
    &\le \sum_{i \in W} |N_{Z_2}(A_1) \cap \calZ_i^{a,b}| 
     + \sum_{i \in W} |N_{Z_2}(A_2) \cap \calZ_i^{a,b}|\\
    &\le 32 |W| \cdot \max \left\{ \frac{d_1}{s} |A_1|, \log D \right\}
     + 32 |W| \cdot \max \left\{ \frac{d_2}{s} |A_2|, \log D \right\} \\
    &\le 64 |W| \cdot \max \left\{ \frac{d_1 |A_1| + d_2 |A_2|}{s}, \log D \right\},
\end{align*}
as desired. The analogous statement for sets in $Z_1 \cup Z_2$ expanding to $Y_2$ follows via the exact same proof.

\paragraph{Expansion.} 
To show expansion, we will follow the outline given in~\cite{HMMP24}. We first show the desired statement for the appropriate mixture of Erd\H{o}s-Renyi graphs, and then transfer the result to the mixed-regular random graphs that we are sampling from.

The Erd\H{o}s-Renyi type distribution of graphs we consider is the following, denoted $\calE_{D_1,D_2,p_1,p_2}$.

\distbox{$\calE_{D_1,D_2,p_1,p_2}$}{
\begin{itemize}
    \item The vertex set consists of $Y_1 \cup Y_2 \cup Z_1 \cup Z_2$, where $|Y_1| = |Z_1| = D_1$ and $|Y_2| = |Z_2| = D_2$.
    \item For pair of vertices in $Y_1 \times Z_2$ and in $Z_1 \times Y_2$, place an edge between them with probability $p_1$.
    \item For each pair of vertices in $Y_2 \times Z_2$, place an edge between them with probability $p_2$.
\end{itemize}
}

\begin{lemma} \label{lem:erdos-renyi-lossless}
    Let $E \sim \calE_{D_1,D_2,p_1,p}$. Then, with probability $1 - O(1/D)$, for all sets $A_1 \subseteq Y_1$ and $A_2 \subseteq Y_2$ of sizes $t_1$ and $t_2$ respectively,
    \[
        |N_{Z_2}(A_1 \cup A_2)| \ge (1 - \eps_{t_1,t_2}) \cdot (t_1p_1 + t_2p_2) D_2,
    \]
    where $\eps_{t_1,t_2} = t_1p_1 + t_2p_2 + \sqrt{\frac{4(t_1+t_2)\log D}{(t_1p_1 + t_2p_2) D_2}}$.
    The same bound holds simultaneously for subsets of $Z_1, Z_2$ expanding into $Y_2$.
\end{lemma}

\begin{proof}
    For $A_1 \subseteq Y_1$ with $|A_1| = t_1$ and $A_2 \subseteq Y_2$ with $|A_2| = t_2$, we have that
    \[
        |N_{Z_2}(A_1 \cup A_2)| = \sum_{v \in Z_2} \mathbbm{1}[v \in N_{Z_2}(A_1 \cup A_2)].
    \]
    For each $v \in Z_2$, the probability that there is an edge between $v$ and $A_1 \cup A_2$ is at least the probability that there is \emph{exactly one} edge between $v$ and $A_1 \cup A_2$, which is $q_{t_1,t_2} := t_1p_1(1-p_1)^{t_1-1}(1-p_2)^{t_2} + t_2p_2(1-p_1)^{t_1}(1-p_2)^{t_2-1} \ge (t_1p_1 + t_2p_2)(1 - t_1p_1 - t_2p_2) := q_{t_1,t_2}$. Then, by the Chernoff bound (noting that the event that $v$ has an edge from $A_1 \cup A_2$ is independent for each $v$), 
    \[
        \Pr[|N_{Z_2}(A_1 \cup A_2)| \le q_{t_1,t_2} D_2 - r \sqrt{q_{t_1,t_2} D_2}] \le \exp(-r^2/2),
    \]
    which implies that 
    \[
        |N_{Z_2}(A_1 \cup A_2)| \ge q_{t_1,t_2} D_2 - \sqrt{4(t_1+t_2)q_{t_1,t_2} D_2 \log D}
    \] 
    except with probability at most $D^{-2(t_1+t_2)}$. This in particular implies  that 
    \begin{align*}
        |N_{Z_2}(A_1 \cup A_2)| 
        &\ge (t_1p_1 + t_2p_2)(1 - t_1p_1 - t_2p_2) D_2 - \sqrt{4(t_1+t_2)(t_1p_1 + t_2p_2) D_2 \log D} \\
        &= (1 - \eps_{t_1,t_2}) (t_1p_1 + t_2p_2) D_2, \numberthis \label{eqn:exp-bound}
    \end{align*}
    where $\eps_{t_1,t_2} = t_1p_1 + t_2p_2 + \sqrt{\frac{4(t_1+t_2)\log D}{(t_1p_1 + t_2p_2)D_2}}$, except with probability at most $D^{-2(t_1+t_2)}$. Then, by a union bound over all sets $A_1 \subseteq Y_1$ of size $t_1$ and $A_2 \subseteq Y_2$ of size $t_2$, we obtain that the probability that \Cref{eqn:exp-bound} holds simultaneously for \emph{all} $A_1 \subseteq Y_1$, $A_2 \subseteq Y_2$ of sizes $t_1,t_2$ respectively is at least $1-D_1^{t_1}D_2^{t_2} D^{-2(t_1+t_2)} \ge 1 - D^{t_1+t_2}$. Finally, we union bound over all $t_1, t_2 \ge 0$, not both $0$, to obtain that the the probability that \Cref{eqn:exp-bound} holds for \emph{all} sets $A_1, A_2$, not both empty, is at least $1 - \sum_{\substack{t_1,t_2 \ge 0 \\ (t_1, t_2) \not= (0,0)}} D^{t_1+t_2} = 1 - O(1/D)$.
\end{proof}

Now that we know that with high probability a random graph sampled from $\calE_{D_1,D_2,p_1,p_2}$ satisfies the property of expansion, we'd like to say that the same thing holds for random graphs sampled from $\calG_{D_1,D_2,d_1,d_1',d_2}$. To do this, we will use the following theorem stating that a random Erd\H{o}s-Renyi graph is with high probability a subset of a slightly larger regular graph. The point is that this allows us to argue that the slightly larger regular graph inherits the properties of the Erd\H{o}s-Renyi graph. For us, we will apply this theorem to the graphs on $(Y_1, Z_2)$, $(Z_1, Y_2)$, and $(Y_2,Z_2)$ to show that with high probability a random graph sampled from $G_{D_1,D_2,d_1,d_1',d_2}$ contains a random graph sampled from $\calE_{D_1,D_2,p_1,p_2}$. 

\begin{theorem}[\cite{HMMP24}, Theorem C.1] \label{thm:embedding}
    Fix $n_1, n_2, d_1, d_2$ such that $m = n_1 d_1 = n_2 d_2$. There is a universal constant $C$ such that if $\gamma \in (0, 1)$ satisfies $\gamma \ge C \left( \frac{d_1d_2}{m} + \frac{\log m}{\min \{ d_1, d_2 \}} \right)^{1/3}$, 
    then for $p = \frac{(1-\gamma)m}{n_1n_2}$, there is a joint distribution of $E \sim \calE_{n_1, n_2, p}$ and $G \sim \calG_{n_1, n_2, d_1, d_2}$ such that
    \[
        \Pr[E \subset G] = 1 - o(1). 
    \]
    Here, $\calE_{n_1,n_2,p}$ denotes the distribution of bipartite graphs on $(n_1, n_2)$ vertices, where each edge is sampled independently with probability $p$, and $\calG_{n_1,n_2,d_1,d_2}$ is the distribution of random $(d_1, d_2)$-biregular bipartite graphs on $(n_1,n_2)$ vertices.
\end{theorem}

Applying Theorem~\ref{thm:embedding} to the graphs $(Y_1, Z_2)$, $(Z_1, Y_2)$, and $(Y_2, Z_2)$, we obtain the following corollary:

\begin{corollary} \label{cor:embedding}
    For $D_1, D_2, d_1, d'_1, d_2$ such that $D_1 d_1 = d_2 d'_1$, let $m_1 = D_1d_1$ and $m_2 = D_2d_2$. There is a universal constant $C$ such that if $\gamma \in (0, 1)$ satisfies $\gamma \ge C\left( \frac{d_1d'_1}{m_1} + \frac{d_2^2}{m_2} + \frac{\log (m_1m_2)}{\min\{ d_1, d'_1, d_2 \}} \right)^{1/3}$, then for $p_1 = \frac{(1-\gamma)m_1}{D_1D_2}$ and $p_2 = \frac{(1-\gamma)m_2}{D_2^2}$, there is a joint distribution of $E \sim \calE_{D_1,D_2,p_1,p_2}$ and $G \sim \calG_{D_1,D_2,d_1,d'_1,d_2}$ such that
    \[
        \Pr[E \subset G] = 1 - o(1).
    \]
\end{corollary}

Now, let us finish off the proof of expansion for a random graph sampled from $\calG_{D_1,D_2,d_1,d'_1,d_2}$. Let $\gamma = o_D(1)$ satisfy the bounds in \Cref{cor:embedding}, and let $p_1 = \frac{(1-\gamma)d_1}{D_2}$ and $p_2 = \frac{(1-\gamma)d_2}{D_2}$. Sample $(E, G)$ from the joint distribution of $\calE_{D_1,D_2,p_1,p_2}$ and $\calG_{D_1,D_2,d_1,d'_1,d_2}$ as given by \Cref{cor:embedding}. With probability $1-o(1)$, $E \subset G$. Furthermore, by \Cref{lem:erdos-renyi-lossless}, with probability $1-O(1/D)$, $E$ satisfies the property that for all $A_1 \subseteq Y_1, A_2 \subseteq Y_2$, 
\begin{align*}
    |N_{Z_2}(A_1 \cup A_2)| 
    &\ge (1 - \eps_{|A_1|,|A_2|}) \cdot (p_1|A_1| + p_2|A_2|) D_2 \\
    &= (1 - \eps_{|A_1|,|A_2|})(1-\gamma) \cdot (d_1|A_1| + d_2|A_2|) \\
    &\ge (1 - \gamma - \eps_{|A_1|,|A_2|}) \cdot (d_1|A_1| + d_2|A_2|) \\
    &\ge \left( 1 - \gamma - \frac{d_1|A_1| + d_2|A_2|}{D_2} - \sqrt{\frac{4\log D}{(1-\gamma) \min \{ d_1, d_2 \}}} \right) \cdot (d_1|A_1| + d_2|A_2|) \\
    &=: (1 - \eps) \cdot (d_1|A_1| + d_2|A_2),
\end{align*}
thus because $E \subset G$ the same bound holds for $G$ as well. Further, with probability $1 - O(1/D)$, the same bound holds for subsets of $Z_1 \cup Z_2$ expanding to $Y_2$. For $d_1|A_1| + d_2|A_2| \le \mu |D_2|$ where $\mu = o_D(1)$, and $d_1,d_2 = \omega(\log D)$,  $\eps = o_D(1)$.

\section{Putting it All Together}\label{sec:all_together}

In this section, we briefly summarise how our results may be put together to prove our main results, Theorems~\ref{thm:main_randomised} and~\ref{thm:main_explicit}. The main aim here is to confirm that the various parameters may be chosen of the correct size and in the correct order to ensure each required condition is satisfied.

\begin{proof}[Proof of Theorem~\ref{thm:main_randomised}]
    We may choose $\delta = 1/2$ in Theorem~\ref{thm:random_lossless_Z_graphs} to obtain infinite families of $(n,m,\Delta_1, \Delta_2, \eta_1, \eta_2, \epsilon_1, \epsilon_2)$ lossless $Z$-graphs, for constants $\eta_1, \eta_2$, for any constant integers $\Delta_1, \Delta_2$, and with any constant imbalance $\frac{n}{m}$. Crucially, we also have $\epsilon_i < \frac{2}{\Delta_i}$ for $i = 1,2,$. In choosing our lossless $Z$-graphs, we first choose the constant imbalance $\frac{n}{m}$ to be large enough, then we choose the integer $\Delta_1$ to be large enough, and finally we choose the integer $\Delta_2$ to be large enough. 
    
    The rate of the quantum error-reduction code resulting from the lossless $Z$-graph, given the construction of Section~\ref{sec:construct_quantum_error_reduction_codes}, is $\frac{n}{n+2m} = \frac{n/m}{n/m+2}$, and so choosing $\frac{n}{m}$ to be a large enough constant enables our quantum error-reduction codes to have any rate in $(0,1)$. For the algorithms associated with our quantum error-reduction codes, we note that their encoding and unencoding may always take place with a linear number of quantum gates and in a constant quantum depth. Similarly, the classical error-reduction algorithms always use a linear total number of classical gates, and the parallel algorithms may be run in constant depth. As for the classical error-reduction algorithms, we may refer to Lemmas~\ref{lem:X_seq_reduction},~\ref{lem:Z_seq_reduction},~\ref{lem:parallel_X_reduction} and~\ref{lem:parallel_Z_reduction} to see that error reduction may be achieved by taking a large enough $\Delta_1$, and then a large enough $\Delta_2$. We emphasise that the parallel algorithms crucially rely on the ability of our randomised lossless $Z$-graphs to achieve $\epsilon_i < \frac{2}{\Delta_i}$, as well as being able to construct them for any constant integers $\Delta_1, \Delta_2$.

    To establish the properties of the final error-correcting codes $\mathcal{Q}_k$ constructed from the error-reduction codes, we first turn to Proposition~\ref{prop:ecc_code_rate} to establish the rate of $\mathcal{Q}_k$; suppose we have some target rate $R \in (0,1)$. For the base case of the induction, the quantum error-correcting code $\mathcal{Q}_0$, one may simply take, for example, a random quantum CSS code~\cite{calderbank1996good} of rate $R$. For the induction, we may start by choosing $r^{(1)}$ and $r^{(2)}$ large enough to achieve the positive rate $R$ of $\mathcal{Q}_k$ (see Equation~\eqref{eq:qerc_to_qecc_rate}); in particular, we must choose $r^{(1)} > 1/2$. Next, we turn to Lemma~\ref{lem:seq_erc_to_ecc} to handle the sequential error correction algorithms. One may choose $\Delta_1, \Delta_2$ sufficiently large in order to achieve a sufficiently small $\epsilon^{(2)}$ (the factor of error reduction used in Lemma~\ref{lem:seq_erc_to_ecc}); this factor is determined in Lemmas~\ref{lem:X_seq_reduction} and~\ref{lem:Z_seq_reduction}. The constant fraction of errors that the code may then correct, $\Delta$, is then determined by the constants $\alpha, \beta, \gamma$ given to us by Lemmas~\ref{lem:X_seq_reduction} and~\ref{lem:Z_seq_reduction}, as well as the numbers $R$, $r^{(1)}$ and $\delta_0$ (the fraction of errors correctable by the constant-size code $\mathcal{Q}_0$).

    Finally, for the parallel algorithms, we turn to Lemma~\ref{lem:parallel_erc_to_ecc}. Again, we may take $\Delta_1, \Delta_2$ to be large enough to achieve small enough $\epsilon^{(1)}$ and $\epsilon^{(2)}$, which are determined by Lemmas~\ref{lem:parallel_X_reduction} and~\ref{lem:parallel_Z_reduction}. The numbers $\delta^{(1)}$ and $\delta^{(2)}$ are again determined by the constants $\alpha, \beta, \gamma$ in those lemmas, and these $\delta^{(i)}$ determine the fraction of correctable errors $\Delta$ in Lemma~\ref{lem:parallel_erc_to_ecc}.
\end{proof}

\begin{proof}[Proof of Theorem~\ref{thm:main_explicit}]
    We again begin by choosing parameters for our family of lossless $Z$-graphs from which we form our explicit quantum error-reduction codes. Referring to Theorem~\ref{thm:Z-lossless}, we see that the imbalance $\frac{n}{m}$ may be chosen to be any constant controlled by the integers $\beta_1, \beta_2$; in particular, $\frac{n}{m} = \frac{\beta_1}{\beta_2}$. In addition, one may choose a constant imbalance of the two degrees $\Delta_1$ and $\Delta_2$; in particular, $\frac{\Delta_2}{\Delta_1} > 0.999\alpha$. In Theorem~\ref{thm:Z-lossless}, we therefore choose $\epsilon = \epsilon_1 = \epsilon_2 = \frac{1}{8}$. Next, we choose the constant imbalance $\frac{n}{m}$ to be large enough (by choosing integers $\beta_1, \beta_2$ such that $\frac{\beta_1}{\beta_2}$ is large enough), and then we choose the constant imbalance $\frac{\Delta_2}{\Delta_1}$ to be large enough (by choosing the constant $\alpha$ to be large enough). Finally, we choose the degree $\Delta_1$ to be large enough.

    We note that, unlike for our randomised construction, the explicit construction does not provide lossless $Z$-graphs for every $n$ large enough, and so some care must be taken to ensure that the lossless $Z$-graphs can fit with the concatenation of quantum error-reduction codes to quantum error-correcting codes. We may, however, proceed in essentially the same way as Spielman~\cite{spielman1995linear}. That is, we start by noting that, given any choice of $\epsilon, \alpha, \beta_1, \beta_2$ in Theorem~\ref{thm:Z-lossless}, the resulting family of lossless $Z$-graphs is dense, i.e., their sizes obey $n_q - n_{q-1} = o(n_q)$, which may be seen by considering the density of primes in arithmetic progressions.

    Next, we consider the families of quantum error-reduction codes in Section~\ref{sec:concat_structure} that we require. First, we require the family $\mathcal{R}_k^{(1)}$ with $n_k = n_0\left(\frac{r^{(1)}}{1-r^{(1)}}\right)^k$ message qubits and $n_{k-1}$ check qubits, for all $k \geq 1$. If we take the ratio $\frac{\beta_1}{\beta_2}$ to be large enough, by the density of the explicit lossless $Z$-graphs, there is some $n_0$ such that, for all $k \geq 1$, there is a lossless $Z$-graph of the desired parameters with $n_q \geq n_k$ and $2m_q \leq n_{k-1}$. To produce a quantum error-reduction code of the desired parameters with $n_k$ message qubits and $n_{k-1}$ check qubits, one may then simply remove some message qubits by setting them to $\ket{0}$ before the encoding circuit, and add in some ``dummy'' check qubits also in the state $\ket{0}$, that do not interact with the codeblock. The family $\mathcal{R}_k^{(2)}$ may be obtained similarly.

    With this, the proof of Theorem~\ref{thm:main_explicit} goes through in essentially the same way as that of Theorem~\ref{thm:main_randomised}. Indeed, choosing a large enough $\frac{n}{m}$ ensures that our quantum error-reduction codes may have any desired rate in $(0,1)$, which again translates into achieving any desired rate $R \in (0,1)$ for the final quantum error-correcting codes by Proposition~\ref{prop:ecc_code_rate} and Equation~\eqref{eq:qerc_to_qecc_rate}. Second, we may demonstrate the performance of the sequential classical decoding algorithms by referring to Lemmas~\ref{lem:X_seq_reduction} and~\ref{lem:Z_seq_reduction}, emphasising again that the order of operations is to first choose a large enough constant $\frac{n}{m}$, and then a large enough constant $\frac{\Delta_2}{\Delta_1}$, and then a large enough constant $\Delta_1$.
    
\end{proof}

\section*{Acknowledgements}

AW acknowledges inspiring discussions with L. Golowich. AW acknowledges funding from NSF grant PHY-2325080. AW acknowledges funding from the MIT-IBM Watson AI Lab. AW acknowldges support from the MIT Department of Physics. This preprint is assigned number MIT-CTP/5994.

TCL is supported in part by funds provided by the U.S. Department of Energy (D.O.E.) under cooperative research agreement DE-SC0009919, and by the Simons Collaboration on Ultra-Quantum Matter, which is a grant from the Simons Foundation (652264, JM).

RYZ is supported by a Schwarzman College of Computing Future Research Fellowship with support from Google.


\bibliographystyle{alpha}
\bibliography{refs.bib}

\end{document}